\documentclass[runningheads,envcountsame,orivec]{llncs}

\usepackage{amsmath,amssymb}
\usepackage{ebproof}
\usepackage{enumitem}
\usepackage[T1]{fontenc}
\usepackage{graphicx}
\usepackage{ifthen}
\usepackage[mathlines]{lineno}
\usepackage{mathtools}
\usepackage{stmaryrd}
\usepackage{tikz}
\usepackage{thm-restate}
\usepackage{xparse}

\usepackage{hyperref}
\usepackage[capitalize]{cleveref}

\usetikzlibrary{calc,decorations.pathmorphing,shapes}

\input{macros.tex}

\begin{document}
\title{Denotational reasoning for asynchronous multiparty session types}
\author{Dylan McDermott \and Nobuko Yoshida}
\institute{University of Oxford\\
\email{dylan@dylanm.org, nobuko.yoshida@cs.ox.ac.uk}}
\authorrunning{D.\ McDermott and N.\ Yoshida}
\maketitle

\begin{abstract}
We provide the first denotational semantics for asynchronous
multiparty session types with precise asynchronous subtyping. 
Our semantics enables us to reason about
asynchronous
message-passing, in which message-sending is non-blocking.
It enables us to prove the correctness of communication optimisations, in particular, those involving reordering of messages.
Our development crucially relies on modelling message-passing as a computational effect.
We apply grading, a paradigm for tracking computational effects, to
asynchronous message-passing, demonstrating that multiparty session
typing can be viewed as an instance of grading. 
We demonstrate the utility of our model by showing that it
forms an adequate denotational semantics for a call-by-value
asynchronous message-passing 
calculus, that ensures communication safety, deadlock-freedom
and liveness in the presence of communication optimisations.
\end{abstract}

\section{Introduction}
\label{sec:introduction}
\emph{Multiparty session types} (MPST) provide a typing discipline
for ensuring that multiple participants communicating via
message-passing conform to a \emph{multiparty protocol}, satisfying 
desired safety and liveness properties such as (1) \textbf{communication safety} 
(aka~\emph{type-safety}) 
(no participant will receive a message with a value of an unexpected type or unexpected label); (2) \textbf{deadlock-freedom} 
(the participants will never get stuck because they are waiting for each other to send a message);
(3) \textbf{liveness} (aka~\emph{progress}~\cite{CairesP10,CastellaniDG24,CastellaniDG23b}) 
(if one participant wishes to communicate with another, then that communication will \emph{eventually} happen following the protocol).
Deadlock-freedom implies that two participants $\pp$ and $\pq$ cannot both be blocked waiting for messages from each other (wait-free), while
liveness requires that, if either of $\pp$ and $\pq$ wants to communicate with the other, then that communication will eventually happen.
The theory of MPST can guarantee these properties for multiple 
participants either by taking the \emph{top-down 
approach} \cite{HYC2016,VeryGentle,SynchronousSubtyping}, which uses a
\emph{global type} to specify a global protocol, or by taking the
\emph{bottom-up approach} \cite{Scalas2019,GPPSY2023}, which 
model-checks a set of local types to ensure a desired property.
A vast number of session type theories have been proposed, 
but one simple question has been left open:
while several semantic studies have been 
made for binary (2-party) session types 
based on, e.g., logical 
relations \cite{balzer2023logicalrelationssessiontypedconcurrency}
and game semantics \cite{CY2019}, 
no extensional denotational semantics for an expressive MPST calculus
exists (see \cref{sec:related-work}). 
Our challenge is to find a simple but effective denotational semantics for MPST,
one that is useful for proving properties of message-passing programs.

We address this challenge by constructing a denotational semantics for MPST that is \emph{extensional} in the sense that it hides non-observable behaviours of programs (e.g.\ internal reductions), but still captures observable behaviours (e.g.\ sending a message to another participant).
This enables us to reason about systems, without non-observable details getting in the way.
For programming languages without message-passing, extensional denotational semantics have been successfully used for program reasoning, for instance, proofs of equivalences between syntactically distinct programs (e.g.\ \cite{kammar2012algebraic}).
We do the same, but for MPST; we can use our semantics to prove the validity of simple optimizations of MPST systems (e.g.\ \cref{ex:global-state-example-bisim} below).

Specifically, we introduce a simple notion of \textbf{computation tree}, and show that it provides a model of asynchronous message-passing.
We introduce \lambdasmp{}, an idealised call-by-value programming language with a type system based on session types.
\lambdasmp{} features \emph{asynchronous} message-passing, in which messages are buffered into queues so that sending a message does not block.
We show that our computation trees form a denotational semantics for \lambdasmp{}, and this semantics is \textbf{adequate} with respect to a typed notion of \emph{bisimilarity} for \lambdasmp{}.
This notion of bisimilarity accounts for asynchrony, and thus our model can reason about program equivalences that rely on asynchrony.

We design \lambdasmp{}, and its interpretation using computation trees, by
following the key insight that sending and receiving of messages are
\emph{computational effects}, analogous to mutating state or raising an
exception.
Recent work~\cite{katsumata2014parametric,orchard2019quantitative} has
established \textbf{grading} as the key technique of
tracking computational effects compositionally, and thus we design \lambdasmp{} as a graded
type system (aka an \emph{effect system}).
A graded type system assigns both a type and a \emph{grade} to each
computation; here the grades are multiparty session types, which provide enough information to enforce our desired safety and liveness properties.
Our type system is not based on linear logic, unlike previous session type systems (see \cref{sec:related-work}).
Integrating session types into a call-by-value 
$\lambda$-calculus elucidates the connection between effects 
and session types, and enables us to apply techniques from the
computational effects literature to session types.
This insight is crucial to the design of our semantics:
the standard way of modelling a graded type systems is to use a
\emph{graded monad}~\cite{borceux2005internal,smirnov2008graded,mellies2012parametric,katsumata2014parametric},
so we show that computation trees form a graded monad that can be used to model \lambdasmp{}.

Asynchronous message-passing is more widely adopted in distributed systems than
\emph{synchronous} message-passing, where a computation that sends a message has to wait until that message is received before making any further progress.
We give an operational semantics for \lambdasmp{}, accounting for asynchrony in the usual way, namely by buffering messages into a notion of \emph{queue}.
For our denotational semantics, we can use a simpler setup that 
does not involve queues, but is still asynchronous.
We also account for asynchrony in our types, using 
the sound and complete (\textbf{precise}) \textbf{asynchronous multiparty 
session subtyping} of Ghilezan et al.\ \cite{GPPSY2023}.
Asynchrony enables a more permissive subtyping, and hence a more permissive type system, than synchronous message-passing, because some deadlocks that occur with synchronous semantics do not occur with an asynchronous semantics. The precise asynchronous subtyping 
enables practically useful \emph{communication optimisations}.   
Several sound algorithms have been
developed 
\cite{DBLP:journals/pacmpl/Castro-PerezY20,DBLP:conf/ppopp/CutnerYV22,BocchiK0025}
and implemented in   
programming languages such as 
Rust, MPI and C.
Asynchronous subtyping has also been mechanised in Rocq~\cite{EY2024}.

In this paper, instead of taking the definition of subtyping from \cite{GPPSY2023}, we reformulate the definition from the ground up.
Our definition of subtyping improves on \cite{GPPSY2023} in that it only involves session types (not their infinite unfoldings as session \emph{trees}), and also enables subtyping for session types with free type variables (which are used in \lambdasmp).
Nevertheless, we show for closed session types that our definition is equivalent to \cite{GPPSY2023}.

\paragraph{\textbf{Contributions.}}
This paper provides an effective denotational semantics for MPST,
based on the observation that message-passing is a computational effect.
\textbf{\cref{sec:basic-calculus}} introduces our calculus \lambdasmp{}, the subsequent sections provide the main contributions of the paper.
\textbf{\cref{sec:session-types}} introduces our new formulation of asynchronous session subtyping. This is the first sound and complete definition of asynchronous session subtyping that permits subtyping in the presence of type variables, and the first that does not rely on session trees.
\textbf{\cref{sec:graded-calculus}} introduces our session type system for \lambdasmp{}. This is the first session type system for a call-by-value language based on grading instead of linearity, and demonstrates the connection between session types and computational effects.
\textbf{\cref{sec:computation-trees}} introduces \emph{computation trees} as a simple representation of asynchronous message-passing computation. It demonstrates how to account for asynchrony without using queues, and defines a graded monad for asynchronous message-passing.
\textbf{\cref{sec:semantics}} shows that computation trees provide a model of asynchronous message-passing, in the form of an \textbf{adequate} denotational semantics for \lambdasmp{}.
\textbf{\cref{sec:session-calculus}} demonstrates the utility of our model, by using the model to prove safety and liveness properties of \lambdasmp{}.

\section{A call-by-value message-passing calculus}\label{sec:basic-calculus}

We begin by introducing the syntax and operational semantics of our message-passing calculus \lambdasmp{}, and giving a few examples.
The purpose of \lambdasmp{} is to demonstrate our denotational semantics, and as such we only include enough features to have non-trivial message-passing programs.
As such we do not focus on the \emph{practicality} of \lambdasmp{}, neither for writing programs in \lambdasmp{}, nor for implementing \lambdasmp{} itself.\footnotemark{}
\footnotetext{%
  In particular, we do not consider decidability.
  Indeed, asynchronous session subtyping is known to be undecidable~\cite{bravetti2017undecidability}, so an implementation would be of a sound and decidable approximation of subtyping, as discussed in \cite{BocchiK0025}.
  An implementation would also have \emph{grade polymorphism}, as implemented for instance in Granule~\cite{orchard2019quantitative}.
}
Since we want to emphasize the connection with computational effects, we base \lambdasmp{} on an existing effectful calculus, namely the \emph{fine-grain call-by-value} calculus~\cite{levy2003modelling}.
We take a simplified\footnotemark{} fine-grain call-by-value, and add asynchronous message-passing as an effect, along with first-order guarded recursive functions.
\footnotetext{%
  We designed our calculus so that deadlock-freedom and liveness are non-trivial, while the calculus is otherwise as simple as possible.
  In particular, compared to fine-grain call-by-value, we do not have higher-order functions.
  This is not because they would be difficult to handle, indeed an advantage of the graded approach is that it easily takes care of higher-order functions, at the cost of a little extra complexity.
  Compared to some works on session types, we also do not have session \emph{delegation}.
  For our session types we targeted the same level of expressiveness as Ghilezan et al.\ \cite{GPPSY2023}, who also do not have delegation (precise asynchronous subtyping with session delegation is an open problem).
}

First, we start with some basic terminology.
A \emph{message} $m = \msg{\ell}{v}$ is a pair of a \emph{label} $\ell$ and a \emph{payload} $v$.
We assume throughout the paper that there is some fixed set $\Label$ of labels; $\ell$ ranges over elements of $\Label$.
A payload value is a constant value, for concreteness, we consider three \emph{ground types} $\basetype$, namely $\unittype$, $\booltype$, $\inttype$, and require a message payload to be a constant $v$ of one of those types, i.e.\,either $\star$ (of type $\unittype$), an integer $n$ (of type $\inttype$) or one of the booleans $\trueterm$ and $\falseterm$ (of type $\booltype$).
We write $v : \basetype$ for the typing relation between constants and ground types.
We also assume there is a fixed set of \emph{participants} $\pp, \pq, \pr, \dots$.
The idea is that the participants perform tasks in parallel and communicate by exchanging messages between each other.
They could, for instance, be implemented as separate programs running on separate servers, with messages being sent across a network.

\emph{Terms} of $\lambdasmp$ are stratified into \emph{values} $v, w$ and \emph{computations} $t, u$. 
They are generated inductively by the following grammar, subject to the constraints discussed below.
\[\begin{array}{r@{}l}
  v, w
  ~\Coloneq~&
  x~|~\unitterm~|~n~|~\trueterm~|~\falseterm
  \\[1ex]
  t, u
  ~\Coloneq~&
  \returnterm{v}~|~\bindterm{x}{t}{u}
  ~|~\addterm{v}{w}~|~\lessterm{v}{w}~|~\ifterm{v}{t_1}{t_2}
  \\[1ex]
  |~&\sendterm{\ell}{v}{\pp}{t}~|~\recvterm{\pp}{\ell_i}{x_i : \basetype_i}{t_i}{i \in I}
  \\[1ex]
  |~&\letrecterm{f}{(x_1, \dots, x_n)}{t}{u} ~|~f (v_1, \dots, v_n)
\end{array}\]
A value $v$ is either a variable $x$, or one of the constants
$\star$, $n$, $\trueterm$ or $\falseterm$.

The first part of the grammar of computations is the core of fine-grain
call-by-value: a computation can \emph{return} a value $v$, or it can be
a sequence $\bindterm{x}{t}{u}$ of computations.
The latter first evaluates $t$ and then, if $t$
returns a result, evaluates $u$ with $x$ bound to the result of $t$. 
Computations also 
include integer addition $\addterm{v}{w}$, integer
comparison $v < w$, and conditionals $\ifterm{v}{t_1}{t_2}$.
These operate on values; thus to e.g.\ compare two integer-returning computations
one would first use $\syntaxkeyword{let}$ to first evaluate the operands.

The computation $\sendterm{\ell}{v}{\pp}{t}$ sends the message $(\ell, v)$
to participant $\pp$, and then continues as the computation $t$.
The value $v$ could for instance be the result of a computation evaluated using $\syntaxkeyword{let}$.
Since we use an \emph{asynchronous} semantics, sending a message does not block; the computation does not have to wait for the message to be received before continuing as $t$.
The computation
$
  \longrecvterm{\pp}{\mltype{\ell_1}{x_1 : \basetype_1}.\,t_1, \dots, \mltype{\ell_n}{x_n : \basetype_n}.\,t_n}
$
receives one message from the participant $\pp$; if that message has label $\ell_i$, 
it continues as $t_i$, with $x_i$ bound to the message
payload.
If the message label is not one of the $\ell_i$, or the payload does not have the corresponding type $\basetype_i$, then this is a \emph{communication error}; in our operational semantics below, reduction gets stuck.
The message labels $\ell_i$ are drawn from our fixed set $\Label$.
We require the labels $\ell_i$ to be distinct from each other,
and we require $n > 0$.
We typically write
$
  \recvterm{\pp}{\ell_i}{x_i : \basetype_i}{t_i}{i \in I}
$
where $I$ is a finite non-empty set, but we do not consider $I$ to be
part of the syntax.
This computation blocks until the corresponding message is received.

Computations also include recursive function definitions, which are written as $\letrecterm{f}{(x_1, \dots, x_n)}{t}{u}$, and applications $f(v_1, \dots, v_n)$ of these functions to arguments.
Function names $f$ are distinct from variables $x$, and they are not values.
We require every recursive definition to be guarded, in the sense that every occurrence of $f$ in $t$ appears inside a $\sendtermkeyword$ or a $\recvtermkeyword$.
This means that to do a recursive call, one first has to send or receive a message; this ensures that a computation cannot simply diverge, it eventually has to \emph{do} something.

A value or computation is \emph{closed} when it has no free variables $x$ and no free function names $f$ (though we permit it to have free participants).

\begin{figure*}[t]
\begin{gather*}
  \intertext{\fbox{Computation reduction $\reduceaction{t}{u}{\alpha}$}}
  \qquad
  \begin{array}{r@{}l}
  \text{Reduction contexts}
  ~
  \rctx{R}[\hole]
  \Coloneq
  \hole
  ~|&~
  \bindterm{x}{\rctx{R}[\hole]}{u}
  \\~|&~
  \letrecterm{f}{(x_1, \dots, x_n)}{t}{\rctx{R}[\hole]}
  \end{array}
\\[1ex]
\RuleName{Cong}
 \frac
   {\reduceaction{t}{u}{\alpha}}
   {\reduceaction{\rctx{R}[t]}{\rctx{R}[u]}{\alpha}}
   \quad
  \RuleName{LetR}\frac
    {}
    {\reduceaction{\bindterm{x}{\returnterm{v}}{u}}{\subst{u}{x \mapsto v}}{\tauaction}}
  \\
  \RuleName{IfT}
    \frac
      {\phantom{|}}
      {\reduceaction{\ifterm{\trueterm}{t_1}{t_2}}{t_1}{\tauaction}}
  \quad
  \RuleName{IfF}
    \frac
      {}
      {\reduceaction{\ifterm{\falseterm}{t_1}{t_2}}{t_2}{\tauaction}}
  \\[0.1ex]
  \RuleName{Sub}
    \frac
      {m = n_1 - n_2}
      {\reduceaction{(\subtractterm{n_1\!}{\!n_2})}{\returnterm{m}}{\tauaction}}
  ~
  \RuleName{LeT}
    \frac
      {n_1 < n_2}
      {\reduceaction{(\lessterm{n_1\!}{\!n_2})}{\returnterm{\trueterm}}{\tauaction}}
  ~
  \RuleName{LeF}
    \frac
      {n_1 \geq n_2}
      {\reduceaction{(\lessterm{n_1\!}{\!n_2})}{\returnterm{\falseterm}}{\tauaction}}
\\[0.3ex]
\RuleName{Send}
  \frac
    {}
    {\reduceaction{\sendterm{\ell}{v}{\pp}{t}}{t}{\sendaction{\pp}{\msg{\ell}{v}}}}
  \quad
\RuleName{Recv}
  \frac
    {v : \basetype_j}
    {\reduceaction{\recvterm{\pp}{\ell_i}{x_i : \basetype_i}{t_i}{i \in I}}{\subst{t_j}{x_j \mspace{2mu}{\mapsto}\mspace{2mu} v}}{\mspace{-7mu}\recvaction{\pp}{\msg{\ell_j}{v}}\mspace{-7mu}}}
\\[0.1ex]
\RuleName{Apply}
 \frac
   {}
   {\begin{array}{l}\reduceaction{\letrecterm{f}{(x_1, \dots, x_n)}{t}{\rctx{R}[f(v_1, \dots, v_n)]}\\\mspace{100mu}}{\letrecterm{f}{(x_1, \dots, x_n)}{t}{\rctx{R}[\subst{t}{x_1 \mapsto v_1, \dots, x_n \mapsto v_n}]}}{\tauaction}\end{array}}
\\[2ex]
  \intertext{\fbox{(Asynchronous) configuration reduction $\reduceaction{\config{\queue{\rho}}{t}{\queue{\sigma}}}{\config{\queue{\rho}}{u}{\queue{\sigma}}}{\alpha}$}}
  \begin{array}{@{}l@{}}
    \RuleName{CInt}
    \\[-0.5ex]
  \begin{prooftree}
    \hypo{\reduceaction{t}{u}{\tauaction}}
    \infer1{\reduceaction
      {\config{\queue{\rho}}{t}{\queue{\sigma}}}
      {\config{\queue{\rho}}{u}{\queue{\sigma}}}
      {\tauaction}}
  \end{prooftree}
\end{array}
  \quad
  \begin{array}{@{}l@{}}
    \RuleName{CProd}
    \\[-0.5ex]
  \begin{prooftree}
    \hypo{\reduceaction{t}{u}{\sendaction{\pp}{m}}}
    \infer1{\reduceaction
      {\config{\queue{\rho}}{t}{\queue{\sigma}}}
      {\config{\queue{\rho}}{u}{\consfront{\pp}{m}{\queue{\sigma}}}}
      {\tauaction}}
  \end{prooftree}
\end{array}
  \quad
  \begin{array}{@{}l@{}}
    \RuleName{CCons}
    \\[-0.5ex]
  \begin{prooftree}
    \hypo{\reduceaction{t}{u}{\recvaction{\pp}{m}}}
    \infer1{\reduceaction
      {\config{\consback{\queue{\rho}}{\pp}{m}}{t}{\queue{\sigma}}}
      {\config{\queue{\rho}}{u}{\queue{\sigma}}}
      {\tauaction}}
  \end{prooftree}
\end{array}
  \\
  \begin{array}{@{}l@{}}
    \RuleName{CSend}
    \\[0.5ex]
  \begin{prooftree}
    \infer0{\reduceaction
      {\config{\queue{\rho}}{t}{\consback{\queue{\sigma}}{\pp}{m}}}
      {\config{\queue{\rho}}{t}{\queue{\sigma}}}
      {\sendaction{\pp}{m}}}
  \end{prooftree}
\end{array}
  \qquad
  \begin{array}{@{}l@{}}
    \RuleName{CRecv}
    \\[0.5ex]
  \begin{prooftree}
    \infer0{\reduceaction
      {\config{\queue{\rho}}{t}{\queue{\sigma}}}
      {\config{\consfront{\pp}{m}{\queue{\rho}}}{t}{\queue{\sigma}}}
      {\recvaction{\pp}{m}}}
  \end{prooftree}
\end{array}
\end{gather*}
\caption{Operational semantics of \lambdasmp{}\label{fig:operationalsemantics}.}
\end{figure*}

We define an operational semantics\footnotemark{} for \lambdasmp{}, in the form of a transition system labelled by \emph{local actions} $\alpha$.
\footnotetext{We extend this to an operational semantics for a notion of \emph{session} in \cref{sec:session-calculus} below.}
\[
\begin{array}{r@{~}r@{~}l@{\quad}l}
  \alpha &\Coloneqq& \tauaction &\text{(internal action)}
  \\
  &|& \sendaction{\pp}{m} &\text{(send message $m$ to participant $\pp$)}
  \\
  &|& \recvaction{\pp}{m} &\text{(receive message $m$ from participant $\pp$)}
\end{array}
\]
We define our operational semantics in two steps.
First, we define a synchronous notion of reduction $\reduceaction{t}{u}{\alpha}$ for closed computations, one that does not incorporate any form of reordering of messages from different participants.
This is generated by the rules in the first part of \cref{fig:operationalsemantics}.
We use a notion of reduction context $\rctx{R}[\hole]$ for congruence rules; such a context is a computation with a single hole $\hole$, indicating the position of the computation to reduce. We write $\rctx{R}[t]$ for the replacement of the hole with a computation $t$.

The second step is to use queues to model \emph{asynchronous} message-passing.
A \emph{queue} $\queue{\sigma}$ is a function that assigns, to every participant $\pp$, a finite ordered list of messages, such that this list is empty for all but finitely many $\pp$ (so each queue only contains a finite number of messages).
We write $\queueempty$ for the empty queue (which maps every participant to the empty list), write $\consfront{\pp}{m}{\queue{\sigma}}$ for adding a message $m$ to the front of $\pp$'s list, and write $\consback{\queue{\sigma}}{\pp}{m}$ for adding $m$ to the back of $\pp$'s list.
Thus, if $m$ and $m'$ have distinct labels or payloads, the queues $\consfront{\pp}{m}{\consfront{\pq}{m'}{\queue{\sigma}}}$
and $\consfront{\pq}{m'}{\consfront{\pp}{m}{\queue{\sigma}}}$ are equal exactly when $\pp \neq \pq$; the ordering of messages between different participants does not matter.
A \lambdasmp{} \emph{configuration} $\mathcal{C} = \config{\queue{\rho}}{t}{\queue{\sigma}}$ consists of a (receive) queue $\queue{\rho}$, closed computation $t$, and a (send) queue $\queue{\sigma}$.
The queue $\queue{\rho}$ contains messages yet to be consumed by $t$, while the queue $\queue{\sigma}$ contains messages that have been produced by $t$.

Our operational semantics is defined at the bottom of \cref{fig:operationalsemantics}; it consists of a reduction relation $\reduceaction{\config{\queue{\rho}}{t}{\queue{\sigma}}}{\config{\queue{\rho}'}{u}{\queue{\sigma}'}}{\alpha}$ for configurations.
This essentially reduces $t$ synchronously as above, except that all messages are buffered into one of the queues.
A configuration of the form $\config{\queueempty}{\mathcal{R}[\returnterm{v}]}{\queueempty}$, where $\mathcal{R}$ has no non-recursive $\syntaxkeyword{let}$s, does not reduce; we consider this to be a completed execution, with result $v$.
We define a partial function $\Result$ from configurations to values, with $\Result\config{\queueempty}{\mathcal{R}[\returnterm{v}]}{\queueempty} = v$ and $\Result{}$ undefined otherwise.

\begin{example}[modified from an example of \cite{GPPSY2023}]\label{ex:gppsy-example}
  The following computation $t$ receives the outcome of a computation from a participant $\pp$, then tells $\pq$ whether to continue with some other computation, or stop.
  The result of $t$ is a boolean, indicating whether it sent a $\ml{stop}$ message to $\pq$.
  \begin{gather*}
    \begin{array}{r@{}r@{~}l}\\
      t = \longrecvterm{\pp}{&
        \mltype{\mathtt{success}}{x : \inttype}.&
        \bindterm{y}{\lessterm{0}{x}}{\\&&
        \begin{array}{@{}r@{}l}
        \ifterm{y}
        {&\sendterm{\mathtt{cont}}{x}{\pq}{\returnterm{\falseterm}}\\}
        {&\sendterm{\mathtt{stop}}{\trueterm}{\pq}{\returnterm{\trueterm}}}
        \end{array}}
        \\
        &\mltype{\mathtt{error}}{x : \booltype}.&
        \sendterm{\mathtt{stop}}{\falseterm}{\pq}{\returnterm{\trueterm}}
      }
    \end{array}
  \end{gather*}
  For instance, if $\pp$'s computation fails, then $t$ will send $\msg{\mathtt{stop}}{\falseterm}$ to $\pq$.
  We have computation reductions, and hence configuration reductions, as follows (omitting the intermediate configurations).
  \begin{gather*}
      \reduceaction
      {t}
      {
        \reduceaction
        {\sendterm{\mathtt{stop}}{\falseterm}{\pq}{\returnterm{\trueterm}}}
        {\returnterm{\trueterm}}
        {\!\!\sendaction{\pq}{\msg{\mathtt{stop}}{\falseterm}}\!\!}
      }
    {\!\!\recvaction{\pp}{\msg{\mathtt{error}}{\trueterm}}\!\!}
    \\
    \reduceaction
    {\config{\queueempty}{t}{\queueempty}}
    {
      \reduceaction
      {}
      {
        \reduceaction
        {}
        {
          \reduceaction
          {}
          {~\config{\queueempty}{\returnterm{\trueterm}}{\queueempty}}
          {\sendaction{\pq}{\msg{\mathtt{stop}}{\falseterm}}}
        }
        {\tauaction}
      }
      {\tauaction}
    }
    {\recvaction{\pp}{\msg{\mathtt{error}}{\trueterm}}}
  \end{gather*}
  The above reductions do not require asynchrony.
  To demonstrate asynchronous message-passing, consider a computation $u$ that decides which message to send to $\pq$, without first waiting for the message from $\pp$.
  \begin{gather*}
    \begin{array}{@{}r@{}l}
      u = \ifterm{y}
      {&\sendterm{\mathtt{cont}}{0}{\pq}{u'(\falseterm)}\\}
      {&\sendterm{\mathtt{stop}}{\falseterm}{\pq}{u'(\trueterm)}}
    \end{array}
    \\~\text{where}~u'(v) = \recvterm{\pp}{\ell}{x : \basetype}{\returnterm{v}}{\mltype{\ell}{\basetype} \in \{\mltype{\mathtt{success}}{\inttype}, \mltype{\mathtt{error}}{\booltype}\}}
  \end{gather*}
  Since we buffer messages into queues, we have the following reductions, in which a message is received from $\pp$ first -- just like in the reduction of $\config{\queueempty}{t}{\queueempty}$ above.
  Here 
  $\queue{\rho} = \consfront{\pp}{\msg{\mathtt{error}}{\trueterm}}{\queueempty}$
  and $\queue{\sigma} = \consfront{\pq}{\msg{\mathtt{stop}}{\falseterm}}{\queueempty}$.
  \begin{align*}
    \reduceaction
    {\config{\queueempty}{\bindterm{y}{\returnterm{\falseterm}}{u}}{\queueempty}}
    {&~
      \reduceaction
      {\config{\queue{\rho}}{\bindterm{y}{\returnterm{\falseterm}}{u}}{\queueempty}\\}
      {
        \reduceaction
        {}
        {
        \reduceaction
        {}
        {&~
          \reduceaction
          {\config{\queue{\rho}}{u'(\trueterm)}{\queue{\sigma}}\\}
          {
            \reduceaction
            {}
            {&~\config{\queueempty}{\returnterm{\trueterm}}{\queueempty}}
            {\!\!\scalebox{0.7}{$\sendaction{\pq}{\msg{\mathtt{stop}}{\falseterm}}$}}
          }
          {\tauaction}
        }
        {\tauaction}
        }
        {\tauaction}
      }
      {\tauaction}
    }
    {\recvaction{\pp}{\msg{\mathtt{error}}{\trueterm}}}
  \end{align*}
\end{example}

\begin{example}[Global state]
  \label{global-state-computations}
  Our second example shows that we can emulate other computational
  effects via message-passing, in a manner reminiscent of 
  \emph{effect handlers}~\cite{plotkin2009handlers}.
  Specifically, we focus on mutable state.
  For concreteness, we assume the state consists of a single integer.

  Consider two participants $\role{s}$ (for state) and $\role{c}$ (for client).
  Tracking the state is the role of $\role{s}$, which waits for a message from $\role{c}$: a $\ml{get}$ message indicates it should send the value of the state back to $\role{c}$, as the payload of a $\ml{st}$ message; a $\ml{put}$ message indicates it should update the current value of the state; and a $\ml{done}$ message indicates that the interaction is finished, the state is no longer needed.
  A possible implementation of $\role{s}$ would be the following computation $t_{\role{s}}$, where $x$ is the current value of the state (initially $0$).
  \[
    \begin{array}{l}
    t_{\role{s}} =
    \letrecterm{f}{x}{
      \longrecvterm{\role{c}}{
      \\\mspace{60mu}
              \mltype{\ml{get}}{z : \unittype}.\,\sendterm{\ml{st}}{x}{\role{c}}{\applyterm{f}{x}},
      \\\mspace{60mu}
              \mltype{\ml{put}}{y : \inttype}.\,\applyterm{f}{y},
      \\\mspace{60mu}
              \mltype{\ml{done}}{z : \unittype}.\,\returnterm{\unitterm}
      \\\mspace{40mu}
      }
    }{\applyterm{f}{0}}
    \end{array}
    \mspace{-14mu}
    \begin{array}{r@{}l}
      t_{\role{c}, 1} =~
      &\sendterm{\ml{get}}{\unitterm}{\role{s}}{\\
      &\longrecvterm{\role{s}}{\mltype{\ml{st}}{x : \inttype}.\\
      &\sendterm{\ml{put}}{0}{\role{s}}{\returnterm{x}}}}
      \\[1ex]
      t_{\role{c}, 2} =~
      &\sendterm{\ml{get}}{\unitterm}{\role{s}}{\\
      &\sendterm{\ml{put}}{0}{\role{s}}{\\
      &\longrecvterm{\role{s}}{\mltype{\ml{st}}{x : \inttype}.\returnterm{x}}}}
    \end{array}
  \]
  On the right above, we also give two partial implementations of the client; each gets the value of the state, and then sets the state to $0$.
  One waits for the $\ml{st}$ message before sending $\ml{put}$; the other eagerly sends the $\ml{put}$ message, which is valid because of asynchrony.
  The latter can be thought of an \emph{optimisation} of the former.
  Neither sends a $\ml{done}$ message; one example of a complete implementation of $\role{c}$ is $\bindterm{x}{t_{\role{c}, 1}}{\sendterm{\ml{done}}{\unitterm}{\role{s}}{\returnterm{x}}}$.
\end{example}

\section{Asynchronous multiparty session types}\label{sec:session-types}
Our type system for \lambdasmp{} is based around (multiparty) session types.
A session type $\sty{T}$ describes a \emph{protocol} that must be followed by a given participant in a distributed system.
They are generated as follows, where $X$ ranges over \emph{session type variables}.
\[
  \sty{T}, \sty{U} \Coloneqq
  \ltypeend
  ~|~
  \ltypesend{\pp}{\ell_i}{\basetype_i}{\sty{T}_i}{i \in I}
  ~|~
  \ltyperecv{\pp}{\ell_i}{\basetype_i}{\sty{T}_i}{i \in I}
  ~|~
  \ltypevar{X}
  ~|~
  \ltyperec{\ltypevar{X}}{\sty{T}}
\]
$\ltypeend$ denotes the end of the protocol, it requires that there are no further interactions with the other participants.
An \emph{internal choice}
$\ltypesend{\pp}{\ell_i}{\basetype_i}{\sty{T}_i}{i \in I}$
means send some message $\msg{\ell_i}{v}$, where $v : \basetype_i$, to participant $\pp$, and continue according to $\sty{T}_i$.
Sending a message $\pp$ with a label not in $\{\ell_i \mid i \in I\}$, or with the wrong payload type, is not permitted.
We require the labels $\ell_i$ to be distinct from each other, and that $I$ is non-empty.
An \emph{external choice} $\ltyperecv{\pp}{\ell_i}{\basetype_i}{\sty{T}_i}{i \in I}$ means receive a message $\msg{\ell_i}{v}$ from participant $\pp$, and then continue according to $\sty{T}_i$. The participant $\pp$ chooses the message, but is required to ensure that the label is one of the labels $\ell_i$, and that the payload has the corresponding type $\basetype_i$.
Similar syntactic constraints apply, in particular distinctness of message labels.
Finally, $\ltyperec{\ltypevar{X}}{\sty{T}}$ is a recursive protocol, binding the type variable $\ltypevar{X}$; the session type $\ltyperec{\ltypevar{X}}{\sty{T}}$ is equivalent to $\subst{\sty{T}}{\ltypevar{X} \mapsto \ltyperec{\ltypevar{X}}{\sty{T}}}$.
Recursion is required to be \emph{guarded}, in the sense that every occurrence of $\ltypevar{X}$ in $\sty{T}$ is inside an internal or external choice.

The (single-step) \emph{unfolding} $\unfold{\sty{T}}$ of a session type $\sty{T}$ is defined as follows.
The unfolding is an equivalent session type to $\sty{T}$, but due to guardedness,
$\unfold{\sty{T}}$ is always either $\ltypeend$, or an internal or external choice, or a type variable.
\[
  \unfold{\ltyperec{\ltypevar{X}}{\sty{T}}} = \subst{\unfold{\sty{T}}}{\ltypevar{X} \mapsto \ltyperec{\ltypevar{X}}{\sty{T}}}
  \qquad
  \unfold{\sty{T}} = \sty{T}~~\text{if}~\sty{T}~\text{is not a recursive type}
\]
If $\Theta$ is a set containing all of the free type variables of a session type $\sty{T}$, we will refer to $\sty{T}$ as a session type \emph{over $\Theta$}.
A session type is \emph{closed} when it has no free type variables.

\begin{example}\label{ex:gppsy-example-types}
  Recall the computations $t$ and $u$ from \cref{ex:gppsy-example}, where $t$ first receives a message from $\pp$, but $u$ first sends a message to $\pq$.
  The computations $t$ and $u$ follow the protocols described by $\sty{T}$ and $\sty{U}$ respectively,
  and our type system in \cref{sec:graded-calculus} assigns these session types to the computations.
  \begin{gather*}
    \sty{T} =
    \ltyperecvcases{\pp}{
      \mltype{\ml{success}}{\inttype}.\sty{T}'
      \\
      \mltype{\ml{error}}{\booltype}.\sty{T}'
    }
    ~\text{where}~
    \sty{T'} = \ltypesendcases{\pq}{\mltype{\ml{cont}}{\inttype}.\, \ltypeend \\ \mltype{\ml{stop}}{\booltype}.\, \ltypeend}
    \\
    \sty{U} =
    \ltypesendcases{\pq}{
      \mltype{\ml{cont}}{\inttype}.\sty{U}'
      \\
      \mltype{\ml{stop}}{\booltype}.\sty{U}'
    }
    ~\text{where}~
    \sty{U'} = \ltyperecvcases{\pp}{\mltype{\ml{success}}{\inttype}.\, \ltypeend \\ \mltype{\ml{error}}{\booltype}.\, \ltypeend}
  \end{gather*}
  The computation $u$ is also a valid implementation of $\sty{T}$; indeed, under the definition of \emph{subtyping} below (\cref{def:subtyping}), we have $\sty{U} \subtype \sty{T}$ (but not $\sty{T} \subtype \sty{U}$).

  Both of $\sty{T}$ and $\sty{U}$ permit sending a message to $\pq$, without necessarily waiting for a message from $\pp$; in the case of $\sty{T}$ this relies on asynchrony.
  Similarly, they both require the implementation to accept a message from $\pp$.
  On the other hand, an implementation of $\sty{U}$ is required to send a message to $\pq$, without waiting for one from $\pp$, while an implementation of $\sty{T}$ is permitted to wait.
\end{example}

\subsection{Asynchronous multiparty session subtyping}
Subtyping is a crucial aspect of session type systems.
The intuition is that $\sty{T} \subtype \sty{U}$ holds whenever each implementation of $\sty{T}$ can be used as an implementation of $\sty{U}$, without causing any safety or liveness issues.
Ghilezan et al.~\cite{GPPSY2023} provide a definition of $\sty{T} \subtype \sty{U}$, for closed session types, that is precise for asynchronous message-passing, in the sense that it exactly captures this intuition.
By this metric, it is the best possible notion of subtyping, and so we adopt it here.
However, their definition is inconvenient for several reasons: the definition relies on session \emph{trees} (which are infinite data structures, unlike session \emph{types}), and they do not provide a definition of subtyping for non-closed session types (which we rely on to type \cref{global-state-computations} above).
We therefore provide a more convenient reformulation of asynchronous session typing.
We prove that our definition of subtyping is equivalent to \cite{GPPSY2023}, for closed session types, in the appendix (\cref{thm:gppsy-equivalence}).
For closed session types, our definition of subtyping is equivalent to that of \cite{GPPSY2023}.

As \Cref{ex:gppsy-example-types} demonstrates, in the presence of asynchrony, determining whether an implementation is permitted or required to send or receive a message is non-trivial.
We define relations and predicates on session types that capture exactly when these are the case; these are key to our definition of subtyping.

For sending, the relation $\reduceinternal{\sty{U}}{\sty{U'}}{\pp}{\ell}{\basetype}$ means an implementation of $\sty{U}$ is permitted to send a message $\msg{\ell}{v}$ with $v : \basetype$, to $\pp$, and then follow $\sty{U}'$.
(But it does not necessarily have to send such a message.)
The predicate $\Sends{\pp}(\sty{T})$ means an implementation of $\sty{T}$ is required to send a message to $\pp$, and it is not permitted to wait for a message before doing so.
These are both defined inductively, the rules are at the top of \cref{fig:inductive}.

The two base cases \RuleName{$\reduceinternalsymbol$-base} and \RuleName{Sends-base} are self-explanatory.
The rule \RuleName{$\reduceinternalsymbol$-$\oplus$} encodes the fact that asynchrony does not impose any ordering between messages sent to different participants; since $\pp \neq \pq$, the two messages will be placed into different parts of the send queue, and can be removed from the queue in any order.
Thus, if we know that the implementation is permitted to send a message to $\pp$ in \emph{some} branch of the internal choice on $\pq$, then it is permitted to send a message to $\pp$.
The implementation gets to choose which branches of the internal choice on $\pq$ wishes to keep, because it is an \emph{internal} choice.
\RuleName{Sends-$\oplus$} similarly enables us to reorder messages in session types.
For instance, a participant of type $\ltypesendone{\pq}{\ell'}{\basetype'}{\ltypesendone{\pp}{\ell}{\basetype}{\ltypeend}}$ may send a message to $\pp$ first (the session types $\ltypesendone{\pq}{\ell'}{\basetype'}{\ltypesendone{\pp}{\ell}{\basetype}{\ltypeend}}$ and $\ltypesendone{\pp}{\ell}{\basetype}{\ltypesendone{\pq}{\ell'}{\basetype'}{\ltypeend}}$ are asynchronous subtypes of each other).
For \RuleName{$\reduceinternalsymbol$-$\&$}, even though the implementation is required to accept a message from $\pq$, that does not prevent it from eagerly sending a message to $\pp$, even if $\pp = \pq$. This is the case because sending a message will not block; after sending the message to $\pp$, the implementation will immediately be ready to accept a message from $\pq$.
However, note that we cannot discard branches of the external choice: the implementation does not get to choose to disallow certain messages from $\pq$ just because it is sending a message to $\pp$.
There is no analogous rule for $\Sends{\pp}$, because $\Sends{\pp}$ forbids waiting for a message.
Finally, the two recursion rules \RuleName{$\reduceinternalsymbol$-rec} and \RuleName{Sends-rec} ensure that every session type $\sty{T}$ is treated as equivalent to its unfolding $\unfold{\sty{T}}$.

For receiving, the relation $\reduceexternal{\sty{T}}{\sty{T'}}{\pp}{\ell}{\basetype}$ means an implementation of $\sty{T}$ is \emph{required} to accept a message $\msg{\ell}{v}$ with $v : \basetype$, from $\pp$, if $\pp$ chooses to send one; after doing so, the implementation is required to follow $\sty{T}'$.
The predicate $\Recvs{\pp}(\sty{T})$ means an implementation of $\sty{T}$ is permitted to wait for a message from $\pp$, in particular, it is not required to send any messages before such a message from $\pp$ arrives.
These are again defined inductively, the rules (bottom half of \cref{fig:inductive}) are exactly the duals of the sending rules (i.e.\ we obtain them by swapping sending and receiving).

\begin{figure}[t]
\begin{minipage}[t]{0.65\textwidth}
\begin{gather*}
  \shortintertext{\fbox{$\reduceinternal{\sty{U}}{\sty{U}'}{\pp}{\ell}{\basetype}$}}
  \\[-2.5ex]
  \begin{array}{l}
  \RuleName{$\reduceinternalsymbol$-base}
  ~
  \begin{prooftree}
    \hypo{}
    \infer1
    {\reduceinternal{\ltypesend{\pp}{\ell_i}{\basetype_i}{\sty{U}_i}{i \in I}}{\sty{U}_i}{\pp}{\ell_i}{\basetype_i}}
  \end{prooftree}
  \end{array}
  \\
  \begin{array}{l}
  \RuleName{$\reduceinternalsymbol$-$\oplus$}
  ~
  \begin{prooftree}
    \hypo{\pp \neq \pq}
    \hypo{J \subseteq I}
    \hypo{{\reduceinternal{\sty{U}_i}{\sty{U}'_i}{\pp}{\ell}{\basetype}}~\text{for all}~i \in J}
    \infer3
    {\reduceinternal{\ltypesend{\pq}{\ell'_i}{\basetype'_i}{\sty{U}_i}{i \in I}}{\ltypesend{\pq}{\ell'_i}{\basetype'_i}{\sty{U}'_{i}}{i \in J}}{\pp}{\ell}{\basetype}}
  \end{prooftree}
  \end{array}
  \\[0ex]
  \begin{array}{l}
  \RuleName{$\reduceinternalsymbol$-$\&$}
  ~
  \begin{prooftree}
    \hypo{{\reduceinternal{\sty{U}_i}{\sty{U}'_i}{\pp}{\ell}{\basetype}}~\text{for all}~i \in I}
    \infer1
    {\reduceinternal{\ltyperecv{\pq}{\ell'_i}{\basetype'_i}{\sty{U}_i}{i \in I}}{\ltyperecv{\pq}{\ell'_i}{\basetype'_i}{\sty{U}'_{i}}{i \in I}}{\pp}{\ell}{\basetype}}
  \end{prooftree}
  \end{array}
  \\
  \begin{array}{l}
  \RuleName{$\reduceinternalsymbol$-rec}
  ~
  \begin{prooftree}
    \hypo{\reduceinternal{\subst{\sty{U}}{\ltypevar{X} \mapsto \ltyperec{\ltypevar{X}}{\sty{U}}}}{\sty{U}'}{\pp}{\ell}{\basetype}}
    \infer1
    {\reduceinternal{\ltyperec{\ltypevar{X}}{\sty{U}}}{\sty{U}'}{\pp}{\ell}{\basetype}}
  \end{prooftree}
  \end{array}
\end{gather*}
\vspace{-5ex}
\end{minipage}%
\begin{minipage}[t]{0.35\textwidth}
\begin{gather*}
  \shortintertext{\fbox{$\Sends{\pp}(\sty{T})$}}
  \begin{array}{l}
  \RuleName{Sends-base}
  \\[1ex]
  ~\begin{prooftree}
    \hypo{}
    \infer1
    {\Sends{\pp}(\ltypesend{\pp}{\ell_i}{\basetype_i}{\sty{T}_i}{i \in I})}
  \end{prooftree}
  \end{array}
  \\
  \begin{array}{l}
  \RuleName{Sends-$\oplus$}
  \\[1ex]
  ~\begin{prooftree}
    \hypo{\Sends{\pp}(\sty{T}_j)~~\text{for all}~j \in J}
    \infer1
    {\Sends{\pp}(\ltypesend{\pq}{\ell_j}{\basetype_i}{\sty{T}_j}{j \in J})}
  \end{prooftree}
  \end{array}
  \\
  \begin{array}{l}
  \RuleName{Sends-rec}
  \\[1ex]
  ~\begin{prooftree}
    \hypo{\Sends{\pp}(\sty{T})}
    \infer1
    {\Sends{\pp}(\ltyperec{\ltypevar{X}}{\sty{T}})}
  \end{prooftree}
  \end{array}
\end{gather*}
\end{minipage}

\begin{minipage}[t]{0.65\textwidth}
\begin{gather*}
  \shortintertext{\fbox{$\reduceexternal{\sty{T}}{\sty{T}'}{\pp}{\ell}{\basetype}$}}
  \\[-2.5ex]
  \begin{array}{l}
  \RuleName{$\reduceexternalsymbol$-base}
  ~
  \begin{prooftree}
    \hypo{}
    \infer1
    {\reduceexternal{\ltyperecv{\pp}{\ell_i}{\basetype_i}{\sty{T}_i}{i \in I}}{\sty{T}_i}{\pp}{\ell_i}{\basetype_i}}
  \end{prooftree}
  \end{array}
  \\
  \begin{array}{l}
  \RuleName{$\reduceexternalsymbol$-$\&$}
  ~
  \begin{prooftree}
    \hypo{\pp \neq \pq}
    \hypo{J \subseteq I}
    \hypo{{\reduceexternal{\sty{T}_i}{\sty{T}'_i}{\pp}{\ell}{\basetype}}~\text{for all}~i \in J}
    \infer3
    {\reduceexternal{\ltyperecv{\pq}{\ell'_i}{\basetype'_i}{\sty{T}_i}{i \in I}}{\ltyperecv{\pq}{\ell'_i}{\basetype'_i}{\sty{T}'_{i}}{i \in J}}{\pp}{\ell}{\basetype}}
  \end{prooftree}
  \end{array}
  \\[0ex]
  \begin{array}{l}
  \RuleName{$\reduceexternalsymbol$-$\oplus$}
  ~
  \begin{prooftree}
    \hypo{{\reduceexternal{\sty{T}_i}{\sty{T}'_i}{\pp}{\ell}{\basetype}}~\text{for all}~i \in I}
    \infer1
    {\reduceexternal{\ltypesend{\pq}{\ell'_i}{\basetype'_i}{\sty{T}_i}{i \in I}}{\ltypesend{\pq}{\ell'_i}{\basetype'_i}{\sty{T}'_{i}}{i \in I}}{\pp}{\ell}{\basetype}}
  \end{prooftree}
  \end{array}
  \\
  \begin{array}{l}
  \RuleName{$\reduceexternalsymbol$-rec}
  ~
  \begin{prooftree}
    \hypo{\reduceexternal{\subst{\sty{T}}{\ltypevar{X} \mapsto \ltyperec{\ltypevar{X}}{\sty{T}}}}{\sty{T}'}{\pp}{\ell}{\basetype}}
    \infer1
    {\reduceexternal{\ltyperec{\ltypevar{X}}{\sty{T}}}{\sty{T}'}{\pp}{\ell}{\basetype}}
  \end{prooftree}
  \end{array}
\end{gather*}
\end{minipage}%
\begin{minipage}[t]{0.35\textwidth}
\begin{gather*}
  \shortintertext{\fbox{$\Recvs{\pp}(\sty{U})$}}
  \begin{array}{l}
  \RuleName{Recvs-base}
  \\[1ex]
  ~\begin{prooftree}
    \hypo{}
    \infer1
    {\Recvs{\pp}(\ltyperecv{\pp}{\ell_i}{\basetype_i}{\sty{U}_i}{i \in I})}
  \end{prooftree}
  \end{array}
  \\
  \begin{array}{l}
  \RuleName{Recvs-$\&$}
  \\[1ex]
  ~\begin{prooftree}
    \hypo{\Recvs{\pp}(\sty{U}_j)~~\text{for all}~j \in J}
    \infer1
    {\Recvs{\pp}(\ltyperecv{\pq}{\ell_j}{\basetype_i}{\sty{U}_j}{j \in J})}
  \end{prooftree}
  \end{array}
  \\
  \begin{array}{l}
  \RuleName{Recvs-rec}
  \\[1ex]
  ~\begin{prooftree}
    \hypo{\Recvs{\pp}(\sty{U})}
    \infer1
    {\Recvs{\pp}(\ltyperec{\ltypevar{X}}{\sty{U}})}
  \end{prooftree}
  \end{array}
\end{gather*}
\end{minipage}

\caption{Four inductively defined relations and predicates on multiparty session types}
\label{fig:inductive}
\end{figure}

These relations and predicates provide the bulk of our subtyping definition.
\begin{definition}\label{def:subtyping}
    Let $\Theta$ be a set of session type variables.
    \emph{Asynchronous subtyping} is the largest binary relation $\subtype_{\Theta}$ between session types over $\Theta$, such that the following hold when $\sty{T} \subtype_{\Theta} \sty{U}$.
    (When $\Theta$ is empty, we write just $\sty{T} \subtype \sty{U}$.)
    \begin{enumerate}
        \item If $\unfold{\sty{T}} = \ltypesend{\pp}{\ell_i}{\basetype_i}{\sty{T}_i}{i \in I}$, then for every $i \in I$, there is some $\sty{U}_i$ such that $\reduceinternal{\sty{U}}{\sty{U}_i}{\pp}{\ell_i}{\basetype_i}$ and $\sty{T}_i \subtype_{\Theta} \sty{U}_i$.
        \item If $\unfold{\sty{T}} = \ltyperecv{\pp}{\ell_i}{\basetype_i}{\sty{T}_i}{i \in I}$, then $\Recvs{\pp}(\sty{U})$.
        \item If $\unfold{\sty{U}} = \ltypesend{\pp}{\ell_i}{\basetype_i}{\sty{U}_i}{i \in I}$, then $\Sends{\pp}(\sty{T})$.
        \item If $\unfold{\sty{U}} = \ltyperecv{\pp}{\ell_i}{\basetype_i}{\sty{U}_i}{i \in I}$, then for every $i \in I$, there is some $\sty{T}_i$ such that $\reduceexternal{\sty{T}}{\sty{T}_i}{\pp}{\ell_i}{\basetype_i}$ and $\sty{T}_i \subtype_{\Theta} \sty{U}_i$.
        \item For every $\ltypevar{X} \in \Theta$, we have $\unfold{\sty{T}} = \ltypevar{X}$ if and only if $\unfold{\sty{U}} = \ltypevar{X}$.
    \end{enumerate}
\end{definition}
This is a coinductive definition, and as is usual for coinductive definitions, we can construct $\subtype_{\Theta}$ as the union over all relations satisfying the above conditions; see \cref{union-pre-simulation} for details.

\begin{example}
  Consider the session types of \cref{ex:gppsy-example-types}.
  We do not have $\sty{T} \subtype \sty{U}$; (2) and (3) do not hold, because $\Recvs{\pp}(\sty{U})$ and $\Sends{\pq}(\sty{T})$ are both false.
  However, we do have $\sty{U} \subtype \sty{T}$.
  This is because we have $\sty{T}' \subtype \sty{T}'$ and $\sty{U}' \subtype \sty{U'}$ (subtyping is reflexive by \cref{admissible-subtyping} below).
  To satisfy (1), it is therefore enough to note that $\reduceinternal{\sty{T}}{\sty{U}'}{\pq}{\ml{cont}}{\inttype}$ and $\reduceinternal{\sty{T}}{\sty{U}'}{\pq}{\ml{stop}}{\booltype}$.
  To satisfy (4), it enough to note that $\reduceexternal{\sty{U}}{\sty{T}'}{\pp}{\ml{success}}{\inttype}$ and $\reduceexternal{\sty{U}}{\sty{T}'}{\pp}{\ml{error}}{\booltype}$.
  (2), (3) and (5) are trivial.
\end{example}

\begin{example}\label{ex:weird-example}
  Asynchronous subtyping requires us to take great care with infinite protocols; the following example demonstrates one of the subtleties.
  Consider the following closed session types, where $\pp \neq \pq$.
  \begin{gather*}
    \sty{U} = \ltyperec{\ltypevar{X}}{\ltyperecvcases{\pq}{\mltype{\ell_1}{\basetype_1}.\,\ltyperecvone{\pp}{\ell}{\basetype}{\ltypeend} \\ \mltype{\ell_2}{\basetype_2}.\,\ltypevar{X}}}
    \quad
    \sty{U'} = \ltyperecvone{\pp}{\ell}{\basetype}{\ltyperec{\ltypevar{X}}{\ltyperecvcases{\pq}{\mltype{\ell_1}{\basetype_1}.\,\ltypeend \\ \mltype{\ell_2}{\basetype_2}.\,\ltypevar{X}}}}
  \end{gather*}
  We have neither $\sty{U} \subtype \sty{U}'$ nor $\sty{U}' \subtype \sty{U}$; we first explain informally why these instances of subtyping should not hold.
  Suppose that $\pq$ only sends messages with label $\ell_2$; this is permitted by both $\sty{U}$ and $\sty{U}'$.
  In this case, $\sty{U}' \subtype \sty{U}$ would lead to a failure of liveness: a participant implementing $\sty{U}'$ is permitted to wait for a message from $\pp$, but according to $\sty{U}$, no such message will ever come.
  The converse $\sty{U} \subtype \sty{U}'$ would lead to a different liveness failure: $\sty{U}'$ permits $\pp$ to send a message, but a participant implementing $\sty{U}$ will not consume that message; it will keep waiting for more messages from $\pq$ instead.

  Proving that $\sty{U}' \subtype \sty{U}$ does not hold is easy: it would imply $\Recvs{\pp}(\sty{U})$, which is false.
  To see why $\sty{U} \subtype \sty{U}'$ does not hold, consider the following closed session types, where $k$ ranges over natural numbers.
  \begin{gather*}
    \sty{T}_0 = \ltyperecvone{\pq}{\ell_1}{\basetype_1}{\ltypeend}
    \qquad\mspace{-9mu}
    \sty{T}_{k + 1} = \ltyperecvcases{\pq}{\mltype{\ell_1}{\basetype_1}.\,\ltypeend \\ \mltype{\ell_2}{\basetype_2}.\,\sty{T}_{k}}
    \mspace{-31mu}
    \qquad
    \sty{T}_\infty = \ltyperec{\ltypevar{X}}{\ltyperecvcases{\pq}{\mltype{\ell_1}{\basetype_1}.\,\ltypeend \\ \mltype{\ell_2}{\basetype_2}.\,\ltypevar{X}}}
    \mspace{-25mu}
  \end{gather*}
  We have $\reduceexternal{\sty{T}_{k + 1}}{\sty{T}_k}{\pq}{\ell_2}{\basetype_2}$ and $\reduceexternal{\sty{T}_{\infty}}{\sty{T}_{\infty}}{\pq}{\ell_2}{\basetype_2}$;
  indeed $\sty{T}_{\infty} \subtype \sty{T}_k$ for all $k$.
  However, there is no $k$ such that $\sty{T}_k \subtype \sty{T}_\infty$; this would imply, by an inductive argument, that $\sty{T}_0 \subtype \sty{T}_\infty$, which is false because there is no $\sty{T}'$ such that $\reduceexternal{\sty{T}_0}{\sty{T}'}{\pq}{\ell_2}{\basetype_2}$.
  We have $\reduceexternal{\sty{U}}{\sty{T}_k}{\pp}{\ell}{\basetype}$ for every $k$, and in fact $\reduceexternal{\sty{U}}{\sty{T}'}{\pp}{\ell}{\basetype}$ implies $\sty{T}'$ is a supertype of some $\sty{T}_k$.
  In particular, we do not have $\reduceexternal{\sty{U}}{\sty{T}_\infty}{\pp}{\ell}{\basetype}$; informally this is because, while an implementation of $\sty{U}$ is required to receive a message from $\pp$, it is not then required to implement $\sty{T}_\infty$.
  We therefore do not have $\sty{U} \subtype \sty{U}'$: this subtyping would imply $\reduceexternal{\sty{U}}{\sty{T}'}{\pp}{\ell}{\basetype}$ for some $\sty{T}' \subtype \sty{T}_\infty$, which cannot exist by the above and transitivity of $\subtype$ (\cref{admissible-subtyping} below).
\end{example}

Subtyping is reflexive, transitive, a congruence, and respects substitution:
\begin{restatable}{lemma}{admissiblesubtyping}\label{admissible-subtyping}
    Each of the following rules is admissible (if the premises hold, then so does the conclusion):
    \begin{gather*}\SwapAboveDisplaySkip
      \begin{prooftree}
        \hypo{\mathclap{\phantom{\sty{S} \subtype_{\Theta} \sty{T}}}}
        \hypo{\mathclap{\phantom{\sty{T} \subtype_{\Theta} \sty{U}}}}
        \infer2{\sty{T} \subtype_{\Theta} \sty{T}}
      \end{prooftree}
      \quad
      \begin{prooftree}
        \hypo{\sty{S} \subtype_{\Theta} \sty{T}}
        \hypo{\sty{T} \subtype_{\Theta} \sty{U}}
        \infer2{\sty{S} \subtype_{\Theta} \sty{U}}
      \end{prooftree}
      \\[1ex]
      \begin{prooftree}
        \hypo{J \subseteq I}
        \hypo{\sty{T}_i \subtype_{\Theta} \sty{U}_i~\text{for all}~i \in J}
        \infer2{(\ltypesend{\pp}{\ell_i\mspace{-2mu}}{\mspace{-2mu}\basetype_i}{\sty{T}_i}{i \in J})\mspace{-2mu} \subtype_{\Theta} \mspace{-2mu}(\ltypesend{\pp\mspace{-2mu}}{\mspace{-2mu}\ell_i}{\basetype_i}{\sty{U}_i}{i \in I})}
      \end{prooftree}
      \quad\!
      \begin{prooftree}
        \hypo{J \subseteq I}
        \hypo{\sty{T}_i \subtype_{\Theta} \sty{U}_i~\text{for all}~i \in J}
        \infer2{(\ltyperecv{\pp\mspace{-2mu}}{\mspace{-2mu}\ell_i}{\basetype_i}{\sty{T}_i}{i \in I})\mspace{-2mu} \subtype_{\Theta} \mspace{-2mu}(\ltyperecv{\pp\mspace{-2mu}}{\mspace{-2mu}\ell_i}{\basetype_i}{\sty{U}_i}{i \in J})}
      \end{prooftree}
      \\[1ex]
      \begin{prooftree}
        \hypo{\ltypevar{X} \not\in \Theta}
        \hypo{\sty{T} \subtype_{\Theta, \ltypevar{X}} \sty{U}}
        \infer2{\phantom{\mathclap{\subst{\sty{T}}{\ltypevar{X_1} \mapsto \sty{T}_1, \dots, \ltypevar{X_n} \mapsto \sty{T}_n} \subtype_{\Theta} {\subst{\sty{U}}{\ltypevar{X_1} \mapsto \sty{U}_1, \dots, \ltypevar{X_n} \mapsto \sty{U}_n}}}}\ltyperec{\ltypevar{X}}{\sty{T}} \subtype_{\Theta} \ltyperec{\ltypevar{X}}{\sty{U}}}
      \end{prooftree}
      \quad
      \begin{prooftree}
        \hypo{\mathclap{\phantom{}}\sty{T} \subtype_{\ltypevar{X}_1, \dots, \ltypevar{X}_n} \sty{U}}
        \hypo{\sty{T}_1 \subtype_{\Theta} \sty{U}_1}
        \hypo{\!\!\cdots\!\!}
        \hypo{\sty{T}_n \subtype_{\Theta} \sty{U}_n}
        \infer4{\subst{\sty{T}}{\ltypevar{X_1} \mapsto \sty{T}_1, \dots, \ltypevar{X_n} \mapsto \sty{T}_n} \subtype_{\Theta} {\subst{\sty{U}}{\ltypevar{X_1} \mapsto \sty{U}_1, \dots, \ltypevar{X_n} \mapsto \sty{U}_n}}}
      \end{prooftree}
    \end{gather*}
\end{restatable}

\subsection{Sequencing}
Unlike in previous works on session types, we have a notion of \emph{sequencing} of computations, namely $\bindterm{x}{t}{u}$.
To assign a type to a sequence, we use a \emph{multiplication} operation $\sty{T} \effmul \sty{T'}$ on session types.
This essentially replaces $\ltypeend$ with $\sty{T}'$ in $\sty{T}$ (once $t$ is done, we run $u$).
The definition is by recursion on $\sty{T}$:
\[
  \begin{array}{r@{~}c@{~}l}
  \ltypeend \effmul \sty{T}'
  &=&
  \sty{T}'
  \\
  \ltypevar{X} \effmul \sty{T}'
  &=&
  \ltypevar{X}
  \\[0.5ex]
  (\ltyperec{\ltypevar{X}}{\sty{T}}) \effmul \sty{T}'
  &=&
  \ltyperec{\ltypevar{X}}{(\sty{T} \effmul \sty{T}')}
  ~~\mathrlap{\text{if}~\ltypevar{X}~\text{not free in}~\sty{T'}}
  \end{array}
  \quad
  \begin{array}{r@{~}c@{~}l}
  (\ltypesend{\pp}{\ell_i}{\basetype_i}{\sty{T}_i}{i \in I}) \effmul \sty{T}'
  &=&
  \ltypesend{\pp}{\ell_i}{\basetype_i}{(\sty{T}_i \effmul \sty{T}')}{i \in I}
  \\
  (\ltyperecv{\pp}{\ell_i}{\basetype_i}{\sty{T}_i}{i \in I}) \effmul \sty{T}'
  &=&
  \ltyperecv{\pp}{\ell_i}{\basetype_i}{(\sty{T}_i \effmul \sty{T}')}{i \in I}
  \\[0.5ex]
  \phantom{\mathclap{(\ltyperec{\ltypevar{X}}{\sty{T}}) \effmul \sty{T}'}}
  \end{array}
\]
The set of closed session types forms a preordered monoid: the preorder is the asynchronous subtyping relation $\sty{T} \effleq \sty{U}$; the monoid operation is our multiplication $\sty{T} \effmul \sty{U}$; and the unit of the monoid is $\ltypeend$.
Multiplication $(\effmul)$ is associative, and also monotone (if $\sty{T} \effleq \sty{U}$ and $\sty{T}' \effleq \sty{U}'$ then $\sty{T} \effmul \sty{T}' \effleq \sty{U} \effmul \sty{U}'$).

In general, following Katsumata~\cite{katsumata2014parametric}, a preordered monoid of \emph{grades} is the basic structure required to give a graded type system; in particular, the multiplication operation is required to give a graded typing rule for sequencing.
The grades in this paper are session types, and thus the fact that we can organize session types into a preordered monoid is crucial.
We use the multiplication operation in the typing rule \RuleName{Let} of the following section.

\section{Session-type-graded type system for $\lambdasmp{}$}\label{sec:graded-calculus}

We now come to our type system for $\lambdasmp{}$.
The goal is to assign, to each configuration, a type $\basetype$ and a session
type $\sty{T}$, but to do so in such a way that we can prove safety and liveness properties.
This is a \emph{graded} type system in the terminology of for instance~\cite{orchard2019quantitative}; the \emph{grades} here are the session types.
Indeed, the typing rules for sequencing, returning and for subtyping are standard from graded type systems.
This section demonstrates that we can view session types as an instance of grading.

We first define a typing judgement for values. This has the form $\valuetypedg{\Gamma}{v}{\basetype}$, where 
the typing context $\Gamma$ assigns types $\basetype'$ to variables $x$.
The rules for value typing are at the top of \cref{typing-rules}.

\begin{figure}[t]
  \vspace{-1.5ex}
  \begin{gather*}
    \shortintertext{Value typing $\boxed{\valuetypedg{\Gamma}{v}{\basetype}}$}
    \begin{prooftree}
      \hypo{(x : \basetype) \in \Gamma}
      \infer1{\valuetypedg{\Gamma}{x}{\basetype}}
    \end{prooftree}
    \quad
    \begin{prooftree}
      \hypo{\phantom{(x : \basetype) \in \Gamma}}
      \infer1{\valuetypedg{\Gamma}{\unitterm}{\unittype}}
    \end{prooftree}
    \quad
    \begin{prooftree}
      \hypo{\phantom{(x : \basetype) \in \Gamma}}
      \infer1{\valuetypedg{\Gamma}{n}{\inttype}}
    \end{prooftree}
    \quad
    \begin{prooftree}
      \hypo{\phantom{(x : \basetype) \in \Gamma}}
      \infer1{\valuetypedg{\Gamma}{\trueterm}{\booltype}}
    \end{prooftree}
    \quad
    \begin{prooftree}
      \hypo{\phantom{(x : \basetype) \in \Gamma}}
      \infer1{\valuetypedg{\Gamma}{\falseterm}{\booltype}}
    \end{prooftree}
  \end{gather*}
  \vspace{-2ex}
  \begin{gather*}
    \shortintertext{Computation typing $\boxed{\comptypedg{\Theta}{\Psi}{\Gamma}{t}{\basetype}{\sty{T}}}$}
\TRULE{$\subtype$}    
    \frac
      {
        \comptypedg{\Theta}{\Psi}{\Gamma}{t}{\basetype}{\sty{T}}
        \quad
        \sty{T} \subtype_{\Theta} \sty{U}
      }
      {\comptypedg{\Theta}{\Psi}{\Gamma}{t}{\basetype}{\sty{U}}}
    \quad
\TRULE{Ret}        
    \frac
      {\valuetypedg{\Gamma}{v}{\basetype}}
      {\comptypedg{\Theta}{\Psi}{\Gamma}{\returnterm{v}}{\basetype}{\effunit}}
  \\[1ex]
\TRULE{Let}        
    \frac
      {
        \comptypedg{\Theta}{\Psi}{\Gamma}{t}{\basetype}{\sty{T}}
        \quad
        \comptypedg{\Theta}{\Psi}{\Gamma, x : \basetype}{u}{\basetype'}{\sty{T}'}
      }
      {\comptypedg{\Theta}{\Psi}{\Gamma}{\bindterm{x}{t}{u}}{\basetype'}{\sty{T} \effmul \sty{T}'}}
    \quad
\TRULE{$+$}        
\frac
      {~
        \valuetypedg{\Gamma}{v}{\inttype}
        \quad
        \valuetypedg{\Gamma}{w}{\inttype}
      }
      {\comptypedg{\Theta}{\Psi}{\Gamma}{\addterm{v}{w}}{\inttype}{\effunit}}
  \\[0.5ex]
\TRULE{$<$}            
    \frac
      {
        \valuetypedg{\Gamma}{v}{\inttype}
        \quad
        \valuetypedg{\Gamma}{w}{\inttype}
      }
      {\comptypedg{\Theta}{\Psi}{\Gamma}{\lessterm{v}{w}}{\booltype}{\effunit}}
    ~
\TRULE{If}                
    \frac
      {
        \valuetypedg{\Gamma}{v}{\booltype}
        \quad\!
        \comptypedg{\Theta}{\Psi}{\Gamma}{t_i}{\basetype}{\sty{T}}
        ~\text{for all}~i \in \{1, 2\}
      }
      {\comptypedg{\Theta}{\Psi}{\Gamma}{\ifterm{v}{t_1}{t_2}}{\basetype}{\sty{T}}}
  \\[0.8ex]
\TRULE{Send}                        
    \frac
      {
        \valuetypedg{\Gamma}{v}{\basetype}
        \quad
        \comptypedg{\Theta}{\Psi}{\Gamma}{t}{\basetype'}{\sty{T}}
      }
      {\comptypedg{\Theta}{\Psi}{\Gamma}{\sendterm{\ell}{v}{\pp}{t}}{\basetype'}{(\ltypesendone{\pp}{\ell}{\basetype}{\sty{T}})}}
  \\[0.8ex]
\TRULE{Recv}                        
    \frac
      {
        \comptypedg{\Theta}{\Psi}{\Gamma, x_i : \basetype_i}{t_i}{\basetype'}{\sty{T}_i}
        ~\text{for all}~i \in I
      }
      {\comptypedg{\Theta}{\Psi}{\Gamma}{\recvterm{\pp}{\ell_i}{x_i : \basetype_i}{t_i}{i \in I}}{\basetype'}{(\ltyperecv{\pp}{\ell_i}{\basetype_i}{\sty{T}_i}{i \in I})}}
  \\[0.8ex]
  \TRULE{LetRec}
    \frac
      {
        \begin{array}{@{}l@{}}
        \comptypedg{\Theta, X}{\Psi, f : (\basetype_1, \dots, \basetype_n) \gto{\ltypevar{X}} \basetype'}{\Gamma, x_1 : \basetype_1, \dots, x_n : \basetype_n}{t}{\basetype'}{\sty{T}}
        \\
        \comptypedg{\Theta}{\Psi, f : (\basetype_1, \dots, \basetype_n) \gto{\ltyperec{\ltypevar{X}}{\sty{T}}} \basetype'}{\Gamma}{u}{\basetype''}{\sty{T}'}
        \end{array}
      }
      {\comptypedg{\Theta}{\Psi}{\Gamma}{\letrecterm{f}{(x_1, \dots, x_n)}{t}{u}}{\basetype''}{\sty{T}'}}
  \\[0.5ex]
  \TRULE{App}
    \frac
      {
        (f : (\basetype_1, \dots, \basetype_n) \gto{\sty{T}} \basetype') \in \Psi
        \quad
        \valuetypedg{\Gamma}{v_1}{\basetype_1}
        ~\cdots~
        \valuetypedg{\Gamma}{v_n}{\basetype_n}
      }
      {\comptypedg{\Theta}{\Psi}{\Gamma}{f(v_1, \dots, v_n)}{\basetype'}{\sty{T}}}
  \end{gather*}
  \vspace{-3ex}
  \begin{gather*}
    \shortintertext{Configuration typing $\boxed{\conftypedg{\config{\queue{\rho}}{t}{\queue{\sigma}}}{\basetype}{\sty{T}}}$}
    \begin{array}{l}
    \TRULE{CBase}
    \\[1ex]\displaystyle
    \frac
      {\comptypedg{\emptycontext}{\emptycontext}{\emptycontext}{t}{\basetype}{\sty{T}}}
      {\conftypedg{\config{\queueempty}{t}{\queueempty}}{\basetype}{\sty{T}}}
    \end{array}
    \quad
    \begin{array}{l}
    \TRULE{CSend}
    \\[-1.2ex]\displaystyle
    \frac
      {
        \begin{array}{l}
        v : \basetype'
        \quad
        \reduceinternal{\sty{U}}{\sty{T}}{\pp}{\ell}{\basetype'}
        \\
        \conftypedg{\config{\queue{\rho}}{t}{\queue{\sigma}}}{\basetype}{\sty{T}}
        \end{array}
      }
      {\conftypedg{\config{\queue{\rho}}{t}{\consback{\queue{\sigma}}{\pp}{\msg{\ell}{v}}}}{\basetype}{\sty{U}}}
    \end{array}
    \quad
    \begin{array}{l}
    \TRULE{CRecv}
    \\[-1.2ex]\displaystyle
    \frac
      {
        \begin{array}{l}
        v : \basetype'
        \quad
        \reduceexternal{\sty{T}}{\sty{U}}{\pp}{\ell}{\basetype'}
        \\
        \conftypedg{\config{\queue{\rho}}{t}{\queue{\sigma}}}{\basetype}{\sty{T}}
        \end{array}
      }
      {\conftypedg{\config{\consfront{\pp}{\msg{\ell}{v}}{\queue{\rho}}}{t}{\queue{\sigma}}}{\basetype}{\sty{U}}}
      \end{array}
  \end{gather*}
  \vspace{-2.9ex}
  \caption{Typing of values, computations, and configurations in $\lambdasmp{}$\label{typing-rules}}
\end{figure}

We then define a typing judgement for computations, of the form $\comptypedg{\Theta}{\Psi}{\Gamma}{t}{\basetype}{\sty{T}}$.
The session type $\sty{T}$ describes the behaviour of the computation $t$, in terms of the protocol it follows.
Here $\Theta$ is a finite set of type variables; this contains all of the free type variables in the other components of the judgement.
The \emph{function typing context} $\Psi$ tracks the recursive function definitions that are in scope.
It maps function names $f$ to triples of the form $(\basetype_1, \dots, \basetype_n) \gto{\sty{U}} \basetype'$, where $(\basetype_1, \dots, \basetype_n)$ is a possibly empty list of (argument) types, $\sty{U}$ is a session type over $\Theta$, and $\basetype'$ is a (result) type.
The session type $\sty{U}$ describes the protocol followed by an application $\applyterm{f}{(v_1, \dots, v_n)}$.
Annotating function types with a grade in this way is standard from the grading literature, the annotation is sometimes called the \emph{latent effect} of the function.
The symbol $\fatsemi$ separates the different components of the judgement, it has no meaning by itself.

The computation typing judgement is defined inductively, by the rules in the middle of \cref{typing-rules}.
We include a session subtyping rule \RuleName{$\subtype$}, using the above definition of subtyping; this is useful for typing conditionals, because \RuleName{If} requires the two branches to have the same type.
The \RuleName{Send} and \RuleName{Recv} rules record the send/receive in the session type assigned to the computation.
The \RuleName{LetRec} rule requires the body of the recursive function to have type $\sty{T}$ under the assumption that a recursive call will have type $\ltypevar{X}$; it then assigns the recursive session type $\ltyperec{\ltypevar{X}}{\sty{T}}$ to the recursive function.

Finally, the typing judgement for configurations has the form $\conftypedg{\config{\queue{\rho}}{t}{\queue{\sigma}}}{\basetype}{\sty{T}}$, and is defined inductively by the rules at the bottom of \cref{typing-rules}.
There are no contexts because $t$ is a closed computation; the rules ensure that $\sty{T}$ is a closed session type.
The \RuleName{CBase} rule uses our computation typing judgement with empty contexts (we write $\emptycontext$ for an empty context).
The intuition for \RuleName{CSend} is that the configuration in the conclusion is sending a message to $\pp$; we therefore need to ensure that the type $\sty{U}$ permits it to do so, and we do this using our relation $\reduceinternalsymbol$.
For \RuleName{CRecv}, the conclusion is only well-typed if the configuration in the assumption is able to receive a message from $\pp$; we ensure this is the case using $\reduceexternalsymbol$.
We do not explicitly include a subtyping rule for configurations, but such a rule is admissible: if $\conftypedg{\config{\queue{\rho}}{t}{\queue{\sigma}}}{\basetype}{\sty{T}}$ and $\sty{T} \subtype \sty{U}$, then $\conftypedg{\config{\queue{\rho}}{t}{\queue{\sigma}}}{\basetype}{\sty{U}}$.

\begin{example}\label{ex:gppsy-example-graded}
  Recall the computations and session types from \cref{ex:gppsy-example} and \cref{ex:gppsy-example-types}.
  We have $\conftypedg{\config{\queueempty}{t}{\queueempty}}{\booltype}{\sty{T}}$, and $\comptypedg{\emptycontext}{\emptycontext}{y : \booltype}{u}{\booltype}{\sty{U}}$.
  Both of these use subtyping to type the conditionals.
  Since $\sty{U} \subtype \sty{T}$, we therefore also have $\comptypedg{\emptycontext}{\emptycontext}{y : \booltype}{u}{\booltype}{\sty{T}}$, so for instance we can assign to $\config{\queueempty}{\bindterm{y}{\falseterm}{u}}{\queueempty}$ the same type as $\config{\queueempty}{t}{\queueempty}$.
\end{example}

\begin{example}\label{ex:global-state-example-graded}
  Recall our global state example, from \cref{global-state-computations}.
  Consider the following session types, where $\ltypevar{X}$ is a type variable.
  \[
    \sty{T}_\role{s} = \ltyperecvcases{\role{c}}{
      \mltype{\ml{get}}{\unittype}.\,\ltypesendone{\role{c}}{\ml{st}}{\inttype}{\ltypevar{X}}
      \\
      \mltype{\ml{put}}{\inttype}.\,\ltypevar{X}
      \\
      \mltype{\ml{done}}{\unittype}.\,\ltypeend
    }
    \quad
    \sty{T}_\role{c} = \ltypesendcases{\role{s}}{
      {}
      \mltype{\ml{get}}{\unittype}.\,\ltyperecvone{\role{s}}{\ml{st}}{\inttype}{\ltypevar{X}}
      \\
      \mltype{\ml{put}}{\inttype}.\,\ltypevar{X}
      \\
      \mltype{\ml{done}}{\unittype}.\,\ltypeend
    }
  \]
  We have 
  $\conftypedg{\mathcal{C}_{\role{s}}}{\unittype}{\ltyperec{\ltypevar{X}}{\sty{T}_\role{s}}}$,
  and
  $\conftypedg{\mathcal{C}_{\role{c}, 1}}{\inttype}{\ltyperec{\ltypevar{X}}{\sty{T}_\role{c}}}$
  and $\conftypedg{\mathcal{C}_{\role{c}, 2}}{\inttype}{\ltyperec{\ltypevar{X}}{\sty{T}_\role{c}}}$,
  where
  \begin{gather*}
  \mathcal{C}_{\role{s}} = \config{\queueempty}{t_{\role{s}}}{\queueempty}
  \\
  \mathcal{C}_{\role{c}, i} = \config{\queueempty}{\bindterm{x}{t_{\role{c}, i}}{\sendterm{\ml{done}}{\unitterm}{\role{s}}{\returnterm{x}}}}{\queueempty}
  ~~(i \in \{1, 2\})
  \end{gather*}
  The derivation for $\mathcal{C}_{\role{s}}$ involves a derivation of
  \[\comptypedg{\ltypevar{X}}{f : \inttype \gto{\ltypevar{X}} \unittype}{x: \inttype}{\longrecvterm{\role{c}}{\dots}}{\unittype}{\sty{T}_{\role{s}}}\]
  where $\longrecvterm{\role{c}}{\dots}$ is the body of the recursive function in $t_{\role{s}}$.
\end{example}

The subject reduction theorem uses the inductive relations of \cref{fig:inductive}.
\begin{restatable}[Subject reduction]{theorem}{subjred}\label{subject-reduction}
  Assume that $\conftypedg{\config{\rho}{t}{\sigma}}{\basetype}{\sty{T}}$, and that $\reduceaction{\config{\rho}{t}{\sigma}}{\config{\rho'}{t'}{\sigma'}}{\alpha}$.
  If $\alpha = \tauaction$, then $\conftypedg{\config{\rho'}{t'}{\sigma'}}{\basetype}{\sty{T}}$. If $\alpha = \sendaction{\pp}{\msg{\ell}{v}}$ with $v : \basetype'$, then $\conftypedg{\config{\rho'}{t'}{\sigma'}}{\basetype}{\sty{U}}$ for some $\sty{U}$ such that $\reduceinternal{\sty{T}}{\sty{U}}{\pp}{\ell}{\basetype'}$. If $\alpha = \recvaction{\pp}{\msg{\ell}{v}}$ with $v : \basetype'$, then $\conftypedg{\config{\rho'}{t'}{\sigma'}}{\basetype}{\sty{U}}$ for every $\sty{U}$ such that $\reduceexternal{\sty{T}}{\sty{U}}{\pp}{\ell}{\basetype'}$.
\end{restatable}

\subsection{Typed bisimulation}
To compare a configuration to its interpretation in our denotational semantics below, we use a notion of equivalence between states of transition systems.
Our notion of equivalence is based on \emph{bisimulation}.
However, in the context of session types, the appropriate notion of equivalence is one that is informed by the type.
To see why, consider \cref{ex:gppsy-example-graded}.
The configurations $\config{\queueempty}{t}{\queueempty}$ and $\config{\queueempty}{\bindterm{y}{\falseterm}{u}}{\queueempty}$ do not have the same behaviour; for instance, if $\pp$ sends $\msg{\ml{success}}{1}$, then they send different messages to $\pq$.
However, if we know that $\pp$ cannot send $\ml{success}$, then they are equivalent.
In particular, since $\sty{T} \subtype (\ltyperecvone{\pp}{\ml{error}}{\booltype}{\sty{T}'})$, we can assign the session type $\ltyperecvone{\pp}{\ml{error}}{\booltype}{\sty{T}'}$ to both configurations; by doing so, we are requiring $\pp$ to send an $\ml{error}$ message and not $\ml{success}$.
The aforementioned configurations have the same behaviour when considered as configurations of type $\ltyperecvone{\pp}{\ml{error}}{\booltype}{\sty{T}'}$.
We define a general notion of \emph{typed bisimulation} for \emph{typed transition systems}, to account for the session types.

\begin{definition}
  A \emph{typed transition system} $(S, \reducesymbol, \Result, \hastype{}{})$ with result set $X$, consists of a set $S$ of \emph{states}, a binary relation $\reduceaction{}{}{\alpha}$ on $S$ for each local action $\alpha$, a partial function $\Result$ from $S$ to $X$, and a relation $\hastype{}{}$ between states and closed session types.
  We require that $\hastype{s}{\sty{T}}$ implies $\hastype{s}{\sty{U}}$ whenever $\sty{T} \subtype \sty{U}$.
\end{definition}
We write $\reduceactionmany{s}{s'}{\alpha}$ when there is a finite sequence of reductions ending in action $\alpha$, where all the reductions before the last have action $\tauaction$.
When $\alpha = \tauaction$ we permit the sequence of reductions to be empty ($s = s'$), while for every other $\alpha$ we require there to be at least one reduction.

For each $\basetype$, we obtain a typed transition system with result set $\{v \mid v : \basetype\}$: states are configurations, $\reducesymbol$ and $\Result$ are as defined in \cref{sec:basic-calculus}, and the typing relation $\hastype{\config{\queue{\rho}}{t}{\queue{\sigma}}}{\sty{T}}$ holds when $\conftypedg{\config{\queue{\rho}}{t}{\queue{\sigma}}}{\basetype}{\sty{T}}$.
Our notion of typed bisimulation is as follows.
\begin{definition}\label{def:bisim}
  Let $(S, \reducesymbol, \Result, \hastype{}{})$ and $(S', \reducesymbol, \Result, \hastype{}{})$ be typed transition systems.
  A family of relations $R_{\sty{S}} \subseteq S \times S'$, indexed by closed session types $\sty{S}$, is a \emph{typed bisimulation} when $s R_{\sty{S}} s'$ implies $\hastype{s}{\sty{S}}$, $\hastype{s'}{\sty{S}}$, and also the following.
  \begin{enumerate}
    \item (a) $\reduceaction{s}{t}{\tauaction}$ implies there exists $t'$ such that $\reduceactionmany{s'}{t'}{\tauaction}$ and $t R_{\sty{S}} t'$; and
          (b) $\reduceaction{s'}{t'}{\tauaction}$ implies there exists $t$ such that $\reduceactionmany{s}{t}{\tauaction}$ and $t R_{\sty{S}} t'$.
    \item If~$\unfold{\sty{S}} = \ltypeend$, then (a) $\Result(s) = x$ implies there exists $t'$ such that $\reduceactionmany{s'}{t'}{\tauaction}$ and $\Result(t') = x$; and
          (b) $\Result(s') = x$ implies there exists $t$ such that $\reduceactionmany{s}{t}{\tauaction}$ and $\Result(t) = x$.
    \item If $\Sends{\pp}(\sty{S})$ and $v : \basetype$, then (a) $\reduceaction{s}{t}{\sendaction{\pp}{\msg{\ell}{v}}}$ implies there exist $\sty{T}, t'$ such that $\reduceinternal{\sty{S}}{\sty{T}}{\pp}{\ell}{\basetype}$, $\reduceactionmany{s'}{t'}{\sendaction{\pp}{\msg{\ell}{v}}}$ and $t R_{\sty{T}} t'$; and
          (b) $\reduceaction{s'}{t'}{\sendaction{\pp}{\msg{\ell}{v}}}$ implies there exist $\sty{T}, t$ such that $\reduceinternal{\sty{S}}{\sty{T}}{\pp}{\ell}{\basetype}$, $\reduceactionmany{s}{t}{\sendaction{\pp}{\msg{\ell}{v}}}$ and $t R_{\sty{T}} t'$.
    \item If $\reduceexternal{\sty{S}}{\sty{T}}{\pp}{\ell}{\basetype}$ and $v : \basetype$, then (a) $\reduceaction{s}{\hastype{t}{\sty{T}}}{\recvaction{\pp}{\msg{\ell}{v}}}$ implies there exists $\sty{T}, t'$ such that $\reduceactionmany{s'}{t'}{\recvaction{\pp}{\msg{\ell}{v}}}$ and $t R_{\sty{T}} t'$; and
          (b) $\reduceaction{s'}{{\hastype{t'}{\sty{T}}}}{\recvaction{\pp}{\msg{\ell}{v}}}$ implies there exists $t$ such that $\reduceactionmany{s}{t}{\recvaction{\pp}{\msg{\ell}{v}}}$ and $t R_{\sty{T}} t'$.
  \end{enumerate}
  We write $\bisim$ for the largest typed bisimulation.
\end{definition}
The crucial points to note are that in (3), we only impose requirements on messages sent to $\pp$ when $\Sends{\pp}$ holds, because $\pp$ can only consume a message when that is the case; and that in (4) we only impose requirements on a message received from $\pp$ when the session type requires that message to be accepted.
As is usual for a coinductive definition, $\bisim$ is equal to the union of all typed bisimulations; see \cref{union-bisimulation} for details.
\begin{example}\label{ex:global-state-example-bisim}
  Returning to \cref{ex:gppsy-example-graded}, we have $\config{\queueempty}{t}{\queueempty} \bisim_{\sty{S}} \config{\queueempty}{\bindterm{y}{\falseterm}{u}}{\queueempty}$ for $\sty{S} = \ltyperecvone{\pp}{\ml{error}}{\booltype}{\sty{T}'}$, but not for $\sty{S} = \sty{T}$.
  The difference is that, while $\reduceexternal{\sty{T}}{\sty{T}'}{\pp}{\ml{success}}{\inttype}$ holds, no $\sty{S}'$ satisfies $\reduceexternal{(\ltyperecvone{\pp}{\ml{error}}{\booltype}{\sty{T}'})}{\sty{S}'}{\pp}{\ml{success}}{\inttype}$.

  In the setting of \cref{ex:global-state-example-graded}, we have
  $\config{\queueempty}{t_{\role{c}, 1}}{\queueempty} \bisim_{\sty{S}}\config{\queueempty}{t_{\role{c}, 2}}{\queueempty}$ where
  $\sty{S} = \ltypesendone{\role{s}}{\ml{get}}{\unittype}{\ltyperecvone{\role{s}}{\ml{st}}{\inttype}{\ltypesendone{\role{s}}{\ml{put}}{\inttype}{\ltypeend}}}$,
  even though $t_{\role{c}, 1}$ cannot send a $\ml{put}$ message until it receives a message from $\role{s}$.
  As a consequence, we can view $t_{\role{c}, 2}$ as a sound optimisation of $t_{\role{c}, 1}$.
  Our denotational semantics for $\lambdasmp{}$ provides a sound and complete technique for proving this bisimilarity; see \cref{ex:global-state-example-semantics} below.
  This bisimilarity does not rely in any way on the implementation of $\role{s}$.
\end{example}

\begin{restatable}{lemma}{bisimper}
  For every $\sty{S}$, the relation $\bisim_{\sty{S}}$ is transitive and symmetric.
  Moreover, if $s \bisim_{\sty{S}} s'$, and $\sty{S} \subtype \sty{T}$, then $s \bisim_{\sty{T}} s'$.
\end{restatable}

\section{Computation trees}\label{sec:computation-trees}

A session type describes the protocol that a given participant $\pp$ must follow, in terms of the interactions it is meant to have with other participants in a system.
The purpose of this section is to give some semantic meaning to this.
We do this by introducing \emph{computation trees} as a notion of asynchronous message-passing computation, and discussing how they relate to session trees.
\begin{definition}
  \emph{Computation trees}, over a set $X$ of \emph{results}, are generated coinductively by the following grammar,
  where $x$ ranges over elements of $X$, $m$ ranges over messages, and $M$ ranges over sets of messages.
  \[
    t \Coloneqq
    \treereturn{x}
    ~|~
    \treesend{\role{p}}{m}{t}
    ~|~
    \treerecv{\role{p}}{t_m}{m \in M}
  \]
\end{definition}

Our first task is to show that these support asynchronous message-passing, without the use of queues.
To do this, we make them into a transition system labelled by local actions $\alpha$.
Reduction $\reduceaction{t}{u}{\alpha}$ is defined inductively by the following rules (there are no $\tauaction$ transitions).
The base cases \RuleName{Send} and \RuleName{Recv} are obvious, while the congruence rules make this notion of reduction an asynchronous one.
In particular, to receive a message from $\pp$, we do not need the root of the tree to be a $\treerecvkeyword_{\pp}$.
The congruence rules enable us to reduce occurences of $\treerecvkeyword_{\pp}$ appearing inside deeper in the tree, potentially discarding branches of other $\treerecvkeyword$s (which might not contain $\treerecvkeyword_{\pp}$).
\begin{gather*}
  \RuleName{Send}~
  \begin{prooftree}
    \infer0
    {\reduceaction{\treesend{\pp}{m}{t}}{t}{\sendaction{\pp}{m}}}
  \end{prooftree}
  \quad
  \RuleName{Recv}~
  \begin{prooftree}
    \hypo{m \in M}
    \infer1
    {\reduceaction{\treerecv{\pp}{t_{m'}}{m' \in M}}{t_m}{\recvaction{\pp}{m}}}
  \end{prooftree}
  \\
  \RuleName{SendSend}
  \begin{prooftree}
    \hypo{\pp \neq \pq}
    \hypo{\reduceaction{t}{u}{\sendaction{\pp}{m}}}
    \infer2
    {\reduceaction{\treesend{\pq}{m'}{t}}{\treesend{\pq}{m'}{u}}{\sendaction{\pp}{m}}}
  \end{prooftree}
  ~
  \RuleName{RecvSend}
  \begin{prooftree}
    \hypo{\reduceaction{t}{u}{\recvaction{\pp}{m}}}
    \infer1
    {\reduceaction{\treesend{\pq}{m'}{t}}{\treesend{\pq}{m'}{u}}{\recvaction{\pp}{m}}}
  \end{prooftree}
  \\
  \RuleName{RecvRecv}~
  \begin{prooftree}
    \hypo{\pp \neq \pq}
    \hypo{M' \subseteq M}
    \hypo{\reduceaction{t_{m'}}{u_{m'}}{\recvaction{\pp}{m}}~\text{for all}~m' \in M'}
    \infer3
    {\reduceaction{\treerecv{\pq}{t_{m'}}{m' \in M}}{\treerecv{\pq}{u_{m'}}{m' \in M'}}{\recvaction{\pp}{m}}}
  \end{prooftree}
\end{gather*}

To make this into a typed transition system, we define $\Result(\treereturn{x}) = x$, with $\Result(t)$ undefined otherwise.
The typing relation for computation trees is defined coinductively, in a similar manner to our definition of subtyping.

\begin{definition}\label{def:computation-tree-typing}
We define a typing relation $\hastype{t}{\sty{T}}$ between computation trees $t$ and closed session types $\sty{T}$ coinductively, as the largest relation such that the following hold when $\hastype{t}{\sty{T}}$.
  \begin{enumerate}
    \item If $t = \treesend{\role{p}}{\msg{\ell}{v}}{u}$, with $v : \basetype$, then there is some $\sty{U}$ such that $\reduceinternal{\sty{T}}{\sty{U}}{\pp}{\ell}{\basetype}$ and $\hastype{u}{\sty{U}}$.
    \item If $t = \treerecv{\role{p}}{t_m}{m \in M}$, then $\Recvs{\pp}(\sty{T})$.
    \item If $\unfold{\sty{T}} = \ltypesend{\pp}{\ell_i}{\basetype_i}{\sty{T}_i}{i \in I}$, then there exist $m, u$ such that $\reduceaction{t}{u}{\sendaction{\pp}{m}}$.
    \item If $\unfold{\sty{T}} = \ltyperecv{\pp}{\ell_i}{\basetype_i}{\sty{T}_i}{i \in I}$, then there is some natural number $h$ such that, for every $i \in I$ and $v : \basetype_i$, there is some $\hastype{u}{\sty{T}_i}$ such that $\reduceaction{t}{u}{\recvaction{\pp}{\msg{\ell_i}{v}}}$, where the derivation of the latter has height at most $h$.\footnotemark{}
  \end{enumerate}
\end{definition}
\footnotetext{%
  We need such an $h$ to exist for the same reason that we do not want $\sty{U} \subtype \sty{U}'$ in \cref{ex:weird-example}: if there were no such $h$, then that would mean there is an infinite reduction sequence beginning with $t$, that never involves receiving a message from $\pp$; this would lead to a failure of liveness.
}
(1) and (2) permit $t$ to be a $\treesendkeyword$ or a $\treerecvkeyword$ only when this is permitted by the session type $\sty{T}$.
(3) and (4) require $t$ to send or receive a message, when the session type $\sty{T}$ says it must do so.
We can construct $\hastype{}{}$ concretely as a union of relations; see \cref{union-pre-typing}.

The definition of $\hastype{}{}$ respects subtyping, in the sense that $\hastype{t}{\sty{T}}$ implies $\hastype{t}{\sty{U}}$ when $\sty{T} \subtype \sty{U}$.
Moreover, $t \bisim_{\sty{T}} t$ holds whenever $\hastype{t}{\sty{T}}$, where $\bisim_{\sty{T}}$ is typed bisimilarity between computation trees.
Typed bisimilarity for computation trees is non-trivial, in that $t \bisim_{\sty{T}} t'$ does not generally imply $t = t'$.
One might therefore expect us to need to consider equivalence classes of computation trees instead of just computation trees.
However, we can do better than this: up to typed bisimilarity, every typed computation tree has a normal form.
It is these normal forms we will use in our model; this avoids the need for any quotient.
The normal forms of type $\sty{T}$, with result set $X$, are the elements of the set $\treegmonad(X)_{\sty{T}}$ defined coinductively as follows.
(Concretely, we can construct $\treegmonad(X)$ as a union; see \cref{union-normal-forms} for details.)
\begin{definition}
  We write $\treegmonad(X)$ for the largest family of sets, indexed by closed session types $\sty{T}$, such that $t \in \treegmonad(X)_{\sty{T}}$ implies the following.
  \begin{enumerate}
    \item If $\unfold{\sty{T}} = \ltypeend$, then $t = \treereturn{x}$ for some $x \in X$.
    \item If $\unfold{\sty{T}} = \ltypesend{\pp}{\ell_i}{\basetype_i}{\sty{T}_i}{i \in I}$, then $t = \treesend{\pp}{\msg{\ell_i}{v}}{t'}$ for some $i \in I$, some $v : \basetype_i$, and some $t' \in \treegmonad(X)_{\sty{T}_i}$.
    \item If $\unfold{\sty{T}} = \ltyperecv{\pp}{\ell_i}{\basetype_i}{\sty{T}_i}{i \in I}$, then $t = \treerecv{\pp}{t_m}{m \in M}$ for some family $(t_m)_{m \in M}$, where $M = \{\msg{\ell_i}{v} \mid i \in I \land v : \basetype_i\}$, and $t_{\msg{\ell_i}{v}} \in \treegmonad(X)_{\sty{T}_i}$ for all $(\ell_i, v) \in M$.
  \end{enumerate}
\end{definition}

\begin{restatable}[Normalization of typed computation trees]{lemma}{bisimnormal}\label{bisim-normal}
  Let $\sty{T}$ be a closed session type.
  We have $\hastype{u}{\sty{T}}$ for every $u \in \treegmonad(X)_{\sty{T}}$.
  If $t$ is a computation tree such that $\hastype{t}{\sty{T}}$, then there is a unique $u \in \treegmonad(X)_{\sty{T}}$ such that $t \bisim_{\sty{T}} u$.
\end{restatable}

\subsection{Returning, sequencing, and subtyping}\label{sec:model-sequencing}
As noticed by Katsumata~\cite{katsumata2014parametric}, the appropriate structure for interpreting graded computational effects is a \emph{graded monad}~\cite{borceux2005internal,smirnov2008graded,mellies2012parametric,katsumata2014parametric}.
Specifically, when we give our denotational semantics in \cref{sec:semantics} below, we specify an interpretation of each of our typing rules; a graded monad provides the structure needed to interpret the typing rules \RuleName{$\subtype$}, \RuleName{Ret}, and \RuleName{Let}.

To interpret \lambdasmp{}, we therefore make normal forms of computation trees into a graded monad $\treegmonad$.
This is the appropriate graded monad for interpreting asynchronous message-passing viewed as a computational effect.
To do this, we provide three classes of functions involving computation trees.
\begin{itemize}
    \item \RuleName{$\subtype$}: To interpret \emph{subtyping}, we need a function $\treegmonad(X)_{\sty{T} \subtype \sty{U}} \colon \treegmonad(X)_{\sty{T}} \to \treegmonad(X)_{\sty{U}}$ for each set $X$ and pair of closed session types such that $\sty{T} \subtype \sty{U}$.
      We define these functions by $\treegmonad(X)_{\sty{T} \effleq \sty{U}}(t) = u$, where $u$ is the unique $u \in \treegmonad(X)_{\sty{U}}$ such that $t \bisim_{\sty{U}} u$. (This exists by \cref{bisim-normal} and the fact that $\hastype{}{}$ respects subtyping.)
    \item \RuleName{Ret}: To interpret \emph{returning a value}, we need a \emph{unit} function $\return_X \colon X \to \treegmonad(X)_{\effunit}$ for each $X$. We define $\return_X(x) = \treereturn{x}$.
    \item \RuleName{Let}: To interpret \emph{sequencing}, we need a \emph{bind} function $(\gbind{\sty{T}}{\sty{T'}}) \colon \treegmonad(X)_{\sty{T}} \times (X\to \treegmonad(Y)_{\sty{T}'}) \to \treegmonad(Y)_{\sty{T} \effmul \sty{T'}}$ for each $X, Y, \sty{T}, \sty{T}'$.
      We define $t \gbind{\sty{T}}{\sty{T'}} f$ coinductively by inspecting $t$, using the following clauses.
      \[
        (\treereturn{x}) \bind f = f\,x
        \quad
        \begin{array}{rcl}
          (\treesend{\pp}{m}{t}) \bind f &=& \treesend{\pp}{m}{t \bind f}
          \\
          (\treerecv{\pp}{t_m}{m \in M}) \bind f &=& \treerecv{\pp}{t_m \bind f}{m \in M}
        \end{array}
      \]
\end{itemize}

\section{Denotational semantics of $\lambdasmp$}\label{sec:semantics}
We now demonstrate that computation trees, and specifically the graded monad $\treegmonad$ defined in the previous section, form the basis of a model of $\lambdasmp{}$.
The aim is to interpret each $\lambdasmp$ configuration as an equivalent computation tree, where the notion of equivalence is our typed bisimulation.
This equivalence provides our \emph{adequacy} result (\cref{adequacy} below).

We first explain how to interpret $\lambdasmp{}$ values.
We define the interpretation of a ground type $\basetype$ to be the set $\sem{\basetype} = \{v \mid v : \basetype\}$ of constants of that type.
If $\Gamma$ is a typing context, then a \emph{variable environment} $\gamma$ is a function that assigns, to every $(x : \basetype) \in \Gamma$, a constant $\gamma(x) \in \sem{\basetype}$; we write $\sem{\Gamma}$ for the set of such functions.
We interpret each typed value $\valuetypedg{\Gamma}{v}{\basetype}$ as a function $\sem{v} \colon \sem{\Gamma} \to \sem{\basetype}$, as in \cref{fig:denotational-semantics}.

\begin{figure}[t!]
\begin{gather*}
  \text{\fbox{$\sem{v} \colon \sem{\Gamma} \to \sem{\basetype}$ where $\valuetypedg{\Gamma}{v}{\basetype}$}}
  \quad
  \sem{x}(\gamma) = \gamma(x)
  \quad
  \sem{v}(\gamma) = v~\text{if}~v~\text{is a constant}
\end{gather*}
\vspace{-3ex}
\begin{gather*}
  \shortintertext{\fbox{$\sem{t}_\theta \colon \sem{\Phi}_{\theta} \times \sem{\Gamma} \to \treegmonad(\sem{\basetype})_{\subst{\sty{T}}{\theta}}$ where $\comptypedg{\Theta}{\Phi}{\Gamma}{t}{\basetype}{\sty{T}}$}}
\TRULE{$\effleq$}    
    \frac
      {
        t' = \sem{t}_{\theta}(\phi, \gamma)
        \quad
        \sty{T} \subtype_{\Theta} \sty{U}
      }
      {\sem{t}_{\theta}(\phi, \gamma) = \treegmonad(\sem{\basetype})_{\subst{\sty{T}}{\theta} \subtype \subst{\sty{U}}{\theta}}(t')}
    \quad
\TRULE{Ret}        
    \frac
      {v' = \sem{v}(\gamma)}
      {\sem{\returnterm{v}}_\theta(\phi, \gamma) = \treereturn{v'}}
  \\[1ex]
\TRULE{Let}        
    \frac
      {
        t' = \sem{t}_\theta(\phi, \gamma)
        \quad
        f = \lambda x'.\,\sem{u}_\theta(\phi, (\gamma, x \mapsto x'))
      }
      {\sem{\bindterm{x}{t}{u}}_\theta(\phi, \gamma) = (t' \bind f)}
~
\TRULE{$+$}        
    \frac
      {
        v' = \sem{v}(\gamma)
        \quad
        w' = \sem{w}(\gamma)
      }
      {\sem{\addterm{v}{w}}_\theta(\phi, \gamma) = \treereturn{v' {+} w'}}
\\
\TRULE{$<$}            
    \frac
      {
        v' = \sem{v}(\gamma)
        \quad
        w' = \sem{w}(\gamma)
      }
      {\sem{\lessterm{v}{w}}_\theta(\phi, \gamma) = \treereturn{\text{if}~v' {<} w'~\text{then}\,\mathrm{true}\,\text{else}\,\mathrm{false}}}
  \\[1ex]
\TRULE{If}                
    \frac
      {
        v' = \sem{v}(\gamma)
        \quad
        t_1' = \sem{t_1}_\theta(\phi, \gamma)
        \quad
        t_2' = \sem{t_2}_\theta(\phi, \gamma)
      }
      {\sem{\ifterm{v}{t_1}{t_2}}_\theta(\phi, \gamma) = (\text{if}~v'~\text{then}~t'_1~\text{else}~t'_2)}
  \\[1ex]
\TRULE{Send}                        
    \frac
      {
        v' = \sem{v}(\gamma)
        \quad
        t' = \sem{t}_\theta(\phi, \gamma)
      }
      {\sem{\sendterm{\ell}{v}{\pp}{t}}_\theta(\phi, \gamma) = \treesend{\pp}{\msg{\ell}{v'}}{t'}}
  \\[1ex]
\TRULE{Recv}                        
    \frac
      {
        f_i = \lambda x'_i.\,\sem{t_i}_\theta(\phi, (\gamma, x_i \mapsto x'_i))
        ~\text{for all}~i \in I
      }
      {\sem{\recvterm{\pp}{\ell_i}{x_i}{t_i}{i \in I}}_\theta(\phi, \gamma) = \treerecv{\pp}{f_i(v)}{\msg{\ell_i}{v} \in \{\msg{\ell_i}{v} \mid i \in I \land v : \basetype_i\}}}
  \\[1ex]
\TRULE{LetRec}
    \frac
      {
        \begin{array}{@{}l@{}}
        f' = \lambda(x'_1, \dots, x'_n).\,\sem{t}_{\theta, \ltypevar{X} \mapsto (\ltyperec{\ltypevar{X}}{\sty{T}})}((\phi, f \mapsto f'), (\gamma, x_1 \mapsto x'_1, \dots, x_n \mapsto x'_n))
        \\
        g' = \lambda f''.\,\sem{u}_{\theta}((\phi, f \mapsto f''), \gamma)
        \end{array}
      }
      {\sem{\letrecterm{f}{(x_1, \dots, x_n)}{t}{u}}_\theta(\phi, \gamma) = g'(f')}
  \\[1ex]
  \TRULE{App}
    \frac
      {
        v'_1 = \sem{v_1}(\gamma)
        ~\cdots~
        v'_n = \sem{v_n}(\gamma)
      }
      {\sem{f(v_1, \dots, v_n)}_{\theta}(\phi, \gamma) = \phi(f)(v'_1, \dots, v'_n)}
\end{gather*}
\vspace{-3ex}
\begin{gather*}
  \shortintertext{\fbox{$\sem{\config{\queue{\rho}}{t}{\queue{\sigma}}} \in \treegmonad(\basetype)_{\sty{T}}$ where $\conftypedg{\config{\queue{\rho}}{t}{\queue{\sigma}}}{\basetype}{\sty{T}}$}}
  \begin{array}{l}
    \TRULE{CBase}
    \\
    \displaystyle\frac
      {t' = \sem{t}_\emptycontext(\emptycontext, \emptycontext)}
      {\sem{\config{\queueempty}{t}{\queueempty}} = t'}
  \end{array}
    \quad
  \begin{array}{l}
    \TRULE{CSend}
    \\
    \displaystyle\frac
      {
        \begin{array}{l}
        \reduceinternal{\sty{U}}{\sty{T}}{\pp}{\ell}{\basetype'}
        \\[-1ex]
        \reduceaction{\treegmonad(\basetype)_{\sty{U}} \ni u}{\sem{\config{\queue{\rho}}{t}{\queue{\sigma}}}}{\sendaction{\pp}{\msg{\ell}{v}}}
        \end{array}
      }
      {\sem{\config{\queue{\rho}}{t}{\consback{\queue{\sigma}}{\pp}{\msg{\ell}{v}}}} = u}
  \end{array}
    \quad
  \begin{array}{l}
    \TRULE{CRecv}
    \\
    \displaystyle\frac
      {
        \begin{array}{l}
        \reduceexternal{\sty{T}}{\sty{U}}{\pp}{\ell}{\basetype'}
        \\[-1ex]
        \reduceaction{\sem{\config{\queue{\rho}}{t}{\queue{\sigma}}}}{u \in \treegmonad(\basetype)_{\sty{U}}}{\recvaction{\pp}{\msg{\ell}{v}}}
        \end{array}
      }
      {\sem{\config{\consfront{\pp}{\msg{\ell}{v}}{\queue{\rho}}}{t}{\queue{\sigma}}} = u}
  \end{array}
\end{gather*}
\caption{Interpretation of $\lambdasmp$ values, computations and configurations \label{fig:denotational-semantics}}
\end{figure}

For computations, we interpret the judgement $\comptypedg{\Theta}{\Phi}{\Gamma}{t}{\basetype}{\sty{T}}$ as follows.
A \emph{session type environment} $\theta$, for the set $\Theta$ of type variables, is a function from $\Theta$ to closed session types.
Given such a $\theta$, if $\sty{T}$ is a session type over $\Theta$, then we write $\subst{\sty{T}}{\theta}$ for the session type that results from substituting the free type variables according to $\theta$.
For a function context $\Phi$ over $\Theta$, a \emph{function environment} $\phi$ is a function that assigns, to every $(f : (\basetype_1, \dots, \basetype_n) \gto{\sty{T}} \basetype') \in \Phi$, a function $\phi(f) \colon \sem{\basetype_1} \times \cdots \times \sem{\basetype_n} \to \treegmonad(\sem{\basetype'})_{\subst{\sty{T}}{\theta}}$; we write $\sem{\Phi}_\theta$ for the set of such function environments.
For computations, we interpret each derivation of $\comptypedg{\Theta}{\Phi}{\Gamma}{t}{\basetype}{\sty{T}}$ as a function $\sem{t}_\theta \colon \sem{\Phi}_\theta \times \sem{\Gamma} \to \treegmonad(\sem{\basetype})_{\subst{\sty{T}}{\theta}}$.
This is defined in \cref{fig:denotational-semantics}.
The definition is by induction on the \emph{derivation}, not on the computatation $t$, and hence a priori the computation $t$ will have several interpretations -- one for each derivation.
It turns out this is not the case; we show below that any two derivations of $\comptypedg{\Theta}{\Phi}{\Gamma}{t}{\basetype}{\sty{T}}$ necessarily have the same interpretation.
The interpretations of subtyping, $\returntermkeyword$ and sequencing use the graded monad structure.
The interpretations of $\sendtermkeyword$ and $\recvtermkeyword$ are the obvious computation trees.

Finally, for configurations, we interpret every derivation of $\conftypedg{\config{\queue{\rho}}{t}{\queue{\sigma}}}{\basetype}{\sty{T}}$ as a computation tree $\sem{\config{\queue{\rho}}{t}{\queue{\sigma}}} \in \treegmonad(\sem{\basetype})_{\sty{T}}$.
Rule \TRULE{CBase} uses the interpretation of computations.
Rule \TRULE{CSend} interprets a configuration with a message to send as a computation tree $u$ that reduces by sending such a message; since we consider normal forms of computation trees, there is a unique such $u$.
Rule \TRULE{CRecv} interprets a configuration that has received a message by reducing $\sem{\config{\queue{\rho}}{t}{\queue{\sigma}}}$; again, since we are dealing with normal forms, there is a unique such reduct.

Our main result about our denotational semantics is the following theorem, which
establishes that computation trees correctly interpret configurations.
It is this result that we appeal to when using our semantics to reason about $\lambdasmp{}$.

\begin{theorem}[Correctness of the denotational semantics]\label{semantics-correct}
  We have $\config{\queue{\rho}}{t}{\queue{\sigma}} \bisim_{\sty{T}} \sem{\config{\queue{\rho}}{t}{\queue{\sigma}}}$
  for every derivation of $\conftypedg{\config{\queue{\rho}}{t}{\queue{\sigma}}}{\basetype}{\sty{T}}$.
\end{theorem}

A corollary is that the denotational semantics is sound and complete with respect to typed bisimilarity; this is \emph{adequacy} of the denotational semantics.
\begin{corollary}[Adequacy]\label{adequacy}
  Let $\conftypedg{\config{\queue{\rho}}{t}{\queue{\sigma}}}{\basetype}{\sty{T}}$ and $\conftypedg{\config{\queue{\rho'}}{t'}{\queue{\sigma'}}}{\basetype}{\sty{T}}$ be two configurations of the same type.
  We have $\config{\queue{\rho}}{t}{\queue{\sigma}} \bisim_{\sty{T}} \config{\queue{\rho'}}{t'}{\queue{\sigma'}}$
  if and only if $\sem{\config{\queue{\rho}}{t}{\queue{\sigma}}} = \sem{\config{\queue{\rho'}}{t'}{\queue{\sigma'}}}$.
\end{corollary}
\begin{proof}
  By \cref{semantics-correct} and transitivity of $\bisim_{\sty{T}}$, we have
  $\config{\queue{\rho}}{t}{\queue{\sigma}} \bisim_{\sty{T}} \config{\queue{\rho'}}{t'}{\queue{\sigma'}}$
  iff
  $\sem{\config{\queue{\rho}}{t}{\queue{\sigma}}} \bisim_{\sty{T}} \sem{\config{\queue{\rho'}}{t'}{\queue{\sigma'}}}$.
  The latter bisimulation is an equality by \cref{bisim-normal}.
\end{proof}
Adequacy also establishes that the interpretation of a configuration does not depend on the choice of type derivation.

\begin{example}\label{ex:global-state-example-semantics}
  We return to our global state example (\cref{global-state-computations}), and specifically to the bisimilarity $\config{\queueempty}{t_{\role{c}, 1}}{\queueempty} \bisim_{\sty{S}}\config{\queueempty}{t_{\role{c}, 2}}{\queueempty}$ claimed in \cref{ex:global-state-example-bisim}.
  Due to adequacy, this bisimilarity is \emph{equivalent} to the two configurations having the same interpretation in our denotational semantics.
  Indeed they do have the same interpretation; their interpretation is the following element of $\treegmonad(\sem{\inttype})_{\sty{S}}$.
  \[
    \treesend{\role{s}}{\msg{\ml{get}}{\star}}{\treerecv{\role{s}}{\treesend{\role{s}}{\msg{\ml{put}}{0}}{\treereturn{n}}}{\msg{\ml{st}}{n}}}
  \]
\end{example}
\section{Deadlock-freedom and liveness}\label{sec:session-calculus}
A \emph{session} is a collection of participants running in parallel.
Deadlock-freedom and liveness are properties of sessions.
In this section, we define a notion of $\lambdasmp{}$ session, and then prove our desired safety and liveness properties.
We do this as an application of our denotational semantics; we use the computation-tree interpretation of $\lambdasmp{}$ to prove these properties.

\begin{definition}
  A \emph{session} $\mathcal{M}$ is a finite list $(\namepart{\pr_1}{\mathcal{C}_1}, \namepart{\pr_2}{\mathcal{C}_2}, \dots, \namepart{\pr_n}{\mathcal{C}_n})$, where $\pr_1, \dots, \pr_n$ are distinct participant names, and each $\mathcal{C}_i$ is a configuration involving (at most) the participants in $\{\pr_1, \dots, \pr_n\} \setminus \{\pr_i\}$.
\end{definition}

We use our asynchronous operational semantics of \lambdasmp{} to define a notion of reduction $\reduceaction{\mathcal{M}}{\mathcal{M}'}{\beta}$ for sessions.
This notion of reduction in particular includes a rule for one participant sending a message to another.
A \emph{global action} $\beta$ is either $\tauaction_{\pp}$ (denoting an internal action made by participant $\pp$), or a triple $(\commaction{\pp}{\pq}{m})$ with $\pp \neq \pq$ (denoting that $\pp$ sends message $m$ to $\pq$).
Reduction $\reduceaction{\mathcal{M}}{\mathcal{M'}}{\beta}$ of sessions is defined, using reduction of configurations, by two rules:
\begin{gather*}
  \begin{prooftree}
    \hypo{\mathcal{C}_i = \mathcal{D}_i~\text{for all}~i \neq j}
    \hypo{\reduceaction{\mathcal{C}_j}{\mathcal{D}_j}{\tauaction}}
    \infer2{\reduceaction{(\namepart{\pr_1}{\mathcal{C}_1}, \namepart{\pr_2}{\mathcal{C}_2}, \dots, \namepart{\pr_n}{\mathcal{C}_n})}{(\namepart{\pr_1}{\mathcal{D}_1}, \namepart{\pr_2}{\mathcal{D}_2}, \dots, \namepart{\pr_n}{\mathcal{D}_n})}{\tauaction_{\pr_j}}}
  \end{prooftree}
  \\
  \begin{prooftree}
    \hypo{\mathcal{C}_i = \mathcal{D}_i~\text{for all}~i \not\in\{j, k\}}
    \hypo{\reduceaction{\mathcal{C}_j}{\mathcal{D}_j}{\sendaction{\pr_k}{m}}}
    \hypo{\reduceaction{\mathcal{C}_k}{\mathcal{D}_k}{\recvaction{\pr_j}{m}}}
    \infer3{\reduceaction{(\namepart{\pr_1}{\mathcal{C}_1}, \namepart{\pr_2}{\mathcal{C}_2}, \dots, \namepart{\pr_n}{\mathcal{C}_n})}{(\namepart{\pr_1}{\mathcal{D}_1}, \namepart{\pr_2}{\mathcal{D}_2}, \dots, \namepart{\pr_n}{\mathcal{D}_n})}{\commaction{\pr_j}{\pr_k}{m}}}
  \end{prooftree}
\end{gather*}

Our type system for \lambdasmp{} assigns session types $\sty{T}_i$ to the individual configurations $\mathcal{C}_i$ of a session.
By itself, this does not ensure that the $\sty{T}_i$ are in any way compatible with each other.
We ensure the latter by following the \emph{top-down} approach for MPST~\cite{HYC2016,VeryGentle,SynchronousSubtyping}, in which there is a \emph{global protocol}, and each participant $\pp$ is required to follow the local \emph{projection} of that protocol onto $\pp$.
Global protocols are described by \emph{global types}, which are generated inductively by the following grammar.
\[
  \gty{G}~\Coloneq~
  \gtypeend
  ~|~\gtypecomm{\pp}{\pq}{\ell_i}{\basetype_i}{\gty{G}_i}{i \in I}
  ~|~\gtypevar{X}
  ~|~\gtyperec{\gtypevar{X}}{\gty{G}}
\]
The global type $\gtypeend$ denotes that no further communication between participants will happen. $\gtypecomm{\pp}{\pq}{\ell_i}{\basetype_i}{\gty{G}_i}{i \in I}$ denotes that $\pp$ sends a single message to $\pq$; that message will have the form $\msg{\ell_i}{v}$ with $v : \basetype_i$ for some $i \in I$, and then the protocol will continue as $\sty{G}_i$.
As for internal and external choices, we require $I$ to be non-empty and finite, and we require the labels $\ell_i$ to be distinct from each other.
We also require $\pp \neq \pq$.
A \emph{recursive protocol} $\gtyperec{\gtypevar{X}}{\gty{G}}$ binds the type variable $\gtypevar{X}$; we require that every occurence of $\gtypevar{X}$ is under some communication $\pp \to \pq$.
Just as for local types, we define a notion of single-step unfolding for global types:
\[
  \unfold{\gtyperec{\gtypevar{X}}{\gty{G}}} = \subst{\unfold{\gty{G}}}{\gtypevar{X} \mapsto \gtyperec{\gtypevar{X}}{\gty{G}}}
  \qquad
  \unfold{\gty{G}} = \gty{G}
  ~\text{if}~\gty{G}~\text{is not a recursive type}
\]
A global type $\gty{G}$ is \emph{closed} when it has no free type variables.

In the top-down approach, we determine each participant's (local) session type by \emph{projecting} it from the global type $\gty{G}$.
Projection is a partial function that maps a global type $\gty{G}$ and participant $\pr$ to a session type $\project{G}{\pr}$.
The definition is by recursion on $\gty{G}$, and is given in \cref{fig:projection}.
In the case of a communication not involving $\pr$, we \emph{merge} the projections of the branches.
Merging is a partial binary operation $\sty{T} \bmergeltype \sty{T}'$ on session types.
We use \emph{full merging}, as defined in \cite{Scalas2019}.
The case split in the projection from a recursive global type ensures guardedness.
\begin{figure}[t]
\begin{gather*}
  \begin{array}{l}
  \project{\gtypeend}{\pr}
  =
  \ltypeend
  \qquad
  \project{(\gtypecomm{\pp}{\pq}{\ell_i}{\basetype_i}{\gty{G}_i}{i \in I})}{\pr}
  =
  \begin{cases}
    \ltypesend{\pq}{\ell_i}{\basetype_i}{(\project{\gty{G}_i}{\pr})}{i \in I}
    &\text{if}~
    \pp = \pr
    \\
    \ltyperecv{\pp}{\ell_i}{\basetype_i}{(\project{\gty{G}_i}{\pr})}{i \in I}
    &\text{if}~
    \pq = \pr
    \\
    \mergeltype{(\project{\gty{G}_i}{\pr})}{i \in I}
    &\text{if}~\pr \not\in \{\pp, \pq\}
  \end{cases}
  \end{array}
  \\
  \project{\gtypevar{X}}{\pr}
  =
  \ltypevar{X}
  \qquad
  \project{(\gtyperec{\gtypevar{X}}{\gty{G}})}{\pr}
  =
  \begin{cases}
  \ltypeend
  &\!\!\text{if}~\project{\gty{G}}{\pr} = \ltypevar{X}
  \\
  \ltypevar{X'}
  &\!\!\text{if}~\project{\gty{G}}{\pr} = \ltypevar{X'} \neq \ltypevar{X}
  \\
  \ltyperec{\ltypevar{X}}{(\project{\gty{G}}{\pr})}
  &
  \!\!\text{otherwise}
  \end{cases}
\end{gather*}

\begin{gather*}
  \ltypeend \bmergeltype \ltypeend = \ltypeend
  \quad
  \begin{array}{@{}l@{}}
    (\ltypesend{p}{\ell_i}{\basetype_i}{\sty{T}_i}{i \in I})
    \bmergeltype
    (\ltypesend{p}{\ell_i}{\basetype_i}{\sty{T}'_i}{i \in I})
  \end{array}
  =
  \begin{array}{@{}l@{}}
    (\ltypesend{p}{\ell_i}{\basetype_i}{\sty{T}_i \bmergeltype \sty{T}'_i}{i \in I})
  \end{array}
  \\
  \begin{array}{@{}l@{}}
    (\ltyperecv{p}{\ell_i}{\basetype_i}{\sty{T}_i}{i \in I})
    \\~\bmergeltype~
    (\ltyperecv{p}{\ell_i}{\basetype_i}{\sty{T}'_i}{i \in J})
  \end{array}
  =
  \begin{array}{@{}l@{}}
    (\ltyperecv{p}{\ell_i}{\basetype_i}{\sty{T}_i \bmergeltype \sty{T}'_i}{i \in I
    \cap J})
    \\\&~
    (\ltyperecv{p}{\ell_i}{\basetype_i}{\sty{T}_i}{i \in I \setminus J})
    ~\&~
    (\ltyperecv{p}{\ell_i}{\basetype_i}{\sty{T}'_i}{i \in J \setminus I})
  \end{array}
  \\
  \begin{array}{c}
    \ltypevar{X} \bmergeltype \ltypevar{X} = \ltypevar{X}
    \quad
    (\ltyperec{\ltypevar{X}}{\sty{T}})
    \bmergeltype
    (\ltyperec{\ltypevar{X}}{\sty{T}'} )
    =
    \ltyperec{\ltypevar{X}}{(\sty{T} \bmergeltype \sty{T}')} 
  \end{array}
\end{gather*}
\caption{Definitions of \emph{projection} and of \emph{full merging} of multiparty session types\label{fig:projection}}
\end{figure}

\begin{example}
  In the context of \cref{ex:gppsy-example}, the computation $t$ implements a participant $\pr$ as part of a global protocol described by the following global type $\gty{G}$; we can associate to $t$ the session type $\project{\gty{G}}{\pr}$.
  \begin{gather*}
    \scalebox{0.94}{$
    ~\gty{G} = \gtypecommcases{\pp}{\pr}{
      \mltype{\ml{success}}{\inttype}.\gtypecommcases{\pr\!}{\!\pq}{\mltype{\ml{cont}}{\inttype}.\, \gtypeend \\ \mltype{\ml{stop}}{\booltype}.\, \gtypeend}
      \\
      \mltype{\ml{error}}{\booltype}.\,\gtypecommone{\pr\!}{\!\pq}{\ml{stop}}{\booltype}{\gtypeend}
    }
    \mspace{-22mu}
      \project{\gty{G}\!}{\!\pp} = \ltypesendcases{\pr}{
      \mltype{\ml{success}}{\inttype}.\ltypeend
      \\
      \mltype{\ml{error}}{\booltype}.\,\ltypeend
      }\mspace{-12mu}
    $}
    \\
    \scalebox{0.9}{$
      \project{\gty{G}}{\pq} = \ltyperecvcases{\pr}{\mltype{\ml{cont}}{\inttype}.\, \ltypeend \\ \mltype{\ml{stop}}{\booltype}.\, \ltypeend}
      \mspace{-2mu}
      \project{\gty{G}}{\pr} =
      \ltyperecvcases{\pp}{
      \mltype{\ml{success}}{\inttype}.\ltypesendcases{\pq}{\mltype{\ml{cont}}{\inttype}.\, \ltypeend \\ \mltype{\ml{stop}}{\booltype}.\, \ltypeend}
      \\
      \mltype{\ml{error}}{\booltype}.\ltypesendone{\pq}{\ml{stop}}{\booltype}{\ltypeend}}
    $}
  \end{gather*}
\end{example}

\begin{definition}
  A session $(\namepart{\pr_1}{\mathcal{C}_1}, \namepart{\pr_2}{\mathcal{C}_2}, \dots, \namepart{\pr_n}{\mathcal{C}_n})$ \emph{has type $\gty{G}$} if (1) $\gty{G}$ is closed and contains only the participants in $\{\pr_1, \dots, \pr_n\}$, (2) the projections $\project{\gty{G}}{\pr_i}$ are all defined, and (3) $\conftypedg{\mathcal{C}_i}{\basetype_i}{\project{\gty{G}}{\pr_i}}$ for each $i$.
  When this is the case, we say that the session is \emph{well-typed}.
\end{definition}

\begin{example}
  We define a global type $\gty{G}$ for our global state example.
  The projections are the session types from \cref{ex:global-state-example-graded}.
  \[
  \scalebox{0.95}{$
    \gty{G} = \gtyperec{\gtypevar{X}}{\gtypecommcases{\role{c}}{\role{s}}{
      \mltype{\ml{get}}{\unittype}.\,\gtypecommone{\role{s}}{\role{c}}{\ml{st}}{\inttype}{\gtypevar{X}}
      \\
      \mltype{\ml{put}}{\inttype}.\,\gtypevar{X}
      \\
      \mltype{\ml{done}}{\unittype}.\,\gtypeend
    }}$}
    \qquad
    \begin{array}{r@{~}c@{~}l}
      \project{\gty{G}}{\role{s}} &=& \ltyperec{\ltypevar{X}}{\sty{T}_{\role{s}}}
      \\
      \project{\gty{G}}{\role{c}} &=& \ltyperec{\ltypevar{X}}{\sty{T}_{\role{c}}}
    \end{array}
  \]
  The configurations of \cref{ex:global-state-example-graded} form a session $\mathcal{M} = (\namepart{\role{s}}{\mathcal{C}_{\role{s}}}, \namepart{\role{c}}{\mathcal{C}_{\role{c},2}})$ of type $\gty{G}$.
\end{example}

Subject reduction for configurations provides a similar theorem for sessions.
\begin{restatable}[Subject reduction for sessions]{theorem}{sessionsubjred}\label{session-subject}
  If $\mathcal{M}$ is well-typed and $\reduceaction{\mathcal{M}}{\mathcal{M}'}{\beta}$, then $\mathcal{M}'$ is also well-typed.
\end{restatable}

We now come to our desired safety and liveness properties, which hold for all well-typed sessions.
The proofs of these use our denotational semantics.
\emph{Deadlock-freedom} means that reduction cannot get stuck; either the session terminates (with every configuration returning a result), or some communication will happen.
As stated this theorem only applies to the initial session $\mathcal{M}_0$, but it follows from \cref{session-subject} that deadlock-freedom also applies to any reduct of $\mathcal{M}_0$.
\begin{theorem}[Deadlock-freedom]
  If $\mathcal{M}_0$ is well-typed, then there is a reduction sequence
  $
        \reduceaction{\mathcal{M}_0}{\reduceaction{\cdots}{\mathcal{M}_c}{\beta_c}}{\beta_1}
  $, with $c \geq 0$, such that either (1) $c \geq 1$ and $\beta_c$ has the form $\commaction{\pp}{\pq}{m}$,
  or (2) $\mathcal{M}_c$ has the form
  \[
    (\namepart{\pr_1}{\config{\queue{\queueempty}}{\rctx{R}_1[\returnterm{v_1}]}{\queue{\queueempty}}}, \dots, \namepart{\pr_n}{\config{\queue{\queueempty}}{\rctx{R}_n[\returnterm{v_n}]}{\queue{\queueempty}}})
  \]
\end{theorem}
\begin{proof}
  Let $\gty{G}$ be the type of $\mathcal{M}_0 = (\namepart{\pr_1}{\mathcal{C}_1}, \dots, \namepart{\pr_n}{\mathcal{C}_n})$.
  If $\unfold{\gty{G}} = \gtypeend$, then for every $i$ we have $\unfold{\project{\gty{G}}{\pr_i}} = \ltypeend$, so the computation tree $\sem{\mathcal{C}_i}$ has the form $\treereturn{x_i}$, and by \cref{semantics-correct} there is a reduction $\reduceactionmany{\mathcal{C}_i}{\config{\queue{\queueempty}}{\rctx{R}_i[\returnterm{v_i}]}{\queue{\queueempty}}}{\tauaction}$.
  Otherwise, $\unfold{\gty{G}}$ has the form $\gtypecomm{\pr_j}{\pr_k}{\ell_i}{\basetype_i}{\gty{G}_i}{i \in I}$.
  In this case, $\unfold{\project{\gty{G}}{\pr_j}}$ has the form $\ltypesend{\pr_k}{\ell_i}{\basetype_i}{\sty{T}_i}{i \in I}$ and $\unfold{\project{\gty{G}}{\pr_k}}$ has the form $\ltyperecv{\pr_j}{\ell_i}{\basetype_i}{\sty{U}_i}{i \in I}$.
  It follows from \cref{semantics-correct} that there are reductions $\reduceactionmany{\mathcal{C}_j}{\mathcal{D}_j}{\sendaction{\pr_k}{m}}$ and $\reduceactionmany{\mathcal{C}_k}{\mathcal{D}_k}{\recvaction{\pr_j}{m}}$, so that we can take $\beta_c = \commaction{\pr_j}{\pr_k}{m}$.
\end{proof}

Finally, liveness means that (1) if a configuration is waiting to receive a message, then it will eventually do so; and (2) if a configuration sends a message, then that message will eventually be consumed.
This assumes \emph{fair} scheduling (the scheduler does not starve any participant).
Our formulation of liveness is analogous to that of \cite{GPPSY2023}, and is stronger than the liveness considered in \cite{Scalas2019} (cf.\ \cite[Footnote~4]{GPPSY2023}).
\begin{restatable}[Liveness]{theorem}{liveness}
  Let
  $
    \reduceaction{\mathcal{M}_0}{\reduceaction{\mathcal{M}_1}{\cdots}{\beta_2}}{\beta_1}
  $
  be a (possibly infinite) reduction sequence, with
  $\mathcal{M}_i = (\namepart{\pr_1}{\mathcal{C}_{i1}}, \dots, \namepart{\pr_n}{\mathcal{C}_{in}})$, and
  $\mathcal{C}_{ij} = \config{\rho_{ij}}{t_{ij}}{\sigma_{ij}}$.
  Assume that the reduction sequence is \emph{fair}, meaning for every $i$ such that there exists a reduction $\reduceaction{\mathcal{M}_i}{\mathcal{M}'}{\beta'}$, there is some $i' > i$ such that $\beta_{i'} = \beta'$.
  If $\mathcal{M}_0$ is well-typed, then we have the following.
  \begin{enumerate}
    \item If $t_{ij}$ has the form $\rctx{R}[\recvterm{\pp}{\ell_k}{x_k}{u_k}{k \in K}]$, and $\rho_{ij}$ has no messages from $\pp$, then there is some $i' > i$ and $m \in M$ such that $\beta_{i'} = (\commaction{\pp}{\pr_j}{m})$.
    \item If $\beta_i = (\commaction{\pp}{\pr_j}{m})$, then there is some $i' > i$ such that $t_{i'j}$ has the form $\rctx{R}[\recvterm{\pp}{\ell_k}{x_k}{u_k}{k \in K}]$.
  \end{enumerate}
\end{restatable}
The idea behind the proof of liveness is that, if the global type $\gty{G}$ contains a communication $\pp \to \pq$ between two participants of the session, then at some point during the execution a $\pp \to \pq$ transition becomes available. We prove the latter by appealing to \cref{semantics-correct}, using the fact that such a transition becomes available in the model.
Fairness implies that the transition will be taken at some point.
For both cases of liveness, the required communication appears at some finite depth in $\gty{G}$, and thus happens somewhere along the reduction sequence.

\section{Related work}\label{sec:related-work}
\paragraph{Semantics of session types}
Originating from \cite{Honda1993,Honda1998},  
strong foundations of session types have been developed based on 
linear logic, exploiting a Curry-Howard correspondence between 
linear logic and typed session $\pi$-calculi \cite{Toninho13,CairesP10,DBLP:conf/icfp/Wadler12}. 
The most 
effective and advanced semantics of these use
\emph{logical relations} to tackle
various programming language features including 
parametricity \cite{CairesPPT13} and
higher-order functions \cite{Toninho13}.
Among them,
\cite{vandenheuvel_et_al:LIPIcs.ECOOP.2024.40},
which is built on \cite{balzer2023logicalrelationssessiontypedconcurrency}, 
enables tracking of cyclic dependencies of channels 
via classical logic session types, and applies this 
to information flow analysis. The work
in \cite{yao2024semanticlogicalrelationstimed}  
develops logical relations based on 
instuitionistic linear logic enriched with temporal predicates.
The logical relations works
are limited to \emph{binary} session types;
we do multiparty. 
Jacobs et al.~\cite{JacobsBK22a} introduce 
MPGV, a functional MPST language with multiparty
session types, based on the linear logic perspective. They use
separation logic to
define configuration invariants 
to maintain the acyclic nature of the communication topology and 
to establish subject reduction.
They support \emph{session delegation}, which we
leave to future work, but they do not
develop a (denotational) semantics for their calculus.
Castellani et al.~\cite{CastellaniDG23b,CastellaniDG24}
interpret asynchronous and synchronous multiparty sessions 
as \emph{flow event structures}.
They do not have asynchronous subtyping; our work is the first denotational semantics for MPST with asynchronous subtyping.
They also do not consider standard programming constructs 
such as sequencing, unlike us.
Moreover, their work focuses on interpreting \emph{sessions} and global protocols.
As such, they do not interpret (local) computations in the manner that we do, so their semantics cannot be used to reason about individual participants in isolation.
In a separate line of work,
Castellan and Yoshida~\cite{CY2019} use
a connection 
between linear logics and game semantics
to describe a fully abstract game semantics 
for binary session typed processes. They leave the extension to 
asynchrony and multiparty open. 
The aforementioned works do not study MPST from the perspective of computational effects.

\paragraph{Effects and session types}
Orchard and Yoshida~\cite{OY2016} show that one can encode a binary session-typed
$\pi$-calculus in a graded variant of PCF, session types being encoded
as grades, and vice-versa.
This work is orthogonal to ours.
Their graded PCF does not have message-passing, and they do not consider message-passing as a computational effect.
Thus they do not show how to track message-passing in an effectful calculus, nor do they provide a denotational semantics for session types.
Instead their motivation is about encodability results, and they do not study 
safety, deadlock-freedom and liveness properties.
Their work also considers only \emph{binary} session types (not multiparty), and does not consider asynchrony.
There are few other works that approach message-passing from the
perspective of computational effects.
Sanada~\cite{sanada2023category} uses message-passing to exemplify \emph{category-graded effect handlers}, but does not give a denotational semantics.
Marshall and Orchard~\cite{MARSHALL2024105234} take the linear logic
perspective on 
synchronous binary session types, and use grades to track linearity.
They observe that communication primitives are effects, but
do not use grades to track them directly,
and leave deadlock-freedom open (cf.\ \cite[Section~11]{MARSHALL2024105234}).

\paragraph{Asynchronous subtyping}
Recent work on asynchronous subtyping includes the undecidablity result of~\cite{bravetti2017undecidability},
works focused on taming this undecidablity~\cite{bravetti2018,BocchiK0025}, and mechanization~\cite{EY2024}.
Our formulation of asynchronous subtyping, by characterising when the protocol requires or permits a message to be sent or received, is entirely different to the formulations appearing in these works, and yet is equivalent to the sound and complete subtyping of \cite{GPPSY2023}.
Compared to \cite{GPPSY2023}, our formulation avoids any use of session \emph{trees}.
Session trees are infinite structures, and thus cause some difficulties when it comes to implementing subtyping~\cite{EY2024}; our reformulation shows that we do not need to involve session trees.
Moreover, session-tree unfolding is defined only for closed session types, and hence \cite{GPPSY2023} define subtyping only for closed session types.
By avoiding session trees, we are able to provide the first extension of this subtyping to non-closed session types.

\paragraph{Typed bisimulations}
Kouzapas et al.~study a bisimulation method which characterises
a typed contextual equality in a higher-order binary session $\pi$-calculus
\cite{KPY2017}, and applied it to measure the expressiveness of higher-order session processes \cite{kouzapas2019relative}. 
Kouzapas and Yoshida~\cite{kouzapas2014globally}
propose a typed bisimulation,
controlled by declared global types,
for a synchronous multiparty session $\pi$-calculus.
We use typed bisimulations to justify the correctness of
our denotational semantics.

\paragraph{Models of computational effects}
Kavvos~\cite{kavvos2025adequacy}, improving on some earlier work by Plotkin and Power~\cite{plotkin2001adequacy}, gives a general adequacy result for computational effects.
Our adequacy result (\cref{adequacy}) may follow from a graded variant of Kavvos's result, but there is none in the literature yet.
There are various monadic models of shared-state concurrency~\cite{benton2016effect,dvir2024denotational,rivas2024concurrent,dvir2025two}, as opposed to message-passing;
these bear little resemblance to our model.

\section{Conclusions}
This work is the first to provide a formal mathematical model for reasoning about asynchronous message-passing computation.
We show that every multiparty session type $\sty{T}$ can be interpreted a set of \emph{computation trees}, and thus that computation trees provide a basis for reasoning about asynchrony, even without message queues.
Computation trees provide an adequate denotational semantics for a simple call-by-value programing language with message-passing as a computational effect, namely \lambdasmp{}.
Since it is based on well-studied tools for studying computational effects, we expect that we can add more programming features to \lambdasmp{} without much difficulty.
We also hope our computational-effects perspective on session types will enable the application of more of the vast computational effects literature to session types.
Asynchronous session subtyping is a particular focus of our work.
It is known to be difficult to reason about, but we have found our reformulation to be helpful with such reasoning.

\begin{credits}
\subsubsection{\ackname} We thank Jessica Richards, and the anonymous reviewers, for helpful comments. This work was supported by EPSRC EP/T006544/2, EP/T014709/2, EP/Z533749/1, ARIA and Horizon EU TaRDIS 101093006 (UKRI number 10066667).

\end{credits}

\bibliographystyle{plain}
\bibliography{references}

@inproceedings{VeryGentle,
  author = {Yoshida, Nobuko and Gheri, Lorenzo},
  booktitle = {Distributed Computing and Internet Technology},
  editor = {Hung, Dang Van and D{\textasciiacute}Souza, Meenakshi},
  isbn = {978-3-030-36987-3},
  pages = {73--93},
  publisher = {Springer International Publishing},
  title = {A Very Gentle Introduction to Multiparty Session Types},
  year = {2020}
}

@article{SynchronousSubtyping,
  author = {Silvia Ghilezan and Svetlana Jak{\v s}i{\'c} and Jovanka Pantovi{\'c} and Alceste Scalas and Nobuko Yoshida},
  doi = {https://doi.org/10.1016/j.jlamp.2018.12.002},
  issn = {2352-2208},
  journal = {Journal of Logical and Algebraic Methods in Programming},
  pages = {127-173},
  title = {Precise subtyping for synchronous multiparty sessions},
  volume = {104},
  year = {2019},
}

@inproceedings{OY2016,
  author = "Dominic Orchard and Nobuko Yoshida",
  title = "Effects as sessions, sessions as effects",
  booktitle = "POPL 2016",
  publisher = "ACM",
  year = "2016"
}

@misc{balzer2023logicalrelationssessiontypedconcurrency,
  title={Logical Relations for Session-Typed Concurrency}, 
  author={Stephanie Balzer and Farzaneh Derakhshan and Robert Harper and Yue Yao},
  year={2023},
  eprint={2309.00192},
  archivePrefix={arXiv},
  primaryClass={cs.PL},
}

@inproceedings{CairesP10,
  author    = {Lu\'{\i}s Caires and
               Frank Pfenning},
  title     = {Session Types as Intuitionistic Linear Propositions},
  booktitle = {Proceedings of CONCUR 2010},
  year      = {2010},
  pages     = {222-236},
  publisher = {Springer},
  series    = {LNCS},
  volume    = {6269},
}

@inproceedings{DBLP:conf/icfp/Wadler12,
  author    = {Philip Wadler},
  title     = {Propositions as sessions},
  booktitle = {Proceedings of {ICFP} 2012},
  pages     = {273--286},
  year      = {2012},
  publisher = {ACM},
  doi       = {10.1145/2364527.2364568},
}

@inproceedings{Honda1993,
  title = {Types for Dyadic Interaction},
  booktitle = {{{CONCUR}}'93},
  author = {Honda, Kohei},
  editor = {Best, Eike},
  year = {1993},
  pages = {509--523},
  publisher = {Springer},
  doi = {10.1007/3-540-57208-2_35},
  isbn = {978-3-540-47968-0},
  langid = {english},
}

@inproceedings{Honda1998,
  title = {Language {{Primitives}} and {{Type Discipline}} for {{Structured Communication-Based Programming}}},
  booktitle = {Programming {{Languages}} and {{Systems}} - {{ESOP}}'98, 7th {{European Symposium}} on {{Programming}}, {{Held}} as {{Part}} of the {{European Joint Conferences}} on the {{Theory}} and {{Practice}} of {{Software}}, {{ETAPS}}'98, {{Lisbon}}, {{Portugal}}, {{March}} 28 - {{April}} 4, 1998, {{Proceedings}}},
  author = {Honda, Kohei and Vasconcelos, Vasco Thudichum and Kubo, Makoto},
  editor = {Hankin, Chris},
  year = {1998},
  series = {Lecture {{Notes}} in {{Computer Science}}},
  volume = {1381},
  pages = {122--138},
  publisher = {Springer},
  doi = {10.1007/BFB0053567},
}

@article{HYC2016,
  author    = {Kohei Honda and
               Nobuko Yoshida and
               Marco Carbone},
  title     = {Multiparty Asynchronous Session Types},
  journal   = {Journal of the {ACM}},
  volume    = {63},
  number    = {1},
  pages     = {9:1--9:67},
  year      = {2016},
  doi       = {10.1145/2827695},
}

@article{Scalas2019,
  title = {Less Is More: Multiparty Session Types Revisited},
  shorttitle = {Less Is More},
  author = {Scalas, Alceste and Yoshida, Nobuko},
  year = {2019},
  month = jan,
  journal = {Proceedings of the ACM on Programming Languages},
  volume = {3},
  number = {POPL},
  pages = {1--29},
  issn = {2475-1421},
  doi = {10.1145/3290343},
}

@InProceedings{Toninho13,
author="Toninho, Bernardo
and Caires, Luis
and Pfenning, Frank",
editor="Felleisen, Matthias
and Gardner, Philippa",
title="Higher-Order Processes, Functions, and Sessions: A Monadic Integration",
booktitle="Programming Languages and Systems",
year="2013",
publisher="Springer Berlin Heidelberg",
pages="350--369",
isbn="978-3-642-37036-6"
}

@article{GPPSY2023,
author = {Ghilezan, Silvia and Pantovi\'{c}, Jovanka and Proki\'{c}, Ivan and Scalas, Alceste and Yoshida, Nobuko},
title = {Precise Subtyping for Asynchronous Multiparty Sessions},
year = {2023},
issue_date = {April 2023},
publisher = {Association for Computing Machinery},
volume = {24},
number = {2},
issn = {1529-3785},
doi = {10.1145/3568422},
journal = {ACM Trans. Comput. Logic},
month = {nov},
articleno = {14},
numpages = {73},
}

@article{MARSHALL2024105234,
title = {Non-linear communication via graded modal session types},
journal = {Information and Computation},
volume = {301},
pages = {105234},
year = {2024},
issn = {0890-5401},
doi = {https://doi.org/10.1016/j.ic.2024.105234},
author = {Danielle Marshall and Dominic Orchard},
}

@article{CY2019,
author = {Castellan, Simon and Yoshida, Nobuko},
title = {Two sides of the same coin: session types and game semantics: a synchronous side and an asynchronous side},
year = {2019},
issue_date = {January 2019},
publisher = {Association for Computing Machinery},
volume = {3},
number = {POPL},
doi = {10.1145/3290340},
journal = {Proc. ACM Program. Lang.},
month = jan,
articleno = {27},
numpages = {29},
}

@InProceedings{vandenheuvel_et_al:LIPIcs.ECOOP.2024.40,
  author =  {van den Heuvel, Bas and Derakhshan, Farzaneh and Balzer, Stephanie},
  title =  {{Information Flow Control in Cyclic Process Networks}},
  booktitle =  {38th European Conference on Object-Oriented Programming (ECOOP 2024)},
  pages =  {40:1--40:30},
  series =  {Leibniz International Proceedings in Informatics (LIPIcs)},
  ISBN =  {978-3-95977-341-6},
  ISSN =  {1868-8969},
  year =  {2024},
  volume =  {313},
  editor =  {Aldrich, Jonathan and Salvaneschi, Guido},
  publisher =  {Schloss Dagstuhl -- Leibniz-Zentrum f{\"u}r Informatik},
  doi =    {10.4230/LIPIcs.ECOOP.2024.40},
}

@article{yao2024semanticlogicalrelationstimed,
author = {Yao, Yue and Iraci, Grant and Chuang, Cheng-En and Balzer, Stephanie and Ziarek, Lukasz},
title = {Semantic Logical Relations for Timed Message-Passing Protocols},
year = {2025},
issue_date = {January 2025},
publisher = {Association for Computing Machinery},
volume = {9},
number = {POPL},
doi = {10.1145/3704895},
journal = {Proc. ACM Program. Lang.},
month = jan,
articleno = {59},
numpages = {32},
}

@inproceedings{CairesPPT13,
  author    = {Lu\'{\i}s Caires and
               Jorge A. P{\'e}rez and
               Frank Pfenning and
               Bernardo Toninho},
  title     = {Behavioral Polymorphism and Parametricity in Session-Based
               Communication},
  booktitle = {ESOP},
  year      = {2013},
  pages     = {330-349},
  publisher = {Springer},
  series    = {LNCS},
  volume    = {7792},
}

@article{JacobsBK22a,
  author    = {Jules Jacobs and
               Stephanie Balzer and
               Robbert Krebbers},
  title     = {Multiparty {GV:} functional multiparty session types with certified
               deadlock freedom},
  journal   = {Proc. {ACM} Program. Lang.},
  volume    = {6},
  number    = {{ICFP}},
  pages     = {466--495},
  year      = {2022},
  doi       = {10.1145/3547638},
}

@article{borceux2005internal,
  title={Internal object actions},
  author={Borceux, Francis and Janelidze, George and Kelly, G Max},
  journal={Comment.\ Math.\ Univ.\ Carolin.},
  volume={46},
  number={2},
  pages={235--255},
  year={2005},
}

@article{smirnov2008graded,
  title={Graded Monads and Rings of Polynomials},
  author={Smirnov, A.L.},
  journal = {J. Math. Sci.},
  volume={151},
  number={3},
  pages={3032--3051},
  year={2008},
  publisher={Springer},
  doi = {10.1007/s10958-008-9013-7}
}

@misc{mellies2012parametric,
  title={Parametric Monads and Enriched Adjunctions},
  author={Melli{\`e}s, Paul-Andr{\'e}},
  howpublished={Manuscript},
  year={2012},
}

@incollection{katsumata2014parametric,
  title={Parametric Effect Monads and Semantics of Effect Systems},
  author={Katsumata, {Shin-ya}},
  booktitle={Proc.\ of 41st Ann.\ ACM SIGPLAN-SIGACT Symp.\ on Principles of Programming Languages, POPL~'14, San Diego, CA, USA, January 20-21, 2014},
  pages={633--645},
  publisher ={ACM},
  year={2014},
  doi ={10.1145/2535838.2535846}
}

@article{orchard2019quantitative,
author = {Orchard, Dominic and Liepelt, Vilem-Benjamin and Eades III, Harley},
title = {Quantitative program reasoning with graded modal types},
year = {2019},
publisher = {ACM},
volume = {3},
number = {ICFP},
doi = {10.1145/3341714},
journal = {Proc. ACM Program. Lang.},
month = jul,
articleno = {110},
numpages = {30},
}

@article{levy2003modelling,
  title = {Modelling environments in call-by-value programming languages},
  journal = {Information and Computation},
  volume = {185},
  number = {2},
  pages = {182--210},
  year = {2003},
  doi = {https://doi.org/10.1016/S0890-5401(03)00088-9},
  author = {Paul Blain Levy and John Power and Hayo Thielecke},
}

@article{sanada2023category,
  title={Category-Graded Algebraic Theories and Effect Handlers},
  author={Sanada, Takahiro},
  journal={Electronic Notes in Theoretical Informatics and Computer Science},
  volume={1},
  year={2023}
}

@inproceedings{dvir2024denotational,
  title={A Denotational Approach to Release/Acquire Concurrency},
  author={Dvir, Yotam and Kammar, Ohad and Lahav, Ori},
  booktitle={European Symposium on Programming},
  pages={121--149},
  year={2024},
  organization={Springer}
}

@inproceedings{rivas2024concurrent,
  title={Concurrent monads for shared state},
  author={Rivas, Exequiel and Uustalu, Tarmo},
  booktitle={Proceedings of the 26th International Symposium on Principles and Practice of Declarative Programming},
  pages={1--13},
  year={2024}
}

@inproceedings{benton2016effect,
  title={Effect-dependent transformations for concurrent programs},
  author={Benton, Nick and Hofmann, Martin and Nigam, Vivek},
  booktitle={Proceedings of the 18th International Symposium on Principles and Practice of Declarative Programming},
  pages={188--201},
  year={2016}
}

@article{CastellaniDG24,
  author       = {Ilaria Castellani and
                  Mariangiola Dezani{-}Ciancaglini and
                  Paola Giannini},
  title        = {Global Types and Event Structure Semantics for Asynchronous Multiparty
                  Sessions},
  journal      = {Fundam. Informaticae},
  volume       = {192},
  number       = {1},
  pages        = {1--75},
  year         = {2024},
  doi          = {10.3233/FI-242188},
}

@article{CastellaniDG23b,
  author       = {Ilaria Castellani and
                  Mariangiola Dezani{-}Ciancaglini and
                  Paola Giannini},
  title        = {Event structure semantics for multiparty sessions},
  journal      = {J. Log. Algebraic Methods Program.},
  volume       = {131},
  pages        = {100844},
  year         = {2023},
  doi          = {10.1016/J.JLAMP.2022.100844},
}

@inproceedings{kammar2012algebraic,
  title={Algebraic foundations for effect-dependent optimisations},
  author={Kammar, Ohad and Plotkin, Gordon D},
  booktitle={Proceedings of the 39th annual ACM SIGPLAN-SIGACT symposium on Principles of programming languages},
  pages={349--360},
  year={2012}
}

@inproceedings{plotkin2009handlers,
  title={Handlers of algebraic effects},
  author={Plotkin, Gordon and Pretnar, Matija},
  booktitle={European Symposium on Programming},
  pages={80--94},
  year={2009},
  organization={Springer}
}

@inproceedings{dvir2025two,
  title={Two-sorted algebraic decompositions of Brookes’s shared-state denotational semantics},
  author={Dvir, Yotam and Kammar, Ohad and Lahav, Ori and Plotkin, Gordon},
  booktitle={Foundations of Software Science and Computation Structures: 28th International Conference, FoSSaCS 2025, Held as Part of the International Joint Conferences on Theory and Practice of Software, ETAPS 2025, Hamilton, ON, Canada, May 3--8, 2025, Proceedings},
  volume={15691},
  pages={377},
  year={2025},
  organization={Springer}
}

@article{bravetti2017undecidability,
  title={Undecidability of asynchronous session subtyping},
  author={Bravetti, Mario and Carbone, Marco and Zavattaro, Gianluigi},
  journal={Information and Computation},
  volume={256},
  pages={300--320},
  year={2017},
  publisher={Elsevier}
}

@article{bravetti2018,
title = {On the boundary between decidability and undecidability of asynchronous session subtyping},
journal = {Theoretical Computer Science},
volume = {722},
pages = {19-51},
year = {2018},
issn = {0304-3975},
doi = {https://doi.org/10.1016/j.tcs.2018.02.010},
author = {Mario Bravetti and Marco Carbone and Gianluigi Zavattaro},
}

@inproceedings{BocchiK0025,
  author       = {Laura Bocchi and
                  Andy King and
                  Maurizio Murgia and
                  Simon Thompson},
  editor       = {Patricia Bouyer and
                  Jaco van de Pol},
  title        = {Abstract Subtyping for Asynchronous Multiparty Sessions},
  booktitle    = {36th International Conference on Concurrency Theory, {CONCUR} 2025,
                  August 26-29, 2025, Aarhus, Denmark},
  series       = {LIPIcs},
  volume       = {348},
  pages        = {10:1--10:19},
  publisher    = {Schloss Dagstuhl - Leibniz-Zentrum f{\"{u}}r Informatik},
  year         = {2025},
  doi          = {10.4230/LIPICS.CONCUR.2025.10},
}

@inproceedings{EY2024,
  author = {Burak Ekici and Nobuko Yoshida},
  title = {{Completeness of Asynchronous Session Tree Subtyping in Coq}},
  booktitle = {15th International Conference on Interactive Theorem Proving, 2024, Tbilisi, Georgia},
  series = {LIPIcs},
  volume = {309},
  pages = {6:1--6:20},
  publisher = {Schloss Dagstuhl - Leibniz-Zentrum f{"{u}}r Informatik},
  year = 2024
}

@article{DBLP:journals/pacmpl/Castro-PerezY20,
  author       = {David Castro{-}Perez and
                  Nobuko Yoshida},
  title        = {{CAMP:} cost-aware multiparty session protocols},
  journal      = {Proc. {ACM} Program. Lang.},
  volume       = {4},
  number       = {{OOPSLA}},
  pages        = {155:1--155:30},
  year         = {2020},
  doi          = {10.1145/3428223},
}

@inproceedings{DBLP:conf/ppopp/CutnerYV22,
  author       = {Zak Cutner and
                  Nobuko Yoshida and
                  Martin Vassor},
  editor       = {Jaejin Lee and
                  Kunal Agrawal and
                  Michael F. Spear},
  title        = {Deadlock-free asynchronous message reordering in {R}ust with multiparty
                  session types},
  booktitle    = {PPoPP '22: 27th {ACM} {SIGPLAN} Symposium on Principles and Practice
                  of Parallel Programming, Seoul, Republic of Korea, April 2 - 6, 2022},
  pages        = {246--261},
  publisher    = {{ACM}},
  year         = {2022},
  doi          = {10.1145/3503221.3508404},
}

@article{KPY2017,
  author = {Dimitrios Kouzapas and Jorge A. Pérez and Nobuko Yoshida},
  title = {{Characteristic Bisimulation for Higher-Order Session Processes}},
  journal = {Acta Informatica},
  volume = {54},
  issue = {3},
  pages = {271--341},
  publisher = {Springer},
  doi = "10.1007/s00236-016-0289-7",
  year = 2017
}

@article{kouzapas2014globally,
  title={Globally governed session semantics},
  author={Kouzapas, Dimitrios and Yoshida, Nobuko},
  journal={Logical Methods in Computer Science},
  volume={10},
  year={2014},
  publisher={Episciences. org}
}

@article{kouzapas2019relative,
  title={On the relative expressiveness of higher-order session processes},
  author={Kouzapas, Dimitrios and P{\'e}rez, Jorge A and Yoshida, Nobuko},
  journal={Information and Computation},
  volume={268},
  pages={104433},
  year={2019},
  publisher={Elsevier}
}

@article{kavvos2025adequacy,
  title={Adequacy for Algebraic Effects Revisited},
  author={Kavvos, GA},
  journal={Proceedings of the ACM on Programming Languages},
  volume={9},
  number={OOPSLA1},
  pages={927--955},
  year={2025},
  publisher={ACM}
}

@inproceedings{plotkin2001adequacy,
  title={Adequacy for algebraic effects},
  author={Plotkin, Gordon and Power, John},
  booktitle={International Conference on Foundations of Software Science and Computation Structures},
  pages={1--24},
  year={2001},
  organization={Springer}
}

\appendix
\section{Detailed proofs}
\subsection{Proofs for \cref{sec:session-types}}

Before giving the proofs for this section, we slightly elaborate on our definition of subtyping.
\begin{definition}\label{asynchronous-type-pre-simulation}
    Let $\Theta$ be a set of session type variables.
    An \emph{asynchronous type pre-simulation}\footnotemark{} is a relation $R$ between session types over $\Theta$, such that the following conditions hold when $\sty{T} \,R\, \sty{U}$.
    \footnotetext{We call this a \emph{pre-simulation} rather than a simulation, because the premises of (1) -- (4) only look at the syntactically outermost choice in the session types. A \emph{simulation} would require, for instance, that $\Recvs{\pp}{(\sty{T})}$ implies $\Recvs{\pp}{(\sty{U})}$ (and indeed subtyping does satisfy this stronger condition -- see \cref{subtype-lifting}).}
    \begin{enumerate}
        \item If $\unfold{\sty{T}} = \ltypesend{\pp}{\ell_i}{\basetype_i}{\sty{T}_i}{i \in I}$, then for every $i \in I$, there is some $\sty{U}_i$ such that $\reduceinternal{\sty{U}}{\sty{U}_i}{\pp}{\ell_i}{\basetype_i}$ and $\sty{T}_i \,R\, \sty{U}_i$.
        \item If $\unfold{\sty{T}} = \ltyperecv{\pp}{\ell_i}{\basetype_i}{\sty{T}_i}{i \in I}$, then $\Recvs{\pp}(\sty{U})$.
        \item If $\unfold{\sty{U}} = \ltypesend{\pp}{\ell_i}{\basetype_i}{\sty{U}_i}{i \in I}$, then $\Sends{\pp}(\sty{T})$.
        \item If $\unfold{\sty{U}} = \ltyperecv{\pp}{\ell_i}{\basetype_i}{\sty{U}_i}{i \in I}$, then for every $i \in I$, there is some $\sty{T}_i$ such that $\reduceexternal{\sty{T}}{\sty{T}_i}{\pp}{\ell_i}{\basetype_i}$ and $\sty{T}_i \,R\, \sty{U}_i$.
        \item For every $\ltypevar{X} \in \Theta$, we have $\unfold{\sty{T}} = \ltypevar{X}$ if and only if $\unfold{\sty{U}} = \ltypevar{X}$.
    \end{enumerate}
\end{definition}
Asynchronous subtyping $\subtype_\Theta$ is, by definition, the largest asynchronous type pre-simulation, in the sense that it is an asynchronous type pre-simulation, and every other such asynchronous type simulation is contained in $\subtype_\Theta$.
Thus, to show that $\sty{T} \subtype_\Theta \sty{U}$ holds for a specific $\sty{T}$ and $\sty{U}$, it suffices to exhibit an asynchronous type pre-simulation $R$ such that $\sty{T} \,R\, \sty{U}$ holds.
Concretely, we can construct $\subtype_\Theta$ as
\[
  \subtype_\Theta = \bigcup \{R \mid R~\text{is an asynchronous type pre-simulation}\}
\]
This union is an asynchronous type pre-simulation by the following lemma, and is clearly the largest such.
\begin{lemma}\label{union-pre-simulation}
  Fix a set $\Theta$ of session type variables.
  If $\mathcal{R}$ is a set of asynchronous type pre-simulations, then the union $\bigcup \mathcal{R}$ is also an asynchronous type pre-simulation.
  Hence $\subtype_\Theta$ is the union
  \[
    \bigcup \{R \mid R~\text{is an asynchronous type pre-simulation}\}
  \]
\end{lemma}
\begin{proof}
  It is easy to verify that $\bigcup \mathcal{R}$ satisfies each of the five conditions of \cref{asynchronous-type-pre-simulation}.
  Clearly taking the union over all asynchronous type pre-simulations gives the largest such, namely $\subtype_\Theta$.
\end{proof}

\begin{lemma}\label{inductive-predicates-substitution}
  Let $\pp$ be a participant, and let $\sty{V}_1, \dots, \sty{V}_n$ be session types.
  Also let $\sty{U}$ be a session type, and define $\sty{T} = \subst{\sty{U}}{\ltypevar{X_1} \mapsto \sty{V}_1, \dots, \ltypevar{X_n} \mapsto \sty{V}_n}$.
  \begin{enumerate}
    \item If $\reduceinternal{\sty{U}}{\sty{U}'}{\pp}{\ell}{\basetype}$, then $\reduceinternal{\sty{T}}{\sty{T}'}{\pp}{\ell}{\basetype}$, where $\sty{T}' = \subst{\sty{U}'}{\ltypevar{X_1} \mapsto \sty{V}_1, \dots, \ltypevar{X_n} \mapsto \sty{V}_n}$.
    \item If $\Recvs{\pp}(\sty{U})$, then $\Recvs{\pp}(\sty{T})$.
    \item If $\reduceexternal{\sty{U}}{\sty{U}'}{\pp}{\ell}{\basetype}$, then $\reduceexternal{\sty{T}}{\sty{T}'}{\pp}{\ell}{\basetype}$, where $\sty{T}' = \subst{\sty{U}'}{\ltypevar{X_1} \mapsto \sty{V}_1, \dots, \ltypevar{X_n} \mapsto \sty{V}_n}$.
    \item If $\Sends{\pp}(\sty{U})$, then $\Sends{\pp}(\sty{T})$.
  \end{enumerate}
\end{lemma}
\begin{proof}
  Each of the proofs is a trivial induction.
\end{proof}

\begin{lemma}\label{inductive-predicates-substitution-backwards}
  Let $\pp$ be a participant, and let $\sty{V}_1, \dots, \sty{V}_n$ be session types.
  Let $\sty{U}$ be a session type with free type variables $\ltypevar{X_1}, \dots, \ltypevar{X_n}$, and define $\sty{U}' = \subst{\sty{U}}{\ltypevar{X_1} \mapsto \sty{V}_1, \dots, \ltypevar{X_n} \mapsto \sty{V}_n}$.
  \begin{enumerate}
    \item If $\Recvs{\pp}(\sty{U}')$, then either (a) $\Recvs{\pp}(\sty{V}_k)$ for some $k$, or (b) $\Recvs{\pp}(\sty{U})$.
    \item If $\Sends{\pp}(\sty{U}')$, then either (a) $\Sends{\pp}(\sty{V}_k)$ for some $k$, or (b) $\Sends{\pp}(\sty{U})$.
  \end{enumerate}
\end{lemma}
\begin{proof}
  We give the proof of (1); the proof of (2) is similar.

  We proceed by induction on the structure of $\sty{U}$.
  \begin{itemize}
    \item If $\sty{U}$ is a type variable, then we have (a) trivially.
    \item If $\sty{U} = \ltypesend{\pq}{\ell_j}{\basetype_j}{\sty{U}_j}{j \in J}$, then we do not have $\Recvs{\pp}(\sty{U}')$; this is a contradiction.
    \item If $\sty{U} = \ltyperecv{\pq}{\ell_j}{\basetype_j}{\sty{U}_j}{j \in J}$, then we have $\Recvs{\pp}(\subst{\sty{U}_j}{\ltypevar{X_1} \mapsto \sty{V}_1, \dots, \ltypevar{X_n} \mapsto \sty{V}_n})$ for every $j$. The inductive hypothesis tells that either (a) holds, or that $\Recvs{\pp}(\sty{U}_j)$ holds for every $j$, in which case (b) follows.
    \item If $\sty{U} = \ltyperec{\ltypevar{X}_{n + 1}}{\sty{U}''}$, then we have $\Recvs{\pp}(\subst{\sty{U}''}{\ltypevar{X_1} \mapsto \sty{V}_1, \dots, \ltypevar{X_n} \mapsto \sty{V}_n, \ltypevar{X_{n+1}} \mapsto \sty{V}_{n + 1}})$, where we define $\sty{V}_{n + 1} = \ltypevar{X}_{n + 1}$. Since $\Recvs{\pp}(\sty{V}_{n + 1})$ does not hold, the inductive hypothesis tells us that either (a) holds, or that $\Recvs{\pp}(\sty{U}'')$ holds, in which case (b) holds.
  \end{itemize}
\end{proof}

\begin{lemma}\label{inductive-predicates-unfold}\
  \begin{enumerate}
    \item $\reduceinternal{\sty{U}}{\sty{U}'}{\pp}{\ell}{\basetype}$ if and only if $\reduceinternal{\unfold{\sty{U}}}{\sty{U}'}{\pp}{\ell}{\basetype}$.
    \item $\Recvs{\pp}(\sty{U})$ if and only if $\Recvs{\pp}(\unfold{\sty{U}})$.
    \item $\reduceexternal{\sty{U}}{\sty{U}'}{\pp}{\ell}{\basetype}$ if and only if $\reduceexternal{\unfold{\sty{U}}}{\sty{U}'}{\pp}{\ell}{\basetype}$.
    \item $\Sends{\pp}(\sty{U})$ if and only if $\Sends{\pp}(\unfold{\sty{U}})$.
  \end{enumerate}
\end{lemma}
\begin{proof}
  We give the proofs of (1) and (2); the proofs of (3) and (4) are similar.
 
  (1) is immediate from the definition of $\reduceinternalsymbol$.

  For (2), we proceed by induction on the structure of $\sty{U}$.
  In both inductions the only non-trivial case is when $\sty{U}$ is a recursive type $\ltyperec{\ltypevar{X}}{\sty{U}'}$, so that $\unfold{\sty{U}} = \subst{\unfold{\sty{U}'}}{\ltypevar{X} \mapsto \sty{U}}$.
  In this case, we have that $\Recvs{\pp}(\sty{U})$ is equivalent to $\Recvs{\pp}(\sty{U}')$ by definition, and this is equivalent to $\Recvs{\pp}(\unfold{\sty{U}'})$ by the inductive hypothesis. The latter implies $\Recvs{\pp}(\subst{\unfold{\sty{U}'}}{\ltypevar{X} \mapsto \sty{U}})$ by \cref{inductive-predicates-substitution}. In the other direction, by \cref{inductive-predicates-substitution-backwards}, $\Recvs{\pp}(\subst{\unfold{\sty{U}'}}{\ltypevar{X} \mapsto \sty{U}})$ implies that either $\Recvs{\pp}(\sty{U})$ or $\Recvs{\pp}(\unfold{\sty{U}'})$, both of which are equivalent to the desired result.
\end{proof}

\begin{lemma}\label{subtype-congruences}\
    \begin{enumerate}
        \item If $\unfold{\sty{U}} = \ltyperecv{\pp}{\ell_i}{\basetype_i}{\sty{U}_i}{i \in I}$, then $\sty{T} \subtype_\Theta \sty{U}$ holds iff, for every $i \in I$, there is some $\sty{T}_i$ such that $\reduceexternal{\sty{T}}{\sty{T}_i}{\pp}{\ell_i}{\basetype_i}$ and $\sty{T}_i \subtype_\Theta \sty{U}_i$.
        \item If $\unfold{\sty{T}} = \ltypesend{\pp}{\ell_i}{\basetype_i}{\sty{T}_i}{i \in I}$, then $\sty{T} \subtype_\Theta \sty{U}$ holds iff, for every $i \in I$, there is some $\sty{U}_i$ such that $\reduceinternal{\sty{U}}{\sty{U}_i}{\pp}{\ell_i}{\basetype_i}$ and $\sty{T}_i \subtype_\Theta \sty{U}_i$.
    \end{enumerate}
\end{lemma}
\begin{proof}
  We give the proof of (1); the proof of (2) is similar.

  The only if direction is trivial.
  For the if direction, note that item (3) in the definition of subtyping holds trivially, while (4) is by assumption.
  To establish the remaining cases, since $I$ is finite, there is some natural number $h$ such that for every $i$, the derivation of $\reduceexternal{\sty{T}}{\sty{T}_i}{\pp}{\ell_i}{\basetype_i}$ has height at most $h$.
  We can thus proceed by induction on $h$.
  If $h = 0$, then $\sty{U}$ has the form $\ltyperecv{\pp}{\ell'_j}{\basetype'_j}{\sty{U}'_j}{j \in I}$.
  We prove that $\sty{T} \subtype_\Theta \sty{U}$ by showing the remaining cases in the definition of subtyping individually: (1) and (5) are trivial, and for (2), we have $\Recvs{\pp}(\sty{U})$ by \cref{inductive-predicates-unfold}, because $\Recvs{\pp}(\unfold{U})$.
  If $h > 0$, then we have the following cases:
  \begin{itemize}
    \item If $\sty{T} = \ltyperecv{\pq}{\ell'_j}{\basetype'_j}{\sty{T}'_j}{j \in J}$ with $\pp \neq \pq$, then (1) and (5) are trivial. For (2), $\sty{T}_i$ is necessarily an external choice on $q$, so that $\Recvs{\pq}(\sty{U}_i)$ for each $i$; this implies $\Recvs{\pq}(\unfold{\sty{U}})$ and hence $\Recvs{\pq}(\sty{U})$ by \cref{inductive-predicates-unfold}.
    \item If $\sty{T} = \ltypesend{\pq}{\ell'_j}{\basetype'_j}{\sty{T}'_j}{j \in J}$, then we have $\sty{T}_i = \ltypesend{\pq}{\ell'_j}{\basetype'_j}{\sty{T}'_{ij}}{j \in J}$, with $\reduceexternal{\sty{T}'_j}{\sty{T}'_{ij}}{\pp}{\ell_i}{\basetype_i}$.
      The only remaining non-trivial case of subtyping is (1), for which we need to show that there exist $\sty{U}'_j$ such that $\reduceinternal{\sty{U}}{\sty{U}'_j}{\pq}{\ell'_j}{\basetype'_j}$. and $\sty{T}'_j \subtype_\Theta \sty{U}'_j$.
      We take $\sty{U}'_j = \ltyperecv{\pp}{\ell_i}{\basetype_i}{\sty{U}'_{ij}}{i \in I}$, where $\reduceinternal{\sty{U}_i}{\sty{U}'_{ij}}{\pq}{\ell'_j}{\basetype'_j}$ and $\sty{T}'_{ij} \subtype_\Theta \sty{U}'_{ij}$.
      The latter exist because $\sty{T}_i \subtype_\Theta \sty{U}_i$, we have $\reduceinternal{\unfold{\sty{U}}}{\sty{U}'_j}{\pq}{\ell'_j}{\basetype'_j}$ and hence $\reduceinternal{\sty{U}}{\sty{U}'_j}{\pq}{\ell'_j}{\basetype'_j}$ by \cref{inductive-predicates-unfold}.
      Finally, $\sty{T}'_{j} \subtype_\Theta \sty{U}'_{j}$ holds by the inductive hypothesis.
    \item If $\sty{T} = \ltyperec{\ltypevar{X}}{\sty{T}'}$, then from the inductive hypothesis we obtain $\unfold{\sty{T}} \subtype_\Theta \sty{U}$, which implies $\sty{T} \subtype_\Theta \sty{U}$ trivially.\qed
  \end{itemize}
\end{proof}

We define relations $\prectype{\sty{U}}{\ltyperecv{\pp}{\ell_i}{\basetype_i}{\sty{U}_i}{i \in I}}$ and $\prectype{\sty{T}}{\ltypesend{\pp}{\ell_i}{\basetype_i}{\sty{T}_i}{i \in I}}$ inductively.
\begin{gather*}
  \begin{prooftree}
    \infer0[\RuleName{$\prectypesymbol\&$-base}]
    {\prectype{\ltyperecv{\pp}{\ell_i}{\basetype_i}{\sty{U}_i}{i \in I}}{\ltyperecv{\pp}{\ell_i}{\basetype_i}{\sty{U}_i}{i \in I}}}
  \end{prooftree}
  \\[2ex]
  \begin{prooftree}
    \hypo{\pp \neq \pq}
    \hypo{\prectype{\sty{U}_j}{\ltyperecv{\pp}{\ell_i}{\basetype_i}{\sty{U}_{i, j}}{i \in I_j}}~\text{for all}~j \in J}
    \hypo{I = \bigcup_{j \in J} I_j}
    \hypo{J_i = \{j \in J \mid i \in I_j\}}
    \infer4[\RuleName{$\prectypesymbol\&$-ind}]
    {\prectype{\ltyperecv{\pq}{\ell'_j}{\basetype'_j}{\sty{U}_j}{j \in J}}{\ltyperecv{\pp}{\ell_i}{\basetype_i}{\ltyperecv{\pq}{\ell'_j}{\basetype'_j}{\sty{U}_{i, j}}{j \in J_i}}{i \in I}}}
  \end{prooftree}
  \\[2ex]
  \begin{prooftree}
    \hypo{\prectype{\subst{\sty{U}}{\ltypevar{X} \mapsto \ltyperec{\ltypevar{X}}{\sty{U}}}}{\ltyperecv{\pp}{\ell_i}{\basetype_i}{\sty{U}_i}{i \in I}}}
    \infer1[\RuleName{$\prectypesymbol\&$-rec}]
    {\prectype{(\ltyperec{\ltypevar{X}}{\sty{U}})}{\ltyperecv{\pp}{\ell_i}{\basetype_i}{\sty{U}_i}{i \in I}}}
  \end{prooftree}
  \\[2ex]
  \begin{prooftree}
    \infer0[\RuleName{$\prectypesymbol\oplus$-base}]
    {\prectype{\ltypesend{\pp}{\ell_i}{\basetype_i}{\sty{U}_i}{i \in I}}{\ltypesend{\pp}{\ell_i}{\basetype_i}{\sty{U}_i}{i \in I}}}
  \end{prooftree}
  \\[2ex]
  \begin{prooftree}
    \hypo{\pp \neq \pq}
    \hypo{\prectype{\sty{U}_j}{\ltypesend{\pp}{\ell_i}{\basetype_i}{\sty{U}_{i, j}}{i \in I_j}}~\text{for all}~j \in J}
    \hypo{I = \bigcup_{j \in J} I_j}
    \hypo{J_i = \{j \in J \mid i \in I_j\}}
    \infer4[\RuleName{$\prectypesymbol\oplus$-ind}]
    {\prectype{\ltypesend{\pq}{\ell'_j}{\basetype'_j}{\sty{U}_j}{j \in J}}{\ltypesend{\pp}{\ell_i}{\basetype_i}{\ltypesend{\pq}{\ell'_j}{\basetype'_j}{\sty{U}_{i, j}}{j \in J_i}}{i \in I}}}
  \end{prooftree}
  \\[2ex]
  \begin{prooftree}
    \hypo{\prectype{\subst{\sty{U}}{\ltypevar{X} \mapsto \ltyperec{\ltypevar{X}}{\sty{U}}}}{\ltypesend{\pp}{\ell_i}{\basetype_i}{\sty{U}_i}{i \in I}}}
    \infer1[\RuleName{$\prectypesymbol\oplus$-rec}]
    {\prectype{(\ltyperec{\ltypevar{X}}{\sty{U}})}{\ltypesend{\pp}{\ell_i}{\basetype_i}{\sty{U}_i}{i \in I}}}
  \end{prooftree}
\end{gather*}

\begin{lemma}\label{prectype-substitution}
  Let $\pp$ be a participant, and let $\sty{V}_1, \dots, \sty{V}_n$ be session types.
  Also let $\sty{U}$ be a session type, and define $\sty{T} = \subst{\sty{U}}{\ltypevar{X_1} \mapsto \sty{V}_1, \dots, \ltypevar{X_n} \mapsto \sty{V}_n}$.
  \begin{enumerate}
    \item If $\prectype{\sty{U}}{\ltyperecv{\pp}{\ell_i}{\basetype_i}{\sty{U}_i}{i \in I}}$, then $\prectype{\sty{T}}{\ltyperecv{\pp}{\ell_i}{\basetype_i}{\sty{T}_i}{i \in I}}$, where $\sty{T}_i = \subst{\sty{U}_i}{\ltypevar{X_1} \mapsto \sty{V}_1, \dots, \ltypevar{X_n} \mapsto \sty{V}_n}$.
    \item If $\prectype{\sty{U}}{\ltypesend{\pp}{\ell_i}{\basetype_i}{\sty{U}_i}{i \in I}}$, then $\prectype{\sty{T}}{\ltypesend{\pp}{\ell_i}{\basetype_i}{\sty{T}_i}{i \in I}}$, where $\sty{T}_i = \subst{\sty{U}_i}{\ltypevar{X_1} \mapsto \sty{V}_1, \dots, \ltypevar{X_n} \mapsto \sty{V}_n}$.
  \end{enumerate}
\end{lemma}
\begin{proof}
  Each of the proofs is a trivial induction.
\end{proof}

\begin{lemma}\label{recvs-prectype}\
  \begin{enumerate}
    \item We have $\Recvs{\pp}{(\sty{U})}$ iff there exist $\{(\ell_i, \basetype_i, \sty{U}_i)\}_{i \in I}$ such that $\prectype{\sty{U}}{\ltyperecv{\pp}{\ell_i}{\basetype_i}{\sty{U}_i}{i \in I}}$.
    \item We have $\Sends{\pp}{(\sty{T})}$ iff there exist $\{(\ell_i, \basetype_i, \sty{T}_i)\}_{i \in I}$ such that $\prectype{\sty{T}}{\ltypesend{\pp}{\ell_i}{\basetype_i}{\sty{T}_i}{i \in I}}$.
  \end{enumerate}
\end{lemma}
\begin{proof}
  We give the proof of (1); the proof of (2) is similar.

  The if direction is an easy induction on the derivation of $\prectypesymbol$, using \cref{inductive-predicates-unfold} for the recursive case.
  For the only if direction, we proceed by induction on $\Recvs{\pp}{(\sty{U})}$.
  Uhe only non-trivial case is \RuleName{Recvs-rec}. In this case, if $\Recvs{\pp}{(\sty{U})}$ holds, then we have $\prectype{\sty{U}}{\ltypesend{\pp}{\ell_i}{\basetype_i}{\sty{U}_i}{i \in I}}$.
  By \cref{prectype-substitution}, we therefore have $\prectype{\unfold{\ltyperec{\ltypevar{X}}{\sty{U}}}}{\ltypesend{\pp}{\ell_i}{\basetype_i}{\sty{U}_i}{i \in I}}$, and hence $\prectype{\ltyperec{\ltypevar{X}}{\sty{U}}}{\ltypesend{\pp}{\ell_i}{\basetype_i}{\sty{U}_i}{i \in I}}$ as required.
\end{proof}

\begin{lemma}\label{subtype-prectype}
    Assume that $\sty{T} \subtype_{\Theta} \sty{U}$.
    \begin{enumerate}
        \item If $\unfold{\sty{T}} = \ltyperecv{\pp}{\ell_i}{\basetype_i}{\sty{T}_i}{i \in I}$, then there exist $J \subseteq I$ and $\{\sty{U}_i\}_{i \in J}$, such that $\prectype{\sty{U}}{\ltyperecv{\pp}{\ell_i}{\basetype_i}{\sty{U}_i}{i \in J}}$, and $\sty{T}_i \subtype_{\Theta} \sty{U}_i$ for each $i \in J$.
        \item If $\unfold{\sty{U}} = \ltypesend{\pp}{\ell_i}{\basetype_i}{\sty{U}_i}{i \in I}$, then there exist $J \subseteq I$ and $\{\sty{T}_i\}_{i \in J}$, such that $\prectype{\sty{T}}{\ltypesend{\pp}{\ell_i}{\basetype_i}{\sty{T}_i}{i \in J}}$, and $\sty{T}_i \subtype_{\Theta} \sty{U}_i$ for each $i \in J$.
    \end{enumerate}
\end{lemma}
\begin{proof}
  We give the proof of (1); the proof of (2) is similar.
  We have $\Recvs{\pp}(\sty{U})$ because $\sty{T} \subtype_{\Theta} \sty{U}$, and hence $\prectype{\sty{U}}{\ltyperecv{\pp}{\ell'_k}{\basetype'_k}{\sty{U}'_k}{k \in K}}$ by \cref{recvs-prectype}.
  We proceed by induction on the latter.
  \begin{itemize}
    \item If the proof is by rule \RuleName{$\prectypesymbol\&$-base}, then $\sty{U} = \ltyperecv{\pp}{\ell'_k}{\basetype'_k}{\sty{U}'_k}{k \in K}$.
      By subtyping, we have for every $i \in I'$ that $\reduceexternal{\sty{T}}{\sty{T}'_k}{\pp}{\ell'_k}{\basetype'_k}$ for some $\sty{T}'_k \subtype_\Theta \sty{U}'_k$.
      This implies $\reduceexternal{\unfold{\sty{T}}}{\sty{T}'}{\pp}{\ell'_k}{\basetype'_k}$ by \cref{inductive-predicates-unfold}.
      The latter can only be by rule \RuleName{$\reduceexternalsymbol$-base}, so for every $k$, there is some $i$ such that $(\ell_i, \basetype_i, \sty{T}_i) = (\ell'_k, \basetype'_k, \sty{T}'_k)$.
      Hence we can write $\sty{U}$ as $\ltyperecv{\pp}{\ell_i}{\basetype_i}{\sty{U}_i}{i \in J}$ where $J \subseteq I$ and $\sty{T}_i \subtype_\Theta \sty{U}_i$.
    \item If the proof is by rule \RuleName{$\prectypesymbol\&$-ind}, then $\sty{U}$ has the form $\ltyperecv{\pq}{\ell'_k}{\basetype'_k}{\sty{U}'_k}{k \in K}$.
      By subtyping, we have for every $k \in K$ that $\reduceexternal{\sty{T}}{\sty{T}'_k}{\pq}{\ell'_k}{\basetype'_k}$ for some $\sty{T}'_k \subtype_\Theta \sty{U}'_k$.
      \Cref{inductive-predicates-unfold} then implies that $\reduceexternal{\unfold{\sty{T}}}{\sty{T}'_k}{\pq}{\ell'_k}{\basetype'_k}$, which means that $\sty{T}'_k = \ltyperecv{\pp}{\ell_i}{\basetype_i}{\sty{T}'_{ki}}{i \in I_k}$, where $I_k \subseteq I$ and $\reduceexternal{\sty{T}_i}{\sty{T}_{ki}}{\pq}{\ell'_k}{\basetype'_k}$.
      By the inductive hypothesis, there exist $J_k \subseteq I_k$ and $\{\sty{U}_{ki}\}_{i \in J_k}$, such that $\prectype{\sty{U}'_k}{\ltyperecv{\pp}{\ell_i}{\basetype_i}{\sty{U}_{ki}}{i \in J_k}}$, and $\sty{T}_{ki} \subtype_{\Theta} \sty{U}_{ki}$ for each $i \in J_k$.
      We can therefore conclude by taking $J = \bigcup_{k \in K} J_k$ and $\sty{U}_i = \ltyperecv{\pq}{\ell'_k}{\basetype'_k}{\sty{U}_{ki}}{k \in K_i}$, where $K_i = \{k \in K \mid i \in J_k\}$; we have $\sty{T}_{i} \subtype_{\Theta} \sty{U}_i$ by \cref{subtype-congruences}.
    \item If the proof is by \RuleName{$\prectypesymbol\&$-rec}, then the result is immediate from the inductive hypothesis.
  \end{itemize}
\end{proof}

\begin{lemma}\label{inductive-predicates-swapping}
  \begin{enumerate}
    \item If $\pp \neq \pq$, $I$ is a non-empty finite set, and $\reduceinternal{\sty{U}}{\reduceinternal{\sty{U}_i}{\sty{U}'_i}{\pq}{\ell'}{\basetype'}}{\pp}{\ell_i}{\basetype_i}$ for each $i \in I$, then there is some $\sty{U}'$ such that $\reduceinternal{\sty{U}}{\reduceinternal{\sty{U}'}{\sty{U}'_i}{\pp}{\ell_i}{\basetype_i}}{\pq}{\ell'}{\basetype'}$ for each $i \in I$.
    \item If $I$ is a non-empty finite set, then there exist $\sty{U}_i$ such that $\prectype{\sty{U}}{\ltyperecv{\pp}{\ell_i}{\basetype_i}{\sty{U}_i}{i \in I}}$ and $\reduceinternal{\sty{U}_i}{\sty{U}'_i}{\pq}{\ell'}{\basetype'}$ for all $i \in I$, exactly when there exists some $\sty{U}'$ such that $\reduceinternal{\sty{U}}{\prectype{\sty{U}'}{\ltyperecv{\pp}{\ell_i}{\basetype_i}{\sty{U}'_i}{i \in I}}}{\pq}{\ell'}{\basetype'}$.
    \item If $\pp \neq \pq$, $\prectype{\sty{U}}{\ltyperecv{\pp}{\ell_i}{\basetype_i}{\sty{U}_i}{i \in I}}$ and $\prectype{\sty{U}_i}{\ltyperecv{\pq}{\ell'_j}{\basetype'_j}{\sty{U}_{ij}}{j \in J_i}}$, then there exist $\sty{U}'_j$ such that $\prectype{\sty{U}}{\ltyperecv{\pq}{\ell'_j}{\basetype'_j}{\sty{U}'_j}{j \in J}}$ and $\prectype{\sty{U}'_j}{\ltyperecv{\pp}{\ell_i}{\basetype_i}{\sty{U}_{ij}}{i \in I_j}}$, where $I_j = \{i \in I \mid j \in J_i\}$ and $J = \cup_{i \in I}J_i$.
    \item If $\reduceinternal{\sty{T}}{\sty{U}}{\pp}{\ell}{\basetype}$ and $\reduceexternal{\sty{T}}{\sty{T}'}{\pq}{\ell'}{\basetype'}$, then there is some $\sty{U}'$ such that $\reduceinternal{\sty{T}'}{\sty{U}'}{\pp}{\ell}{\basetype}$ and $\reduceexternal{\sty{U}}{\sty{U}'}{\pq}{\ell'}{\basetype'}$.
    \item If $\pp \neq \pq$, $I$ is a non-empty finite set, and $\reduceexternal{\sty{U}}{\reduceexternal{\sty{U}_i}{\sty{U}'_i}{\pq}{\ell'}{\basetype'}}{\pp}{\ell_i}{\basetype_i}$ for each $i \in I$, then there is some $\sty{U}'$ such that $\reduceexternal{\sty{U}}{\reduceexternal{\sty{U}'}{\sty{U}'_i}{\pp}{\ell_i}{\basetype_i}}{\pq}{\ell'}{\basetype'}$ for each $i \in I$.
    \item If $I$ is a non-empty finite set, then there exist $\sty{U}_i$ such that $\prectype{\sty{U}}{\ltypesend{\pp}{\ell_i}{\basetype_i}{\sty{U}_i}{i \in I}}$ and $\reduceexternal{\sty{U}_i}{\sty{U}'_i}{\pq}{\ell'}{\basetype'}$ for all $i \in I$, exactly when there exists some $\sty{U}'$ such that $\reduceexternal{\sty{U}}{\prectype{\sty{U}'}{\ltypesend{\pp}{\ell_i}{\basetype_i}{\sty{U}'_i}{i \in I}}}{\pq}{\ell'}{\basetype'}$.
    \item If $\pp \neq \pq$, $\prectype{\sty{U}}{\ltypesend{\pp}{\ell_i}{\basetype_i}{\sty{U}_i}{i \in I}}$ and $\prectype{\sty{U}_i}{\ltypesend{\pq}{\ell'_j}{\basetype'_j}{\sty{U}_{ij}}{j \in J_i}}$, then there exist $\sty{U}'_j$ such that $\prectype{\sty{U}}{\ltypesend{\pq}{\ell'_j}{\basetype'_j}{\sty{U}'_j}{j \in J}}$ and $\prectype{\sty{U}'_j}{\ltypesend{\pp}{\ell_i}{\basetype_i}{\sty{U}_{ij}}{i \in I_j}}$, where $I_j = \{i \in I \mid j \in J_i\}$ and $J = \cup_{i \in I}J_i$.
  \end{enumerate}
\end{lemma}
\begin{proof}
  For (1), since $I$ is finite, there is some natural number $h$ such that for every $i$, the derivation of $\reduceinternal{\sty{U}}{\sty{U}_i}{\pp}{\ell_i}{\basetype_i}$ has height at most $h$.
  We proceed by induction on $h$.
  If $h = 0$, then $\sty{U} = \ltypesend{\pp}{\ell_i}{\basetype_i}{\sty{U}_i}{i \in J}$ with $I \subseteq J$.
  We take $\sty{U}' = \ltypesend{\pp}{\ell_i}{\basetype_i}{\sty{U}'_i}{i \in J}$.
  If $h > 0$, then we have the following cases.
  \begin{itemize}
    \item If $\sty{U} = \ltypesend{\pq}{\ell'_j}{\basetype'_j}{\sty{U}'_j}{j \in J}$, then since $I$ is non-empty, we have $(\ell', \basetype') = (\ell'_j, \basetype'_j)$ for some $j$; we can then take $\sty{U}' = \sty{U}'_j$.
    \item If $\sty{U} = \ltypesend{\pr}{\ell''_j}{\basetype''_j}{\sty{U}''_j}{j \in J}$, with $\pr \not\in \{\pp,\pq\}$, then for every $i \in I$ we have $\sty{U}_i = \ltypesend{\pr}{\ell''_j}{\basetype''_j}{\sty{U}_{ij}}{j \in J_i}$ and $\sty{U}'_i = \ltypesend{\pr}{\ell''_j}{\basetype''_j}{\sty{U}'_{ij}}{j \in J'_i}$, where $J'_i \subseteq J_i \subseteq J$, and $\reduceinternal{\sty{U}''_j}{\reduceinternal{\sty{U}_{ij}}{\sty{U}'_{ij}}{\pq}{\ell'}{\basetype'}}{\pp}{\ell_i}{\basetype_i}$ for each $j \in J'_i$.
      For every $j \in \bigcup_{i \in I}{J'_i}$, the inductive hypothesis provides us with $\sty{U}'_j$ such that $\reduceinternal{\sty{U}''_j}{\reduceinternal{\sty{U}'_j}{\sty{U}'_{ij}}{\pp}{\ell_i}{\basetype_i}}{\pq}{\ell'}{\basetype'}$ for all $i \in I$ such that $j \in J'_i$.
      We define $\sty{U}' = \ltypesend{\pr}{\ell''_j}{\basetype''_j}{\sty{U}''_j}{j \in \bigcup_{i \in I} J'_i}$.
    \item If $\sty{U}$ is an external choice, then the result follows from the inductive hypothesis.
    \item If $\sty{U}$ is a recursive type, then the result follows from the inductive hypothesis using the rules for recursive types.
  \end{itemize}

  For (2), we first give the forwards direction, which is by induction on the derivation of $\prectype{\sty{U}}{\ltyperecv{\pp}{\ell_i}{\basetype_i}{\sty{U}_i}{i \in I}}$.
  For \RuleName{$\prectypesymbol\&$-base}, we take $\sty{U}' = \ltyperecv{\pp}{\ell_i}{\basetype_i}{\sty{U}'_i}{i \in I}$.
  For \RuleName{$\prectypesymbol\&$-ind}, we have $\sty{U} = \ltyperecv{\pr}{\ell''_j}{\basetype''_j}{\sty{U}''_j}{j \in J}$ and $\sty{U}_i = \ltyperecv{\pr}{\ell''_j}{\basetype''_j}{\sty{U}''_{ij}}{j \in J_i}$, with $\pp \neq \pq$, $I = \bigcup_{j \in J} I_j$, $J_i = \{j \in J \mid i \in I_j\}$, and $\prectype{\sty{U''}_j}{\ltyperecv{\pp}{\ell_i}{\basetype_i}{\sty{U}''_{ij}}{i \in I_j}}$.
  Since $\reduceinternal{\sty{U}_i}{\sty{U}'_i}{\pq}{\ell'}{\basetype'}$, for every $j \in J_i$, we have $\reduceinternal{\sty{U}_{ij}}{\sty{U}'_{ij}}{\pq}{\ell'}{\basetype'}$ for some $\sty{U}'_{ij}$.
  The inductive hypothesis provides us with $\sty{U}'_j$ such that $\reduceinternal{\sty{U}''_j}{\prectype{\sty{U}'_j}{\ltyperecv{\pp}{\ell_i}{\basetype_i}{\sty{U}'_{ij}}{i \in I_j}}}{\pq}{\ell'}{\basetype'}$,
  so we conclude by defining $\sty{U}' = \ltyperecv{\pr}{\ell''_j}{\basetype''_j}{\sty{U}'_j}{j \in J}$.
  Finally, for \RuleName{$\prectypesymbol\&$-ind}, the result follows from the inductive hypothesis using the rules for recursive types.
  The reverse direction is a similar induction on $\prectypesymbol$.

  For (3), the proof is by induction on $\prectype{\sty{U}}{\ltyperecv{\pp}{\ell_i}{\basetype_i}{\sty{U}_i}{i \in I}}$, with both cases being similar to those in the proof of (2).

  For (4), we proceed by induction on $\reduceinternal{\sty{T}}{\sty{U}}{\pp}{\ell}{\basetype}$.
  If the derivation is by \RuleName{$\reduceinternalsymbol$-base}, then the derivation of $\reduceexternal{\sty{T}}{\sty{T}'}{\pq}{\ell'}{\basetype'}$ is by \RuleName{$\reduceexternalsymbol$-$\oplus$}, and the result follows immediately from the assumptions of this rule.
  If the derivation is by \RuleName{$\reduceinternalsymbol$-$\oplus$}, then the derivation of $\reduceexternal{\sty{T}}{\sty{T}'}{\pq}{\ell'}{\basetype'}$ is again by \RuleName{$\reduceexternalsymbol$-$\oplus$}.
  The result is then immediate from the inductive hypothesis, by using \RuleName{$\reduceinternalsymbol$-$\oplus$} and \RuleName{$\reduceexternalsymbol$-$\oplus$} again.
  Finally, if the derivation is by \RuleName{$\reduceinternalsymbol$-$\&$}, then the derivation of $\reduceexternal{\sty{T}}{\sty{T}'}{\pq}{\ell'}{\basetype'}$ is by either \RuleName{$\reduceexternalsymbol$-base} or \RuleName{$\reduceexternalsymbol$-$\&$}. In the former case, the result follows immediately, while in the latter case, it follows from the inductive hypothesis, using rules \RuleName{$\reduceinternalsymbol$-$\&$} and \RuleName{$\reduceexternalsymbol$-$\&$} again.

  The proofs of (5), (6) and (7) are similar to the proofs of (1), (2) and (3), respectively.
\end{proof}

\begin{lemma}\label{subtype-lifting}
  Assume that $\sty{T} \subtype_{\Theta} \sty{U}$.
  \begin{enumerate}
    \item If $\reduceinternal{\sty{T}}{\sty{T}'}{\pp}{\ell}{\basetype}$, then there is some $\sty{U'}$ such that $\reduceinternal{\sty{U}}{\sty{U'}}{\pp}{\ell}{\basetype}$ and $\sty{T'} \subtype_{\Theta} \sty{U'}$.
    \item If $\prectype{\sty{T}}{\ltyperecv{\pp}{\ell_i}{\basetype_i}{\sty{T}_i}{i \in I}}$, then there exist $\sty{U}_i$ such that $\prectype{\sty{U}}{\ltyperecv{\pp}{\ell_i}{\basetype_i}{\sty{U}_i}{i \in I}}$ and $\sty{T}_i \subtype_{\Theta} \sty{U}_i$.
    \item If $\Recvs{\pp}{(\sty{T})}$ then $\Recvs{\pp}{(\sty{U})}$.
    \item If $\reduceexternal{\sty{U}}{\sty{U}'}{\pp}{\ell}{\basetype}$, then there is some $\sty{T'}$ such that $\reduceexternal{\sty{T}}{\sty{T'}}{\pp}{\ell}{\basetype}$ and $\sty{T'} \subtype_{\Theta} \sty{U'}$.
    \item If $\prectype{\sty{U}}{\ltypesend{\pp}{\ell_i}{\basetype_i}{\sty{U}_i}{i \in I}}$, then there exist $\sty{T}_i$ such that $\prectype{\sty{T}}{\ltypesend{\pp}{\ell_i}{\basetype_i}{\sty{T}_i}{i \in I}}$ and $\sty{T}_i \subtype_{\Theta} \sty{U}_i$.
    \item If $\Sends{\pp}{(\sty{U})}$ then $\Sends{\pp}{(\sty{T})}$.
  \end{enumerate}
\end{lemma}
\begin{proof}
  We give the proofs of (1), (2), and (3); the proofs of (4), (5) and (6) are similar.

  The proof of (1) is by induction on $\reduceinternal{\sty{T}}{\sty{T}'}{\pp}{\ell}{\basetype}$.
  The base case is an immediate consequence of $\sty{T} \subtype_{\Theta} \sty{U}$.
  For \RuleName{$\reduceinternalsymbol$-$\oplus$}, the result follows from the inductive hypothesis and \cref{inductive-predicates-swapping}(1).
  For \RuleName{$\reduceinternalsymbol$-$\&$}, the result follows from \cref{subtype-prectype}, the inductive hypothesis and \cref{inductive-predicates-swapping}(2).
  For \RuleName{$\reduceinternalsymbol$-Rec}, the result follows from the inductive hypothesis, because $\sty{T} \subtype_{\Theta} \sty{U}$ implies $\unfold{\sty{T}} \subtype_{\Theta} \sty{U}$.

  The proof of (2) is by induction on $\prectype{\sty{T}}{\ltyperecv{\pp}{\ell_i}{\basetype_i}{\sty{T}_i}{i \in I}}$.
  The base case is an immediate consequence of \cref{subtype-prectype}, while for \RuleName{$\prectypesymbol\&$-ind}, the result follows from \cref{subtype-prectype}, the inductive hypothesis, and \cref{inductive-predicates-swapping}(3).
  For \RuleName{$\prectypesymbol$-Rec}, the result follows from the inductive hypothesis, because $\sty{T} \subtype_{\Theta} \sty{U}$ implies $\unfold{\sty{T}} \subtype_{\Theta} \sty{U}$.

  Finally (3) follows from (2), using \cref{recvs-prectype}.
\end{proof}

\admissiblesubtyping*
\begin{proof}
  The congruence rules for internal and external choices are special cases of \cref{subtype-congruences}.

  To prove the congruence rule for $\mu$, we consider the largest relation $R$ such that $\sty{T} R \sty{U}$ implies that either (1) $\sty{T} \subtype_{\Theta} \sty{U}$ or (2) that $\unfold{\sty{T}} = \subst{\sty{T'}}{\ltypevar{X} \mapsto \sty{T''}}$ and $\unfold{\sty{U}} = \subst{\sty{U'}}{\ltypevar{X} \mapsto \sty{U''}}$
  where $\sty{T'}$ is not a type variable, $\sty{T'} \subtype_{\Theta, \ltypevar{X}} \sty{U}'$, and $\sty{T''} R \sty{U''}$.
  Clearly $(\ltyperec{\ltypevar{X}}{\sty{T}}) R (\ltyperec{\ltypevar{X}}{\sty{U}})$ if $\sty{T} \subtype_{\Theta, X} \sty{U}$, so to conclude it is enough to prove that $R$ is an asynchronous type pre-simulation.
  All of the five conditions are easy to verify using \cref{inductive-predicates-substitution}.

  Reflexivity of subtyping is an easy induction on the session type, using the congruence rules along with $\ltypeend \subtype_{\Theta} \ltypeend$ and $\ltypevar{X} \subtype_{\Theta} \ltypevar{X}$.
  Transitivity is an immediate consequence of \cref{subtype-lifting}.
  For substitution, consider $\sty{T} \subtype_{\ltypevar{X}_1, \dots, \ltypevar{X}_n} \sty{U}$.
  We have that $\unfold{\sty{T}}$ is a type variable iff $\unfold{\sty{U}}$ is the same type variable, in which case the result is one of the assumptions.
  If neither is a type variable, then we have to verify conditions (1) to (4) for the substitutions, but these follow immediately from \cref{inductive-predicates-substitution}.
\end{proof}

\begin{lemma}\label{reduce-multiplication}
  \begin{enumerate}
    \item If $\reduceinternal{\sty{T}}{\sty{U}}{\pp}{\ell}{\basetype}$ then $\reduceinternal{\sty{T}\cdot\sty{T}'}{\sty{U}\cdot\sty{T}'}{\pp}{\ell}{\basetype}$.
    \item If $\Sends{\pp}(\sty{T})$, then $\Sends{\pp}(\sty{T}\cdot \sty{T'})$.
    \item If $\reduceexternal{\sty{T}}{\sty{U}}{\pp}{\ell}{\basetype}$ then $\reduceexternal{\sty{T}\cdot\sty{T}'}{\sty{U}\cdot\sty{T}'}{\pp}{\ell}{\basetype}$.
    \item If $\Recvs{\pp}(\sty{T})$, then $\Recvs{\pp}(\sty{T}\cdot \sty{T'})$.
  \end{enumerate}
\end{lemma}
\begin{proof}
  Each of these is a trivial induction.
\end{proof}

\begin{lemma}\label{reduce-multiplication-inverse}
  \begin{enumerate}
    \item If $\reduceinternal{\sty{T}\cdot\sty{T}'}{\sty{U}}{\pp}{\ell}{\basetype}$ and $\Sends{\pp}(\sty{T})$, then $\sty{U} = \sty{U}' \cdot \sty{T}'$, for some $\sty{U}'$ such that $\reduceinternal{\sty{U}}{\sty{U}'}{\pp}{\ell}{\basetype}$.
    \item If $\reduceexternal{\sty{T}\cdot\sty{T}'}{\sty{U}}{\pp}{\ell}{\basetype}$ and $\Recvs{\pp}(\sty{T})$, then $\sty{U} = \sty{U}' \cdot \sty{T}'$, for some $\sty{U}'$ such that $\reduceexternal{\sty{U}}{\sty{U}'}{\pp}{\ell}{\basetype}$.
  \end{enumerate}
\end{lemma}
\begin{proof}
  The proof of (1) is by induction on the derivation of $\reduceinternal{\sty{T}\cdot\sty{T}'}{\sty{U}}{\pp}{\ell}{\basetype}$.
  For rule \RuleName{$\reduceinternalsymbol$-base}, $\Sends{\pp}(\sty{T})$ tells us that $\sty{T}$ is not $\ltypeend$, and hence that $\sty{T}$ is an internal choice on $\pp$, and the result follows immediately.
  For rule \RuleName{$\reduceinternalsymbol$-$\oplus$}, $\Sends{\pp}(\sty{T})$ can only be by rule \RuleName{Sends-$\oplus$}, and the result then follows from the inductive hypothesis.
  The derivation cannot be by rule \RuleName{$\reduceinternalsymbol$-$\&$}, because this would contradict $\Sends{\pp}(\sty{T})$.
  Finally, for rule \RuleName{$\reduceinternalsymbol$-rec}, the result follows from the inductive hypothesis, using \cref{inductive-predicates-unfold} to acquire the required instance of $\Sends{\pp}$.

  The proof of (2) is similar.
\end{proof}

\begin{lemma}
  Session types over $\Theta$ form a preordered monoid; the multiplication is $\cdot$, the unit is $\ltypeend$, and the order is $\subtype_\Theta$.
\end{lemma}
\begin{proof}
  Session subtyping is a preorder by \cref{admissible-subtyping}.
  We have $\ltypeend \cdot \sty{T} = \sty{T}$ by definition, while $\sty{T} \cdot \ltypeend = \sty{T}$ and $(\sty{T}\cdot\sty{T'})\cdot\sty{T''} = \sty{T} \cdot (\sty{T'} \cdot \sty{T''})$ hold by a trivial induction on $\sty{T}$.
  We also have $\sty{T} \cdot \sty{T'} \subtype_\Theta \sty{T} \cdot\sty{U'}$ when $\sty{T} \subtype_\Theta \sty{U}$, again by induction on $\sty{T}$, using the congruence rules of \cref{admissible-subtyping}.
  Finally, we have that $\sty{T} \subtype_\Theta \sty{U}$ implies $\sty{T} \cdot \sty{T'} \subtype_\Theta \sty{U} \cdot\sty{T'}$, we consider two cases.
  If $\unfold{\sty{T}} = \ltypeend$, then $\unfold{\sty{U}} = \ltypeend$, and hence $\unfold{\sty{T} \cdot \sty{T'}} = \unfold{\sty{T'}} = \unfold{\sty{U} \cdot\sty{T'}}$. Otherwise, we have
  $\unfold{\sty{T} \cdot \sty{T'}} = 
  \unfold{\sty{T}} \cdot \sty{T'} \subtype_\Theta
  \unfold{\sty{U}} \cdot \sty{T'} = 
  \unfold{\sty{U} \cdot \sty{T'}}$, the middle inequality following from \cref{reduce-multiplication}.
\end{proof}

\subsection{Proofs for \cref{sec:graded-calculus}}

\begin{lemma}[Subtyping for configurations]\label{configuration-subsumption}
  If $\conftypedg{\config{\queue{\rho}}{t}{\queue{\sigma}}}{\basetype}{\sty{T}}$ and $\sty{T} \subtype \sty{U}$, then $\conftypedg{\config{\queue{\rho}}{t}{\queue{\sigma}}}{\basetype}{\sty{U}}$.
\end{lemma}
\begin{proof}
  By induction on $\conftypedg{\config{\queue{\rho}}{t}{\queue{\sigma}}}{\basetype}{\sty{T}}$.
  The base case is trivial using the subtyping rule for computations, while the two inductive cases are immediate from \cref{subtype-lifting} and the inductive hypothesis.
\end{proof}

\begin{lemma}[Inversion]\label{inversion-lemma}\
  \begin{enumerate}
    \item If $\comptypedg{\Theta}{\Phi}{\Gamma}{\returnterm{v}}{\basetype}{\sty{U}}$ then $\valuetypedg{\Gamma}{v}{\basetype}$ and $\ltypeend \subtype_\Theta h$.
    \item If $\comptypedg{\Theta}{\Phi}{\Gamma}{\bindterm{x}{t}{u}}{\basetype}{\sty{U}}$ then
            there exists $\basetype'$, $\sty{U}'$, $\sty{U}''$ such that
            $\comptypedg{\Theta}{\Phi}{\Gamma}{t}{\basetype'}{\sty{U}'}$, 
            $\comptypedg{\Theta}{\Phi}{\Gamma}{t}{\basetype}{\sty{U}''}$, 
            and $\sty{U}' \cdot \sty{U}'' \subtype_\Theta \sty{U}$.
    \item If $\comptypedg{\Theta}{\Phi}{\Gamma}{\addterm{v}{w}}{\basetype}{\sty{U}}$ then
            $\valuetypedg{\Gamma}{v}{\inttype}$,
            $\valuetypedg{\Gamma}{w}{\inttype}$,
            $\basetype = \inttype$, and $\ltypeend \subtype_\Theta \sty{U}$.
    \item If $\comptypedg{\Theta}{\Phi}{\Gamma}{\lessterm{v}{w}}{\basetype}{\sty{U}}$ then
            $\valuetypedg{\Gamma}{v}{\inttype}$,
            $\valuetypedg{\Gamma}{w}{\inttype}$,
            $\basetype = \booltype$ and $\ltypeend \subtype_\Theta \sty{U}$.
    \item If $\comptypedg{\Theta}{\Phi}{\Gamma}{\ifterm{v}{t_1}{t_2}}{\basetype}{\sty{U}}$,
                    then $\valuetypedg{\Gamma}{v}{\booltype}$,
                    $\comptypedg{\Theta}{\Phi}{\Gamma}{t_1}{\basetype}{\sty{U}}$,
                    $\comptypedg{\Theta}{\Phi}{\Gamma}{t_2}{\basetype}{\sty{U}}$.
    \item If $\comptypedg{\Theta}{\Phi}{\Gamma}{\sendterm{\ell}{v}{\pp}{t}}{\basetype}{\sty{U}}$ then there exists $\basetype', \sty{T}$ such that $\valuetypedg{\Gamma}{v}{\basetype'}$, $\comptypedg{\Theta}{\Phi}{\Gamma}{t}{\basetype}{\sty{T}}$ and $(\ltypesendone{\pp}{\ell}{\basetype'}{\sty{T}}) \subtype_\Theta \sty{U}$.
    \item If $\comptypedg{\Theta}{\Phi}{\Gamma}{\recvterm{\pp}{\ell_i}{x_i}{t_i}{i \in I}}{\basetype}{\sty{U}}$ then there exist $\{\basetype_i\}_{i \in I}$ and $\{\sty{T}_i\}_{i \in I}$ such that $\comptypedg{\Theta}{\Phi}{\Gamma, x_i : \basetype_i}{t_i}{\basetype}{\sty{T}_i}$ for all $i \in I$, and $(\ltyperecv{\pp}{\ell_i}{\basetype_i}{\sty{T}_i}{i \in I}) \subtype_\Theta \sty{U}$.
    \item If $\comptypedg{\Theta}{\Phi}{\Gamma}{\letrecterm{f}{(x_1, \dots, x_n)}{t}{u}}{\basetype}{\sty{U}}$, then there exist $\sty{T}$, $(\basetype_1, \dots, \basetype_n)$, $\basetype'$, and $\ltypevar{X}$ such that
      \begin{gather*}
        \comptypedg{\Theta, \ltypevar{X}}{\Phi, f : (\basetype_1, \dots, \basetype_n) \gto{\sty{\ltypevar{X}}} \basetype'}{\Gamma, x_1 : \basetype_1, \dots, x_n : \basetype_n}{t}{\basetype'}{\sty{T}}
        \\
        \comptypedg{\Theta}{\Phi, f : (\basetype_1, \dots, \basetype_n) \gto{\ltyperec{\ltypevar{X}}{\sty{T}}} \basetype'}{\Gamma}{u}{\basetype}{\sty{U}}
      \end{gather*}
     \item If $\comptypedg{\Theta}{\Phi}{\Gamma}{\applyterm{f}{(v_1, \dots, v_n)}}{\basetype}{\sty{U}}$, then there exist $\sty{T}$, and $\basetype_1, \dots, \basetype_n$ such that
       $\valuetypedg{\Gamma}{v_i}{\basetype_i}$ for all $i$, $(f : (x_1 : \basetype_1, \dots, x_n : \basetype_n) \gto{\sty{T}} \basetype') \in \Phi$, and $\sty{T} \subtype_\Theta \sty{U}$.
  \end{enumerate}
\end{lemma}
\begin{proof}
  Each case is an easy induction on the typing derivation.
\end{proof}

\begin{lemma}\label{subject-reduction-computations}
  Suppose that $\comptypedg{\emptycontext}{\emptycontext}{\emptycontext}{t}{\basetype}{\sty{T}}$.
  \begin{enumerate}
    \item If $\reduceaction{t}{u}{\tauaction}$, then $\comptypedg{\Theta}{\Phi}{\Gamma}{u}{\basetype}{\sty{T}}$.
    \item If $\reduceaction{t}{u}{\sendaction{\pp}{\msg{\ell}{v}}}$ with $v : \basetype$, then $\comptypedg{\Theta}{\Phi}{\Gamma}{u}{\basetype}{\sty{U}}$ for some $\sty{U}$ such that $\reduceinternal{\sty{T}}{\sty{U}}{\pp}{\ell}{\basetype'}$.
    \item If $\reduceaction{t}{u}{\recvaction{\pp}{\msg{\ell}{v}}}$ with $v : \basetype$, then $\Recvs{\pp}(\sty{T})$, and $\comptypedg{\Theta}{\Phi}{\Gamma}{u}{\basetype}{\sty{U}}$ for all $\sty{U}$ such that $\reduceexternal{\sty{T}}{\sty{U}}{\pp}{\ell}{\basetype'}$.
  \end{enumerate}
\end{lemma}
\begin{proof}
  We split the proof into three cases each of which is by induction on $\reduceaction{t}{u}{\alpha}$.
  \begin{itemize}
    \item For $\alpha = \tauaction$ the \RuleName{LetR} case is immediate from the fact that $\ltypeend \subtype_\Theta \sty{T}$ implies $\unfold{\sty{T}} = \ltypeend$, and hence that $\sty{T} \cdot \sty{U} \subtype \ltypeend \cdot \sty{U} = \sty{U}$.
      The \RuleName{IfT}, \RuleName{IfF}, \RuleName{Sub}, \RuleName{LeT}, and \RuleName{LeF} cases are all trivial.
      The \RuleName{Apply} and \RuleName{Cong} cases are both easy inductions on $\rctx{R}$.
    \item For $\alpha = \sendaction{\pp}{\msg{\ell}{v}}$, the \RuleName{Send} case is immediate from the definition of subtyping.
      For the \RuleName{Cong} case, we proceed by induction on $\rctx{R}$. In the base case, and in the case of a recursive function definition, the result is trivial.
      In the case where $\rctx{R}[\hole] = \bindterm{x}{\rctx{R}'[\hole]}{u}$, we obtain $\sty{T}' \cdot \sty{T''} \subtype_\Theta \sty{T}$, and, by the inductive hypothesis, that $\reduceinternal{\sty{T}'}{\sty{U'}}{\pp}{\ell}{\basetype}$, where $\sty{U}'$ is the session type of the reduct of $\rctx{R}'[t]$.
      The result then follows from \cref{subtype-lifting} and \cref{reduce-multiplication}.
    \item For $\alpha = \recvaction{\pp}{\msg{\ell}{v}}$, the \RuleName{Recv} case is immediate from \cref{subtype-lifting}.
      For the \RuleName{Cong} case, we proceed by induction on $\rctx{R}$. In the base case, and in the case of a recursive function definition, the result is trivial.
      In the case where $\rctx{R}[\hole] = \bindterm{x}{\rctx{R}'[\hole]}{u}$, we obtain $\sty{T}' \cdot \sty{T''} \subtype_\Theta \sty{T}$, and, by the inductive hypothesis, that $\Recvs{\pp}(\sty{T}')$, and that for all $\sty{U'}$ such that $\reduceexternal{\sty{T}'}{\sty{U'}}{\pp}{\ell}{\basetype}$, we can assign the session type $\sty{U}'$ to the reduct of $\rctx{R}'[t]$.
      \Cref{reduce-multiplication} and \cref{subtype-lifting} then imply $\Recvs{\pp}(\sty{T})$, and for every $\sty{U}$ such that $\reduceexternal{\sty{T}}{\sty{U}}{\pp}{\ell}{\basetype}$, we obtain $\reduceexternal{\sty{T}'}{\sty{U}'}{\pp}{\ell}{\basetype}$ and $\sty{U}' \cdot \sty{T}'' \subtype \sty{U}$ by \cref{subtype-lifting} and \cref{reduce-multiplication-inverse}.
      The result then follows from \cref{subtype-lifting} and \cref{reduce-multiplication}.
  \end{itemize}
\end{proof}

\begin{lemma}\label{config-recv-type-lemma}
  Assume that $\conftypedg{\config{\queue{\rho}}{t}{\queue{\sigma}}}{\basetype}{\sty{T}}$, that $\reduceexternal{\sty{T}}{\sty{U}}{\pp}{\ell}{\basetype'}$, that $\reduceaction{t}{u}{\recvaction{\pp}{\msg{\ell}{v}}}$ with $v : \basetype'$, and that $\rho$ contains no messages from $\pp$.
  Then $\conftypedg{\config{\queue{\rho}}{u}{\queue{\sigma}}}{\basetype}{\sty{U}}$.
\end{lemma}
\begin{proof}
  By induction on the derivation of $\conftypedg{\config{\queue{\rho}}{t}{\queue{\sigma}}}{\basetype}{\sty{U}}$.
  \begin{itemize}
    \item If the derivation is by rule \RuleName{CBase}, then the result is immediate from \cref{subject-reduction-computations}.
    \item If the derivation is by rule \RuleName{CSend}, then we have $\queue{\sigma} = \consback{\queue{\sigma}'}{\pp}{\msg{\ell}{v}}$, $\conftypedg{\config{\queue{\rho}}{t}{\queue{\sigma'}}}{\basetype}{\sty{T}'}$, and $\reduceinternal{\sty{T}}{\sty{T}'}{\pp}{\ell}{\basetype}$, where $v : \basetype$. The result therefore follows from the inductive hypothesis and \cref{inductive-predicates-swapping}.
    \item If the derivation is by rule \RuleName{CRecv}, then we have $\queue{\rho} = \consfront{\pq}{\msg{\ell}{v}}{\queue{\rho}'}$, $\pq \neq \pp$, $\conftypedg{\config{\queue{\rho'}}{t}{\queue{\sigma}}}{\basetype}{\sty{T}'}$, and $\reduceexternal{\sty{T}'}{\sty{T}}{\pp}{\ell}{\basetype}$, where $v : \basetype$. The result therefore follows from the inductive hypothesis and \cref{inductive-predicates-swapping}.
  \end{itemize}
\end{proof}

\subjred*
\begin{proof}
  There are five rules by which a configuration can reduce.
  We consider each of them in turn.
  \begin{itemize}
    \item For \RuleName{CInt}, the proof is by induction on $\conftypedg{\config{\rho}{t}{\sigma}}{\basetype}{\sty{T}}$.
      In the base case we use the fact that reduction does not change the type of $t$ (\cref{subject-reduction-computations}). The other cases are all immediate from the inductive hypothesis.
    \item For \RuleName{CProd}, the proof is again by induction on $\conftypedg{\config{\rho}{t}{\sigma}}{\basetype}{\sty{T}}$.
      In the base case, we need to show that $\comptypedg{\emptycontext}{\emptycontext}{\emptycontext}{t}{\basetype}{\sty{T}'}$ implies $\conftypedg{\config{\queueempty}{u}{\consfront{\pp}{m}{\queueempty}}}{\basetype}{\sty{T}'}$.
      This is immediate from \cref{subject-reduction-computations}.
      The inductive cases are then immediate from the inductive hypothesis.
    \item For \RuleName{CRecv}, the proof is again by induction on $\conftypedg{\config{\rho}{t}{\sigma}}{\basetype}{\sty{T}}$.
      The derivation cannot be by \RuleName{CBase}.
      It could be by \RuleName{CRecv}, in which case we have $\config{\rho}{t}{\sigma} = \config{\consfront{\pq}{\msg{\ell}{v}}{\queue{\rho'}}}{t}{\queue{\sigma}}$,
      where $\conftypedg{\config{\queue{\rho'}}{t}{\queue{\sigma}}}{\basetype}{\sty{U}}$ for some $\sty{U}$ such that $\reduceexternal{\sty{U}}{\sty{T}}{\pq}{\ell}{\basetype}$.
      There are two cases to consider.
      If $\pp = \pq$ and $\queue{\rho'}$ contains no messages from $\pp$, then the result follows from \cref{config-recv-type-lemma}.
      Otherwise, the result is immediate from the indcutive hypothesis.
      Finally, the derivation of $\conftypedg{\config{\rho}{t}{\sigma}}{\basetype}{\sty{T}}$ could be by \RuleName{CSend}, in which case the result is immediate from the inductive hypothesis.
    \item For \RuleName{CSend}, we proceed by induction on the derivation of $\conftypedg{\config{\rho}{t}{\sigma}}{\basetype}{\sty{T}}$.
      This derivation cannot be by rule \TRULE{CBase}.
      It could be by \RuleName{CSend} on some participant $\pq$, in which case there are two cases to consider.
      If $\pp = \pq$, then the result is immediate.
      If $\pp \neq \pq$, then the result follows from the inductive hypothesis and \cref{inductive-predicates-swapping}.
      Finally, if it is by \RuleName{CRecv}, the result follows from the inductive hypothesis and \cref{inductive-predicates-swapping}.
    \item For \RuleName{CRecv}, the result is immediate from the \TRULE{CRecv} configuration typing rule.
  \end{itemize}
\end{proof}

\begin{lemma}\label{union-bisimulation}
  Let $(S, \reducesymbol, \Result, \hastype{}{})$ and $(S', \reducesymbol, \Result, \hastype{}{})$ be typed transition systems.
  If $\mathcal{R}$ is a set of typed bisimulations, then the union $(\bigcup \mathcal{R})_{\sty{S}} = \bigcup_{R \in \mathcal{R}} R_{\sty{S}}$ is a typed bisimulation.
  Hence $\bisim_{\sty{S}}$ is the union
  \[
    \bigcup \{R_{\sty{S}} \mid R~\text{is a typed bisimulation}\}
  \]
\end{lemma}
\begin{proof}
  It is easy to verify that $\bigcup \mathcal{R}$ satisfies each of the conditions of \cref{def:bisim}.
  Clearly taking the union over all typed bisimulations gives the largest such, namely $\bisim$.
\end{proof}

\begin{lemma}\label{reducemerge}
  Let $\sty{T}$ be a session type over $\Theta$, and let $\{\sty{U}\}_{k \in K}$ be a family of session types over $\Theta$, indexed by a non-empty finite set $K$.
  If $\reduceinternal{\sty{T}}{\sty{U}_k}{\pp}{\ell}{\basetype}$ for every $k \in K$, then then there is some $\sty{U}'$ such that  $\sty{U}'_k \subtype_\Theta \sty{U}$ for all $k \in K$ and $\reduceinternal{\sty{T}}{\sty{U}'}{\pp}{\ell}{\basetype}$.
\end{lemma}
\begin{proof}
  Let $k_0$ be any element of $K$.
  We proceed by induction on the derivation of $\reduceinternal{\sty{T}}{\sty{U}_{k_0}}{\pp}{\ell}{\basetype}$.
  \begin{itemize}
    \item If the derivation is by rule \RuleName{$\reduceinternalsymbol$-base}, then so is the derivation of every other $\reduceinternal{\sty{T}}{\sty{U}_k}{\pp}{\ell}{\basetype}$. We have $\sty{U}_{k_0} = \sty{U}_{k}$ for all $k \in K$, so we can conclude by reflexivity of $\subtype_{\Theta}$.
    \item If the derivation is by rule \RuleName{$\reduceinternalsymbol$-$\oplus$}, then so is the derivation of every other $\reduceinternal{\sty{T}}{\sty{U}_k}{\pp}{\ell}{\basetype}$.
      We have
      \begin{gather*}
        \sty{T} = \ltypesend{\pq}{\ell_i}{\basetype_i}{\sty{T}_i}{i \in I}
        \qquad
        \sty{U}_k = \ltypesend{\pq}{\ell_i}{\basetype_i}{\sty{U}_{ki}}{i \in I_k}
        \\
        \pp \neq \pq
        \qquad
        I_k \subseteq I
        \qquad
        \reduceinternal{\sty{T}_i}{\sty{U}_{ki}}{\pp}{\ell}{\basetype}~\text{for all}~k \in K, i \in I_k
      \end{gather*}
      Define $I' = \bigcup_{k \in K} I_k$.
      For every $i \in I'$, let $\sty{U}'_i$ be any session type such that $\sty{U}_{ki} \subtype_\Theta \sty{U}'_i$ for all $k$ with $i \in I_k$, and such that $\reduceinternal{\sty{T}_i}{\sty{U}'_i}{\pp}{\ell}{\basetype}$.
      These exist by the inductive hypothesis.
      Define $\sty{U}' = \ltypesend{\pq}{\ell_i}{\basetype_i}{\sty{U}'_i}{i \in I'}$.
      We have $\sty{U}_k \subtype_\Theta \sty{U}'$ by \cref{admissible-subtyping}, and $\reduceinternal{\sty{T}}{\sty{U}'}{\pp}{\ell}{\basetype}$ by \RuleName{$\reduceinternalsymbol$-$\oplus$}.
    \item If the derivation is by rule \RuleName{$\reduceinternalsymbol$-$\&$}, then so is the derivation of every other $\reduceinternal{\sty{T}}{\sty{U}_k}{\pp}{\ell}{\basetype}$.
      We have
      \begin{gather*}
        \sty{T} = \ltyperecv{\pq}{\ell_i}{\basetype_i}{\sty{T}_i}{i \in I}
        \qquad
        \sty{U}_k = \ltyperecv{\pq}{\ell_i}{\basetype_i}{\sty{U}_{ki}}{i \in I}
        \\
        \reduceinternal{\sty{T}_i}{\sty{U}_{ki}}{\pp}{\ell}{\basetype}~\text{for all}~k \in K, i \in I
      \end{gather*}
      For every $i \in I$, let $\sty{U}'_i$ be any session type such that $\sty{U}_{ki} \subtype_\Theta \sty{U}'_i$ for all $k \in K$, and such that $\reduceinternal{\sty{T}_i}{\sty{U}'_i}{\pp}{\ell}{\basetype}$.
      These exist by the inductive hypothesis.
      Define $\sty{U}' = \ltyperecv{\pq}{\ell_i}{\basetype_i}{\sty{U}'_i}{i \in I'}$.
      We have $\sty{U}_k \subtype_\Theta \sty{U}'$ by \cref{admissible-subtyping}, and $\reduceinternal{\sty{T}}{\sty{U}'}{\pp}{\ell}{\basetype}$ by \RuleName{$\reduceinternalsymbol$-$\&$}.
  \end{itemize}
\end{proof}

\bisimper*
\begin{proof}
  We first show closure under subtyping, i.e.\ that if $s \bisim_{\sty{S}} s'$ and $\sty{S} \subtype \sty{S}'$ then $s \bisim_{\sty{S}'} s'$.
  To show $s \bisim_{\sty{S}'} s'$, we consider each of the case of \cref{def:bisim} in turn.
  Case (1) is trivial.
  For (2), the result is immediate from the fact that $\unfold{\sty{S}} = \ltypeend$ implies $\unfold{\sty{S}'} = \ltypeend$.
  For (3), we use the fact that $\Sends{\pp}(S')$ implies $\Sends{\pp}(S)$, and $\reduceinternal{\sty{S}}{\sty{T}}{\pp}{\ell}{\basetype}$ implies there is some $\sty{T}'$ such that $\reduceinternal{\sty{S}'}{\sty{T}'}{\pp}{\ell}{\basetype}$ and $\sty{T} \subtype \sty{T}'$; these are instances of \cref{subtype-lifting}.
  For (4), we use fact that $\reduceexternal{\sty{S}'}{\sty{T}'}{\pp}{\ell}{\basetype}$ implies there is some $\sty{T}$ such that $\reduceexternal{\sty{S}}{\sty{T}}{\pp}{\ell}{\basetype}$ and $\sty{T} \subtype \sty{T}'$; this is also an instance of \cref{subtype-lifting}.

  Symmetry of $\bisim_{\sty{T}}$ is trivial.
  Transitivity is mostly also easy to prove; it is analogous to the usual proof of transitivity of bisimulation; the only difficult case is (3).
  In this case, if $s \bisim_{\sty{T}} s' \bisim_{\sty{T}} s''$, we obtain $\reduceactionmany{s}{t}{\sendaction{\pp}{m}}$, $\reduceactionmany{s'}{t'}{\sendaction{\pp}{m}}$, and $\reduceactionmany{s''}{t''}{\sendaction{\pp}{m}}$, where $t \bisim_{\sty{T}} t' \bisim_{\sty{T}'} t''$, and we have both $\reduceinternal{\sty{S}}{\sty{T}}{\pp}{\ell}{\basetype}$ and $\reduceinternal{\sty{S}}{\sty{T}'}{\pp}{\ell}{\basetype}$.
  It is enough to show that there is some $\sty{T''}$ such that $\reduceinternal{\sty{S}}{\sty{T}''}{\pp}{\ell}{\basetype}$ and $t \bisim_{\sty{T''}} t' \bisim_{\sty{T''}} t''$; this follows from \cref{reducemerge} and closure of $\bisim$ under subtyping.
\end{proof}
\subsection{Proofs for \cref{sec:computation-trees}}

\begin{definition}\label{tree-pre-typing}
A \emph{tree pre-typing relation} is a relation $R$ between computation trees and closed session types, such that the following hold when $t \,R\, \sty{T}$.
  \begin{enumerate}
    \item If $t = \treesend{\role{p}}{\msg{\ell}{v}}{u}$, with $v : \basetype$, then there is some $\sty{U}$ such that $\reduceinternal{\sty{T}}{\sty{U}}{\pp}{\ell}{\basetype}$ and $u \,R\, \sty{U}$.
    \item If $t = \treerecv{\role{p}}{t_m}{m \in M}$, then $\Recvs{\pp}(\sty{T})$.
    \item If $\unfold{\sty{T}} = \ltypesend{\pp}{\ell_i}{\basetype_i}{\sty{T}_i}{i \in I}$, then there exist $m, u$ such that $\reduceaction{t}{u}{\sendaction{\pp}{m}}$.
    \item If $\unfold{\sty{T}} = \ltyperecv{\pp}{\ell_i}{\basetype_i}{\sty{T}_i}{i \in I}$, then there is some natural number $h$ such that, for every $i \in I$ and $v : \basetype_i$, there is some $u \,R\, \sty{T}_i$ such that $\reduceaction{t}{u}{\recvaction{\pp}{\msg{\ell_i}{v}}}$, where the derivation of the latter has height at most $h$.
  \end{enumerate}
\end{definition}

\begin{lemma}\label{union-pre-typing}
  If $\mathcal{R}$ is a set of tree pre-typing relations, then the union $\bigcup \mathcal{R}$ is also a tree pre-typing relation.
  Hence $\hastype{}{}$ is the union
  \[
    \bigcup \{R \mid R~\text{is an tree pre-typing relation}\}
  \]
\end{lemma}
\begin{proof}
  It is easy to verify that $\bigcup \mathcal{R}$ satisfies each of the conditions of \cref{tree-pre-typing}.
  Clearly taking the union over all tree pre-typing relations gives the largest such, namely $\hastype{}{}$.
\end{proof}

\begin{lemma}\label{reduction-swapping}
  Let $s$ be a computation tree.
  \begin{enumerate}
      \item If $\pp \neq \pq$, and $\reduceaction{s}{\reduceaction{t}{u}{\sendaction{\pq}{m'}}}{\sendaction{\pp}{m}}$, then there is some $t'$ such that $\reduceaction{s}{\reduceaction{t'}{u}{\sendaction{\pp}{m}}}{\sendaction{\pq}{m'}}$.
      \item If $\reduceaction{s}{\reduceaction{t}{u}{\recvaction{\pq}{m'}}}{\sendaction{\pp}{m}}$, then there is some $t'$ such that $\reduceaction{s}{\reduceaction{t'}{u}{\sendaction{\pp}{m}}}{\recvaction{\pq}{m'}}$.
      \item If $\reduceaction{s}{t}{\sendaction{\pp}{m}}$ and $\reduceaction{s}{t'}{\recvaction{\pq}{m'}}$, then there is some $u$ such that $\reduceaction{t}{u}{\recvaction{\pq}{m'}}$ and $\reduceaction{t'}{u}{\sendaction{\pp}{m}}$.
      \item Let $M$ be a non-empty set of messages, and let $h$ be a natural number. If $\pp \neq \pq$, $\reduceaction{s}{t_m}{\recvaction{\pp}{m}}$, with a proof of height at most $h$, and $\reduceaction{t_m}{u_m}{\recvaction{\pq}{m'}}$ for all $m \in M$, then there is some $u$ such that $\reduceaction{t}{u}{\recvaction{\pq}{m'}}$ and $\reduceaction{u}{u_m}{\recvaction{\pp}{m}}$ for all $m \in M$.
  \end{enumerate}
\end{lemma}
\begin{proof}
  For (1), the proof is by induction on $\reduceaction{s}{t}{\sendaction{\pp}{m}}$.
  The base case is trivial. For the inductive case, we have $s = \treesend{\pr}{m''}{s'}$ with $\pr \neq \pp$.
  If $\pr = \pq$, then the reduction $\reduceaction{t}{u}{\sendaction{\pq}{m'}}$ is a base case, and the result is immediate.
  If $\pr \neq \pq$, then $\reduceaction{t}{u}{\sendaction{\pq}{m'}}$ is an inductive case, and the result follows from the inductive hypothesis.

  The proof of (2) is similar to the proof of (1).

  For (3), the proof is again by induction on $\reduceaction{s}{t}{\sendaction{\pp}{m}}$.
  The base case is again trivial.
  If the derivation is an inductive case, then so is the derivation of $\reduceaction{s}{t'}{\recvaction{\pq}{m'}}$, and the result follows from the inductive hypothesis.

  For (4), the proof is by well-founded induction on $h$, by inspecting $t$.
  \begin{itemize}
    \item If $s = \treerecv{\pp}{t_m}{m \in M'}$, then we have $M \subseteq M'$, and we can simply take $u = \treerecv{\pp}{u_m}{m \in M}$.
    \item If $s = \treerecv{\pr}{s_{m''}}{m'' \in M''}$, with $\pp \neq \pr$, then we have $\reduceaction{s_{m''}}{t_{m, m''}}{\recvaction{\pp}{m}}$ and $t_m = \treerecv{\pr}{t_{m, m''}}{m'' \in M''_m}$ for some $M''_m \subseteq M$. If $\pr = \pq$, then $m' \in M''$, and we can take $u = s_{m'}$. Otherwise, the reductions $\reduceaction{t_m}{u_m}{\recvaction{\pq}{m'}}$ are all by \RuleName{RecvRecv}, and we can appeal to the inductive hypothesis to obtain $u_{m''} \in M'''$, where $M''' = \bigcap_{m \in M} M''_m$. We can then take $u = \treerecv{\pr}{u_{m''}}{m'' \in M'''}$.
    \item If $s = \treesend{\pr}{m''}{s'}$, then we have $\reduceaction{s'}{u'_m}{\recvaction{\pp}{m}}$ and $\reduceaction{u'_m}{t'_m}{\recvaction{\pq}{m'}}$, where $u_m = \treesend{\pr}{m''}{u'_m}$ and $t_m = \treesend{\pr}{m''}{t'_m}$. We can therefore take $u = \treerecv{\pr}{m''}{u'}$, where $u'$ is obtained from the inductive hypothesis.
  \end{itemize}
\end{proof}

\begin{lemma}\label{computation-tree-typing-lifting}
    Let $t$ be a computation tree such that $\hastype{t}{\sty{T}}$.
    \begin{enumerate}
        \item For each $\pp$, there is at most one triple $(u, \ell, v)$ such that, such that $\reduceaction{t}{u}{\sendaction{\pp}{\msg{\ell}{v}}}$, and this $u$ satisfies $\hastype{s}{\sty{U}}$ for some $\sty{U}$ such that $\reduceinternal{\sty{T}}{\sty{U}}{\pp}{\ell}{\basetype}$, where $v : \basetype$. Moreover, if $\Sends{\pp}(\sty{T})$ holds, then there exists such a $u$.
        \item For each $(\pp, \ell, v)$, and type $\sty{U}$ such that $\reduceexternal{\sty{T}}{\sty{U}}{\pp}{\ell}{\basetype}$, where $v : \basetype$, there is some $\hastype{u}{\sty{U}}$ such that $\reduceaction{t}{u}{\recvaction{\pp}{\msg{\ell}{v}}}$.
    \end{enumerate}
\end{lemma}
\begin{proof}
  For (1), uniqueness of $u$ is a trivial induction on the reduction.
  To prove that a suitable $\sty{U}$ exists, we proceed by induction on the reduction.
  \begin{itemize}
    \item If the reduction is by rule \RuleName{Send}, then this is immediate from the definition of the typing relation.
    \item If the reduction is by rule \RuleName{SendSend}, then we have $t = \treesend{\pq}{\msg{\ell'}{\basetype'}}{t'}$, $\reduceaction{t'}{u'}{\sendaction{\pp}{\msg{\ell}{v}}}$, $u = \treesend{\pq}{\msg{\ell'}{\basetype'}}{u'}$, and $\pp \neq \pq$. By the definition of the typing relation, $t'$ has some type $\sty{T}'$ such that $\reduceinternal{\sty{T}}{\sty{T}'}{\pq}{\ell'}{\basetype'}$, where $v' : \basetype'$. By the inductive hypothesis, $u'$ has some type $\sty{U}'$ such that $\reduceinternal{\sty{T'}}{\sty{U'}}{\pp}{\ell}{\basetype}$. The result therefore follows from \cref{inductive-predicates-swapping}.
  \end{itemize}
  Finally, we show that if $\Sends{\pp}(\sty{T'})$ and $\unfold{\sty{T}} = \subst{T'}{\ltypevar{X}_1 \mapsto \sty{T}_1, \dots}$, then such a $u$ exists.
  We do this by induction on the derivation of $\Sends{\pp}(\sty{T'})$.
  \begin{itemize}
    \item If the derivation is by rule \RuleName{Sends-base}, then the result is immediate from the definition of the typing relation.
    \item If the derivation is by rule \RuleName{Sends-$\oplus$}, then we have $\sty{T'} = \ltypesend{\pq}{\ell_i}{\basetype_i}{\sty{T'}_i}{i \in I}$, with $\pp \neq \pq$.
      By the typing relation, there is some reduction $\reduceaction{t}{t'}{\sendaction{\pq}{\msg{\ell_i}{v}}}$ with $v : \basetype_i$; this $t'$ then has type $\sty{T'}_i$ by the above.
      By the inductive hypothesis, we therefore have some reduction $\reduceaction{t'}{u'}{\sendaction{\pp}{\msg{\ell}{v}}}$; the result then follows from \cref{reduction-swapping}.
    \item If the derivation is by rule \RuleName{Send-rec}, then the result is immediate from the inductive hypothesis (introducing a new type variable).
  \end{itemize}

  For (2), we proceed by induction on the derivation of $\reduceexternal{\sty{T}}{\sty{U}}{\pp}{\ell}{\basetype}$.
  \begin{itemize}
    \item If the derivation is by rule \RuleName{$\reduceexternalsymbol$-base}, then existence of $u$ is immediate from the definition of the typing relation.
    \item If the derivation is by rule \RuleName{$\reduceexternalsymbol$-$\&$}, then we have $\sty{T} = \ltyperecv{\pq}{\ell_i}{\basetype_i}{\sty{T}_i}{i \in I}$, with $\pp \neq \pq$. Define $M = \{(\ell_i, v') \mid i \in I \land v: \basetype_i\}$. For every $m \in M$, there is a reduction $\reduceaction{t}{t'_m}{\recvaction{\pq}{m}}$, where $\hastype{t'_{\msg{\ell_i}{v'}}}{\sty{T}_i}$. By the inductive hypothesis, there are reductions $\reduceaction{t'_m}{u'_m}{\recvaction{\pp}{\msg{\ell}{v}}}$ so the result follows from \cref{reduction-swapping}.
    \item If the derivation is by rule \RuleName{$\reduceexternalsymbol$-$\oplus$}, then we have $\sty{T} = \ltypesend{\pq}{\ell_i}{\basetype_i}{\sty{T}_i}{i \in I}$, with $\pp \neq \pq$. There is a reduction $\reduceaction{t}{t'}{\sendaction{\pq}{m}}$, where $\hastype{t'}{\sty{T}_i}$. By the inductive hypothesis, there is a reduction $\reduceaction{t'}{u'}{\recvaction{\pp}{\msg{\ell}{v}}}$ so the result follows from \cref{reduction-swapping}.
  \end{itemize}
\end{proof}

\begin{lemma}\label{typing-subtyping}
  If $\hastype{t}{\sty{T}}$ and $\sty{T} \subtype \sty{U}$, then $\hastype{t}{\sty{U}}$.
\end{lemma}
\begin{proof}
  To prove $\hastype{t}{\sty{U}}$, points (1) and (2) of the definition of the typing relation hold by \cref{subtype-lifting}.
  For (3), note that if $\unfold{\sty{U}}$ is an internal choice on $\pp$, then $\Sends{\pp}(\sty{T})$ holds by subtyping, and hence $t$ sends a message to $\pp$ by \cref{computation-tree-typing-lifting}.
  For (4), note that if $\unfold{\sty{U}} = \ltyperecv{\pp}{\ell_i}{\basetype_i}{\sty{U}_i}{i \in I}$, then for every $i$ there is some $\sty{T}_i$ such that $\reduceexternal{\sty{T}}{\sty{T}_i}{\pp}{\ell_i}{\basetype_i}$ by subtyping, and hence the required reduction exists by \cref{computation-tree-typing-lifting}.
\end{proof}

\begin{lemma}\label{reduction-unique-bisim}
  Assume that $\reduceaction{t_1}{u_1}{\recvaction{\pp}{m}}$ and that
  and $\reduceaction{t_2}{u_2}{\recvaction{\pp}{m}}$.
  Also assume that $\hastype{u_1}{\sty{U}}$ and $\hastype{u_2}{\sty{U}}$ where $\reduceexternal{\sty{T}}{\sty{U}}{\pp}{\ell}{\basetype}$ and $v : \basetype$.
  If $t_1 \bisim_{\sty{T}} t_2$, then 
  $u_1 \bisim_{\sty{T}} u_2$.
\end{lemma}
\begin{proof}
  Define a relation $R_{\sty{T}}$, for each closed session type $\sty{T}$, inductively as follows.
  \[
    \begin{prooftree}
      \hypo{t_1 \bisim_{\sty{T}} t_2}
      \infer1{(t_1, t_2) \in R}
    \end{prooftree}
    \qquad
    \begin{prooftree}
      \hypo{(t_1, t_2) \in R_{\sty{T}}}
      \hypo{\reduceexternal{\sty{T}}{\sty{U}}{\pp}{\ell}{\basetype}}
      \hypo{v : \basetype}
      \hypo{\reduceaction{t_1}{\hastype{u_1}{\sty{U}}}{\recvaction{\pp}{\msg{\ell}{v}}}}
      \hypo{\reduceaction{t_2}{\hastype{u_2}{\sty{U}}}{\recvaction{\pp}{\msg{\ell}{v}}}}
      \infer5{(u_1, u_2) \in R_{\sty{U}}}
    \end{prooftree}
  \]
  We show that this is a typed bisimilarity.
  Let $i$ and $j$ be the two elements of $\{1, 2\}$, in either order, and assume that $(t_1, t_2) \in R_{\sty{T}}$.
  \begin{itemize}
    \item There are no $\tauaction$-reductions.
    \item If $\Result(t_i) = x$, then we have $\Result(t_j)$, by an induction on the derivation of $(t_1, t_2) \in R_{\sty{T}}$.
    \item If $\reduceaction{t_i}{u_i}{\sendaction{\pq}{m'}}$, then an induction on the derivation of $(t_1, t_2) \in R_{\sty{T}}$ shows that there is also a reduction $\reduceaction{t_j}{u_j}{\sendaction{\pq}{m'}}$.
      By \cref{computation-tree-typing-lifting}, there are session types $\sty{U}_i$ and $\sty{U}_j$ such that $\reduceinternal{\sty{T}}{\sty{U}_k}{\pq}{\ell}{\basetype}$ and $\hastype{u_k}{\sty{U}_k}$ ($k \in \{1, 2\}$). By \cref{reducemerge}, we can thus give $u_i$ and $u_j$ a common type $\sty{U}$ such that $\reduceinternal{\sty{T}}{\sty{U}}{\pq}{\ell}{\basetype}$, and we have $(u_1, u_2) \in R_{\sty{U}}$.
    \item If $\reduceexternal{\sty{T}}{\sty{U}}{\pq}{\ell}{\basetype}$, then for each reduction $\reduceaction{t_i}{u_i}{\recvaction{\pq}{\msg{\ell}{v}}}$, there is a reduction $\reduceaction{t_j}{u_j}{\recvaction{\pq}{\msg{\ell}{v}}}$ such that $\hastype{u_k}{\sty{U}}$ ($k \in \{1, 2\}$), by \cref{computation-tree-typing-lifting}. Thus $(t_k, u_k) \in R_{\sty{U}}$.
  \end{itemize}
\end{proof}

\begin{lemma}\label{bisim-reflexive}
  If $s$ is a computation tree such that $\hastype{s}{\sty{S}}$, then $s \bisim_{\sty{S}} s$.
\end{lemma}
\begin{proof}
  We prove each case of \cref{def:bisim} in turn.
  Most of them are trivial, the only slightly non-trivial case is (2), where we need to provide a closed session type $\sty{T}$ such that $\reduceinternal{\sty{S}}{\sty{T}}{\pp}{\ell}{\basetype}$.
  This session type exists by the definition of $\hastype{s}{\sty{S}}$.
\end{proof}

\begin{lemma}\label{treegmonad-hastype}
    If $t \in \treegmonad(X)_{\sty{T}}$, then $\hastype{t}{\sty{T}}$.
\end{lemma}
\begin{proof}
    We prove each case of \cref{def:computation-tree-typing} in turn.
    \begin{enumerate}
        \item For (1), $t = \treesend{\pp}{\msg{\ell}{v}}{u}$ implies $\unfold{\sty{T}}$ has the form $\ltypesend{\pp}{\ell_i}{\basetype_i}{\sty{T}_i}{i \in I}$, with $\ell = \ell_i$ and $v : \basetype_i$, and $u \in \treegmonad(X)_{\sty{T}_i}$ for some $i$. Since $\reduceinternal{\sty{T}}{\sty{T}_i}{\pp}{\ell_i}{\basetype_i}$, we can therefore take $\sty{U} = \sty{T}_i$.
        \item For (2), $t = \treerecv{\pp}{t_m}{m \in M}$ implies $\unfold{\sty{T}}$ is an external choice on $\pp$, so $\Recvs{\pp}(\unfold{\sty{T}})$ trivially, and hence $\Recvs{\pp}(\sty{T})$ by \cref{inductive-predicates-unfold}.
        \item For (3), $\unfold{\sty{T}} = \ltypesend{\pp}{\ell_i}{\basetype_i}{\sty{T}_i}{i \in I}$ implies that $t$ has the form $\treesend{\pp}{m}{u}$, so that $\reduceaction{t}{u}{\sendaction{\pp}{m}}$.
        \item For (4), $\unfold{\sty{T}} = \ltyperecv{\pp}{\ell_i}{\basetype_i}{\sty{T}_i}{i \in I}$ implies that $t = \treerecv{\pp}{u_m}{m \in M}$, where $M = \{(\ell_i, v) \mid i \in I \land v : \basetype_i\}$ and $u_{(\ell_i, v)} \in \treegmonad(X)_{\sty{T}_i}$. We have $\reduceaction{t}{u_m}{\recvaction{\pp}{m}}$.
    \end{enumerate}
\end{proof}

\begin{lemma}\label{bisim-unique}
  If $t, u \in \treegmonad(X)_{\sty{T}}$ and $t \bisim_{\sty{T}} u$, then $t = u$.
\end{lemma}
\begin{proof}
  We prove $t = u$, using the coinduction principle for computation trees, by inspecting $\unfold{\sty{T}}$.
  \begin{itemize}
    \item If $\unfold{\sty{T}} = \ltypeend$, then $t$ has the form $\treereturn{x}$, and $u$ has the form $\treereturn{y}$.
      Since $t \bisim_{\sty{T}} u$, we have $x = y$.
    \item If $\unfold{\sty{T}} = \ltypesend{\pp}{\ell_i}{\basetype_i}{\sty{T}_i}{i \in I}$, then $t$ has the form $\treesend{\pp}{\msg{\ell_i}{v}}{t'}$, where $t' \in \treegmonad(X)_{\sty{T}_i}$, and $u$ has the form $\treesend{\pp}{{\msg{\ell_j}{w}}}{u'}$, where $u' \in \treegmonad(X)_{\sty{T}_j}$.
      Since $t \bisim_{\sty{T}} u$, we have $i = j$ and $v = w$, and $t' \bisim_{\sty{T}'} u'$, for some $\sty{T}'$ such that $\reduceinternal{\sty{T}}{\sty{T}'}{\pp}{\ell_i}{\basetype_i}$. We necessarily have $\sty{T}' = \sty{T}_i = \sty{T}_j$, so $t', u' \in \treegmonad(X)_{\sty{T}'}$.
    \item If $\unfold{\sty{T}} = \ltyperecv{\pp}{\ell_i}{\basetype_i}{\sty{T}_i}{i \in I}$, then define $M = \{\msg{\ell_i}{v} \mid i \in I \land v : \basetype_i\}$. We have that $t$ has the form $\treerecv{\pp}{t'_m}{m \in M}$, and $u$ has the form $\treerecv{\pp}{u'_m}{m \in M}$, where $t'_{(\ell_i, v)}, u'_{(\ell_i, v)} \in \treegmonad(X)_{\sty{T}_i}$ for each $(\ell_i, v) \in M$.
      To conclude it is therefore enough to show that $t'_{(\ell_i, v)} \bisim_{\sty{T}_i} u'_{(\ell_i, v)}$.
      This follows from $t \bisim_{\sty{T}} t$ because $\reduceexternal{\sty{T}}{\sty{T}_i}{\pp}{\ell_i}{\basetype_i}$, and $\hastype{t'_{(\ell_i, v)}}{\sty{T}_i}$ by \cref{treegmonad-hastype}.
  \end{itemize}
\end{proof}

We define, for each closed session type $\sty{T}$ and $\bisim_{\sty{T}}$-equivalence class $t$ of computation trees, a computation tree $\nftree{\sty{T}}{t} \in \treegmonad(X)_{\sty{T}}$ as follows
  \[
    \nftree{\sty{T}}{t}
    =
    \begin{cases}
        \treereturn{\Result(t)}
        &\text{if}~\unfold{\sty{T}} = \ltypeend
        \\
        \treesend{\pp}{\msg{\ell_i}{v}}{\nftree{\sty{T}_i}{u}}
        &\text{if}~\unfold{\sty{T}} = \ltypesend{\pp}{\ell_i}{\basetype_i}{\sty{T}_i}{i \in I}
        \land
        \reduceaction{t}{u}{\sendaction{\pp}{\msg{\ell_i}{v}}}
        \\
        \treerecv{\pp}{\nftree{\sty{T}_i}{u_{\msg{\ell_i}{v}}}}{\msg{\ell_i}{v} \in M}
        &\text{if}~\unfold{\sty{T}} = \ltyperecv{\pp}{\ell_i}{\basetype_i}{\sty{T}_i}{i \in I}
        \land
        \forall m\in M.\,\reduceaction{t}{u_m}{\recvaction{\pp}{m}}
    \end{cases}
  \]
where, in the final case, we take $M = \{\msg{\ell_i}{v} \mid i \in I \land v : \basetype_i\}$, and we require $\hastype{u_{\msg{\ell_i}{v}}}{\sty{T}_i}$.
To see that this is well-defined, note that $\Result(t)$ is defined in the first case, $u$ exists and is unique in the second case, and $u_m$ exists and is unique in the third case, by \cref{computation-tree-typing-lifting} and \cref{reduction-unique-bisim}.

\begin{lemma}\label{reduce-send-nf}
  If $\Sends{\pp}(\sty{T})$ and $\hastype{t}{\sty{T}}$, then there exists a $u$ of some type $\sty{U}$, such that
  \[
    \reduceaction{t}{u}{\sendaction{\pp}{\msg{\ell}{v}}}
    \quad
    \reduceaction{\nftree{\sty{T}}{t}}{\nftree{\sty{U}}{u}}{\sendaction{\pp}{\msg{\ell}{v}}}
    \quad
    \reduceinternal{\sty{T}}{\sty{U}}{\pp}{\ell}{\basetype}
  \]
  where $v : \basetype$.
\end{lemma}
\begin{proof}
  We weaken the assumption $\Sends{\pp}(\sty{T})$ to $\unfold{\sty{T}} = \subst{\sty{T'}}{\ltypevar{X}_1 \mapsto \sty{T}_1, \dots}$ for some $\sty{T'}$ such that $\Sends{\pp}(\sty{T'})$, and then proceed by induction on the derivation of $\Sends{\pp}(\sty{T'})$.
  \begin{itemize}
    \item If the derivation is by rule \RuleName{Sends-base}, then the result is immediate from the definition of the typing relation.
    \item If the derivation is by rule \RuleName{Sends-$\oplus$}, then we have $\sty{T'} = \ltypesend{\pq}{\ell_i}{\basetype_i}{\sty{T'}_i}{i \in I}$, with $\pp \neq \pq$.
      By the typing relation, there is some reduction $\reduceaction{t}{t'}{\sendaction{\pq}{\msg{\ell_i}{v}}}$ with $v : \basetype_i$; this $t'$ then has type $\sty{T'}_i$ by the above.
      By the inductive hypothesis, we therefore have some reduction $\reduceaction{t'}{u'}{\sendaction{\pp}{\msg{\ell}{v}}}$; the result then follows from \cref{reduction-swapping}.
    \item If the derivation is by rule \RuleName{Send-rec}, then the result is immediate from the inductive hypothesis (introducing a new type variable).
  \end{itemize}
\end{proof}

\begin{lemma}\label{reduce-recv-nf}
  If $\reduceinternal{\sty{T}}{\sty{U}}{\pp}{\ell}{\basetype}$ and $\hastype{t}{\sty{T}}$, then there exists a $u$ of some type $\sty{U}$, such that
  \[
    \reduceaction{t}{u}{\recvaction{\pp}{\msg{\ell}{v}}}
    \quad
    \reduceaction{\nftree{\sty{T}}{t}}{\nftree{\sty{U}}{u}}{\recvaction{\pp}{\msg{\ell}{v}}}
  \]
  where $v : \basetype$.
\end{lemma}
\begin{proof}
  By induction on the derivation of $\reduceinternal{\sty{T}}{\sty{U}}{\pp}{\ell}{\basetype}$.
  \begin{itemize}
    \item If the proof is by \RuleName{$\reduceinternalsymbol$-base}, then the result is trivial.
    \item If the proof is by \RuleName{$\reduceinternalsymbol$-$\&$}, then the required reduction of $t$ follows from \cref{reduction-swapping}, while the required reduction of $\nftree{\sty{T}}{u}$ exists by \RuleName{RecvRecv}.
    \item If the proof is by \RuleName{$\reduceinternalsymbol$-$\oplus$}, then again the required reduction of $t$ follows from \cref{reduction-swapping}, while the required reduction of $\nftree{\sty{T}}{u}$ exists by \RuleName{RecvSend}.
    \item If the proof is by \RuleName{$\reduceinternalsymbol$-rec}, then the result is immediate from the inductive hypothesis.
  \end{itemize}
\end{proof}

\bisimnormal*
\begin{proof}
  The first part of the lemma is \cref{treegmonad-hastype}.
  For the second part of the lemma, uniqueness of $u$ is immediate from \cref{bisim-unique}.
  To prove existence, we show that $t \bisim_{\sty{T}} \nftree{\sty{T}}{t}$.
  We check each case of the definition of bisimilarity in turn.
  \begin{enumerate}
    \item There are no $\tauaction$ transitions.
    \item If $\unfold{\sty{T}} = \ltypeend$, then the result is immediate from the definition of $\nftree{\sty{T}}{t}$.
    \item If $\Sends{\pp}(\sty{T})$, then the appropriate reductions exist by \cref{reduce-send-nf}, noting that, by \cref{computation-tree-typing-lifting}, reduction by sending is deterministic.
    \item If $\reduceinternal{\sty{T}}{\sty{U}}{\pp}{\ell}{v}$, then the required reductions exist by \cref{reduce-recv-nf}, noting that, by \cref{reduction-unique-bisim}, reduction by receiving is deterministic up to typed bisimilarity.
  \end{enumerate}
\end{proof}

\begin{lemma}\label{reduce-bind}
  If $\reduceaction{t}{u}{\alpha}$, then $\reduceaction{t \bind f}{u \bind f}{\alpha}$.
\end{lemma}
\begin{proof}
  This is a trivial induction on the proof of $\reduceaction{t}{u}{\alpha}$.
\end{proof}

\begin{lemma}\label{union-normal-forms}
  Let $\mathcal{S}$ be a collection of closed-session-type indexed sets of computation trees, such that for every $S \in \mathcal{S}$, we have that $t \in S_{\sty{T}}$ implies the following.
  \begin{enumerate}
    \item If $\unfold{\sty{T}} = \ltypeend$, then $t = \treereturn{x}$ for some $x \in X$.
    \item If $\unfold{\sty{T}} = \ltypesend{\pp}{\ell_i}{\basetype_i}{\sty{T}_i}{i \in I}$, then $t = \treesend{\pp}{\msg{\ell_i}{v}}{t'}$ for some $v : \basetype_i$ and $t' \in S_{\sty{T}_i}$.
    \item If $\unfold{\sty{T}} = \ltyperecv{\pp}{\ell_i}{\basetype_i}{\sty{T}_i}{i \in I}$, then $t = \treerecv{\pp}{t_m}{m \in M}$ for some family $(t_m)_{m \in M}$, where $M = \{\msg{\ell_i}{v} \mid i \in I \land v : \basetype_i\}$, and $t_{\msg{\ell_i}{v}} \in S_{\sty{T}_i}$ for all $(\ell_i, v) \in M$.
  \end{enumerate}
  Then the union $(\bigcup S)_{\sty{T}} = \bigcup_{S \in \mathcal{S}} S_{\sty{T}}$ satisfies the same conditions.
  Thus the typing relation $\hastype{}{}$ is the union of the largest such $\mathcal{S}$.
\end{lemma}
\begin{proof}
  It is easy to see that the union satisfies the above three conditions.
  The typing relation $\hastype{}{}$ satisfies these conditions by definition, and thus is given by the union of the largest $\mathcal{S}$.
\end{proof}
\subsection{Proofs for \cref{sec:semantics}}

\begin{lemma}\label{comp-recvs}
  If $\conftypedg{\config{\queue{\rho}}{\rctx{R}[\recvterm{\pp}{\ell_i}{x_i : \basetype_i}{t_i}{i \in I}]}{\queueempty}}{\basetype}{\sty{T}}$ and $\queue{\rho}$ does not contain a message from $\pp$, then $\Recvs{\pp}(\sty{T})$.
\end{lemma}
\begin{proof}
  By induction on the derivation of 
  $\conftypedg{\config{\queue{\rho}}{t}{\queueempty}}{\basetype}{\sty{T}}$.
  \begin{itemize}
    \item If the derivation is by \RuleName{CBase}, then we have $\comptypedg{\emptycontext}{\emptycontext}{\emptycontext}{\rctx{R}[\recvterm{\pp}{\ell_i}{x_i : \basetype_i}{t_i}{i \in I}]}{\basetype}{\sty{T}}$.
      Hence $\sty{T}$ is a supertype of some external choice on $\pp$, and so satisfies $\Recvs{\pp}(\sty{T})$ by the definition of subtyping.
    \item The derivation cannot be by \RuleName{CSend}, because the send queue is empty.
    \item If the derivation is by \RuleName{CRecv}, then we have $\queue{\rho} = \consfront{\pq}{\msg{\ell}{v}}{\queue{\rho'}}$, where $\pq \neq \pp$ by assumption, and $\conftypedg{\config{\queue{\rho'}}{t}{\queueempty}}{\basetype}{\sty{U}}$ for some $\sty{U}$ such that $\reduceexternal{\sty{T}}{\sty{U}}{\pq}{\ell}{\basetype}$, where $v : \basetype$.
      The result therefore follows from the inductive hypothesis and \cref{inductive-predicates-swapping}.
  \end{itemize}
\end{proof}

\begin{lemma}\label{semantic-configuration-receives-immediately}
  Let $\mathcal{C} = \config{\queue{\rho}}{\rctx{R}[\recvterm{\pp}{\ell_i}{x_i : \basetype_i}{t_i}{i \in I}]}{\queue{\sigma}}$ be a configuration, such that
  $\hastype{\mathcal{C}}{\sty{U}}$ holds.
  Assume that there is a sequence
  \[
    \reduceexternal{\sty{U} = \sty{S}_0}{\reduceexternal{\cdots}{\sty{S}_n}{\pr_n}{\ell_n}{\basetype_n}}{\pr_1}{\ell_1}{\basetype_1}
  \]
  with $\pp \not\in \{\pr_1, \dots, \pr_n\}$, and either $\unfold{\sty{S}_n} = \ltypeend$, or $\Sends{\pq}(\sty{U})$ for some $\pq$ such that $\queue{\sigma}$ contains no messages for $\pq$.
  Then $\queue{\rho}$ contains a message from $\pp$.
\end{lemma}
\begin{proof}
  By induction on the derivation of $\hastype{\mathcal{C}}{\sty{U}}$.
  \begin{itemize}
    \item In the base case, we have $\hastype{\rctx{R}[\recvterm{\pp}{\ell_i}{x_i : \basetype_i}{t_i}{i \in I}]}{\sty{U}}$, which immediately tells us that $\Recvs{\pp}(\sty{U})$ holds. Thus we cannot have $\unfold{\sty{U}} = \ltypeend$ or $\Sends{\pp}(\sty{U})$.
    \item If the derivation is by rule \RuleName{CSend}, $\mathcal{C} = \config{\queue{\rho}}{\rctx{R}[\recvterm{\pp}{\ell_i}{x_i : \basetype_i}{t_i}{i \in I}]}{\consback{\queue{\sigma}}{\pr}{\mltype{\ell'}{\basetype'}}}$, and there is some $\sty{U}'$ such that $\reduceinternal{\sty{U}}{\sty{U}'}{\pr}{\ell'}{\basetype'}$ and $\config{\queue{\rho}}{\rctx{R}[\recvterm{\pp}{\ell_i}{x_i : \basetype_i}{t_i}{i \in I}]}{\queue{\sigma}}$, where $v' : \basetype'$.
      Form a new sequence
      \[
        \reduceexternal{\sty{U}' = \sty{S}'_0}{\reduceexternal{\cdots}{\sty{S}'_n}{\pr_n}{\ell_n}{\basetype_n}}{\pr_1}{\ell_1}{\basetype_1}
      \]
      where $\reduceinternal{\sty{S}_i}{\sty{S}'_i}{\pr}{\ell'}{\basetype'}$ for all $i$.
      Such a sequence exists by \cref{inductive-predicates-swapping}.
  We cannot have $\unfold{\sty{S}_n} = \ltypeend$, so $\Sends{\pq}(\sty{S}_n)$ holds and $\queue{\sigma}$ does not contain any messages for $\pp$. The latter implies $\pr \neq \pq$, and hence the former, together with $\reduceinternal{\sty{S}_n}{\sty{U}_n}{\pr}{\ell'}{\basetype'}$, implies $\Sends{\pq}(\sty{U}'_n)$.
  The result then follows from the inductive hypothesis.
     \item If the derivation is by rule \RuleName{CRecv}, then
       $\queue{\rho}$ has the form $\consfront{\pr}{\msg{\ell'}{v'}}{\queue{\rho''}}$, and there is some $\sty{T}$ such that $\reduceexternal{\sty{T}}{\sty{U}}{\pr}{\ell'}{\basetype'}$ and $\hastype{\config{\queue{\rho''}}{\rctx{R}[\recvterm{\pp}{\ell_i}{x_i : \basetype_i}{t_i}{i \in I}]}{\queue{\sigma}}}{\sty{T}}$, where $v' : \basetype'$.
       If $\pr = \pp$, then we are done, otherwise we apply the inductive hypothesis, extending the sequence of session types with $\sty{T}$.
  \end{itemize}
\end{proof}

\begin{lemma}\label{subject-reduction-computations-denotations}
  Suppose that $\comptypedg{\emptycontext}{\emptycontext}{\emptycontext}{t}{\basetype}{\sty{T}}$.
  Recalling from \cref{subject-reduction-computations} that the following reducts $u$ have appropriate types, we have the following.
  \begin{enumerate}
    \item If $\reduceaction{t}{u}{\tauaction}$, then $\sem{t} = \sem{u}$.
    \item If $\reduceaction{t}{u}{\sendaction{\pp}{\msg{\ell}{v}}}$ with $v : \basetype$, then $\reduceaction{\sem{t}}{\sem{u}}{\sendaction{\pp}{\msg{\ell}{v}}}$
    \item If $\reduceaction{t}{u}{\recvaction{\pp}{\msg{\ell}{v}}}$ with $v : \basetype$, then $\Recvs{\pp}(\sty{T})$, and for all $\sty{U}$ such that $\reduceexternal{\sty{T}}{\sty{U}}{\pp}{\ell}{\basetype'}$, we have $\reduceaction{\sem{t}}{\sem{u}}{\recvaction{\pp}{m}}$.
  \end{enumerate}
\end{lemma}
\begin{proof}
  We split the proof into three cases each of which is by induction on $\reduceaction{t}{u}{\alpha}$.
  \begin{itemize}
    \item For $\alpha = \tauaction$ the \RuleName{LetR} case is immediate from the fact that $\ltypeend \subtype \sty{T}$ implies $\unfold{\sty{T}} = \ltypeend$, and hence that $\sty{T} \cdot \sty{U} \subtype \ltypeend \cdot \sty{U} = \sty{U}$.
      The \RuleName{IfT}, \RuleName{IfF}, \RuleName{Sub}, \RuleName{LeqT}, and \RuleName{LeqF} cases are all trivial.
      The \RuleName{Apply} and \RuleName{Cong} cases are both easy inductions on $\rctx{R}$.
    \item For $\alpha = \sendaction{\pp}{\msg{\ell}{v}}$, the \RuleName{Send} case is immediate from the definition of subtyping, and \cref{bisim-normal}.
      For the \RuleName{Cong} case, we proceed by induction on $\rctx{R}$. In the base case, and in the case of a recursive function definition, the result is trivial.
      In the case where $\rctx{R}[\hole] = \bindterm{x}{\rctx{R}'[\hole]}{u}$, we use \cref{reduce-bind} to obtain the desired reduction of interpretations.
    \item For $\alpha = \recvaction{\pp}{\msg{\ell}{v}}$, the \RuleName{Recv} case is immediate from \cref{subtype-lifting} and \cref{bisim-normal}.
      For the \RuleName{Cong} case, we proceed by induction on $\rctx{R}$. In the base case, and in the case of a recursive function definition, the result is trivial.
      In the case where $\rctx{R}[\hole] = \bindterm{x}{\rctx{R}'[\hole]}{u}$, we again use \cref{reduce-bind} to obtain the desired reduction of interpretations.
  \end{itemize}
\end{proof}

\begin{lemma}\label{semantic-config-recv-type-lemma}
  Assume that $\hastype{\mathcal{C} = \config{\queue{\rho}}{t}{\queue{\sigma}}}{\sty{T}}$, that $\reduceexternal{\sty{T}}{\sty{U}}{\pp}{\ell}{\basetype'}$, that $\reduceaction{t}{u}{\recvaction{\pp}{\msg{\ell}{v}}}$ with $v : \basetype'$, and that $\rho$ contains no messages from $\pp$.
  Define $\mathcal{D} = \config{\queue{\rho}}{u}{\queue{\sigma}}$. Then $\hastype{\mathcal{D}}{\sty{U}}$ and $\reduceaction{\sem{\mathcal{C}}}{\sem{\mathcal{D}}}{\recvaction{\pp}{\msg{\ell}{v}}}$.
\end{lemma}
\begin{proof}
  By induction on the derivation of $\hastype{\mathcal{C}}{\sty{U}}$.
  \begin{itemize}
    \item If the derivation is by rule \RuleName{CBase}, then the result is immediate from \cref{subject-reduction-computations-denotations}.
    \item If the derivation is by rule \RuleName{CSend}, then we have $\queue{\sigma} = \consback{\queue{\sigma}'}{\pp}{\msg{\ell}{v}}$, $\hastype{\config{\queue{\rho}}{t}{\queue{\sigma'}}}{\sty{T}'}$, and $\reduceinternal{\sty{T}}{\sty{T}'}{\pp}{\ell}{\basetype}$, where $v : \basetype$. The result therefore follows from the inductive hypothesis \cref{inductive-predicates-swapping}, and \cref{reduction-swapping}.
    \item If the derivation is by rule \RuleName{CRecv}, then we have $\queue{\rho} = \consfront{\pq}{\msg{\ell}{v}}{\queue{\rho}'}$, $\pq \neq \pp$, $\hastype{\config{\queue{\rho'}}{t}{\queue{\sigma}}}{\sty{T}'}$, and $\reduceexternal{\sty{T}'}{\sty{T}}{\pp}{\ell}{\basetype}$, where $v : \basetype$. The result therefore follows from the inductive hypothesis, \cref{inductive-predicates-swapping}, and \cref{reduction-swapping}.
  \end{itemize}
\end{proof}

\begin{lemma}\label{reduce-semantic-configuration}\
    Assume that $\hastype{\mathcal{C}}{\sty{U}}$, where $\mathcal{C}$ is a configuration, and that $\reduceaction{\mathcal{C}}{\mathcal{D}}{\alpha}$.
    \begin{enumerate}
        \item If $\alpha = \tauaction$, then $\hastype{\mathcal{D}}{\sty{U}}$, and $\sem{\mathcal{C}} = \sem{\mathcal{D}}$.
        \item If $\alpha = \sendaction{\pp}{\msg{\ell}{v}}$, where $v : \basetype$, then there is some $\sty{T}$ such that $\reduceinternal{\sty{U}}{\sty{T}}{\pp}{\ell}{\basetype}$, $\hastype{\mathcal{D}}{\sty{T}}$, and $\reduceaction{\sem{\mathcal{C}}}{\sem{\mathcal{D}}}{\sendaction{\pp}{\msg{\ell}{v}}}$.
        \item If $\alpha = \recvaction{\pp}{\msg{\ell}{v}}$, where $v : \basetype$, and $\reduceexternal{\sty{U}}{\sty{T}}{\pp}{\ell}{\basetype}$, then $\hastype{\mathcal{D}}{\sty{T}}$, and $\reduceaction{\sem{\mathcal{C}}}{\sem{\mathcal{D}}}{\recvaction{\pp}{\msg{\ell}{v}}}$.
    \end{enumerate}
\end{lemma}
\begin{proof}
  For (1), there are three cases to consider: either the reduction is $\mathcal{C}$'s computation making an internal transition, or it is that computation producing a message, or it is that consuming a message.
  \begin{itemize}
    \item In the first case, we have $\mathcal{C} = \config{\queue{\rho}}{t}{\queue{\sigma}}$, $\mathcal{D} = \config{\queue{\rho}}{u}{\queue{\sigma}}$, and $\reduceaction{t}{u}{\tauaction}$.
      By \cref{subject-reduction-computations}, we have that $u$ has the same type as $t$, and by \cref{subject-reduction-computations-denotations}, we have $\sem{t} = \sem{u}$, and hence the result follows by a trivial induction on the typing derivation for $\mathcal{C}$.
    \item In the second case, we have $\mathcal{C} = \config{\queue{\rho}}{t}{\queue{\sigma}}$, $\mathcal{D} = \config{\queue{\rho}}{u}{\consfront{\pp}{\msg{\ell}{v}}{\queue{\sigma}}}$, and $\reduceaction{t}{u}{\sendaction{\pp}{\msg{\ell}{v}}}$, with $v : \basetype$.
      We proceed by induction on the derivation of $\hastype{\mathcal{C}}{\sty{U}}$.
      In the base case, we have $\hastype{t}{\sty{U}}$, so by \cref{subject-reduction-computations} there is some $\sty{T}$ such that $\reduceinternal{\sty{U}}{\sty{T}}{\pp}{\ell}{\basetype}$ and $\hastype{u}{\sty{U}}$, so that $\hastype{\mathcal{D}}{\sty{U}}$. We also have $\sem{\mathcal{C}} = \reduceaction{\sem{t}}{\sem{u}}{\sendaction{\pp}{\msg{\ell}{v}}}$ by \cref{subject-reduction-computations-denotations}, and hence $\sem{\mathcal{C}} = \sem{\mathcal{D}}$.
      Both of the inductive cases are immediate from the inductive hypothesis.
    \item In the third case, we have $\mathcal{C} = \config{\consback{\queue{\rho}}{\pp}{\msg{\ell}{v}}}{t}{\queue{\sigma}}$, $\mathcal{D} = \config{\queue{\rho}}{u}{\queue{\sigma}}$, and $\reduceaction{t}{u}{\recvaction{\pp}{\msg{\ell}{v}}}$.
      We proceed by induction on the derivation of $\hastype{\mathcal{C}}{\sty{U}}$.
      The derivation cannot be by the base case.
      If the derivation is by rule \RuleName{CRecv}, then $\mathcal{C}$ has the form $\config{\consfront{\pq}{\msg{\ell'}{v'}}{\queue{\rho}'}}{t}{\queue{\sigma}}$, and we have $v' : \basetype'$, and $\hastype{\mathcal{C}'}{\sty{U}'}$, where $\mathcal{C'} = \config{\queue{\rho'}}{t}{\queue{\sigma}}$, and $\reduceexternal{\sty{U}'}{\sty{U}}{\pq}{\ell'}{\basetype'}$.
      If $\pq = \pp$, then we have $\msg{\ell'}{v'} = \msg{\ell}{v}$, and $\queue{\rho} = \queue{\rho'}$ does not contain a message from $\pp$.
      The result then follows from \cref{subject-reduction} and \cref{semantic-config-recv-type-lemma}.
      Otherwise, the result is immediate from the inductive hypothesis.
      If the derivation is by rule \RuleName{CSend}, then the result is again immediate from the inductive hypothesis.
  \end{itemize}

  For (2), we necessarily have $\mathcal{C} = \config{\queue{\rho}}{t}{\consback{\queue{\sigma}}{\pp}{\msg{\ell}{v}}}$ and $\mathcal{D} = \config{\queue{\rho}}{t}{\queue{\sigma}}$.
  We proceed by induction on the derivation of $\hastype{\mathcal{C}}{\sty{U}}$.
  The derivation cannot be by the base case.
  If it is by rule \RuleName{CSend}, then $\mathcal{C}$ has the form $\config{\queue{\rho}}{t}{\consback{\queue{\sigma'}}{\pq}{\msg{\ell'}{v'}}}$, and we have $v' : \basetype'$, and $\hastype{\mathcal{C}'}{\sty{U}'}$, where $\mathcal{C'} = \config{\queue{\rho}}{t}{\queue{\sigma'}}$, and $\reduceinternal{\sty{U}}{\sty{U}'}{\pq}{\ell'}{\basetype'}$.
  By the definition of $\sem{\mathcal{C}}$, we have $\reduceaction{\sem{\mathcal{C}}}{\sem{\mathcal{C'}}}{\sendaction{\pq}{\msg{\ell'}{v'}}}$.
  We consider two cases:
  \begin{itemize}
    \item If $\pq = \pp$, then $(\queue{\sigma'}, \ell', \basetype') = (\queue{\sigma}, \ell, \basetype)$. The result is then immediate, by taking $\sty{T} = \sty{U'}$.
    \item If $\pq \neq \pp$, then by the inductive hypothesis we have $\sty{T'}$ such that $\reduceinternal{\sty{U}'}{\sty{T'}}{\pp}{\ell}{\basetype}$, $\hastype{\config{\queue{\rho}}{t}{\queue{\sigma'}}}{\sty{T'}}$, and $\reduceaction{\sem{\mathcal{C}'}}{\sem{\config{\queue{\rho}}{t}{\queue{\sigma'}}}}{\sendaction{\pp}{\msg{\ell}{v}}}$. By \cref{inductive-predicates-swapping} there is some $\sty{T}$ such that $\reduceinternal{\sty{U}}{\reduceinternal{\sty{T}}{\sty{T'}}{\pq}{\ell'}{\basetype'}}{\pp}{\ell}{\basetype}$. We also have $\reduceaction{\sem{\mathcal{C}}}{\reduceaction{\sem{\config{\queue{\rho}}{t}{\queue{\sigma}}}}{\sem{\mathcal{C'}}}{\sendaction{\pq}{\msg{\ell'}{\basetype'}}}}{\sendaction{\pp}{\msg{\ell}{v}}}$, by \cref{reduction-swapping}. Finally, we have $\hastype{\config{\queue{\rho}}{t}{\queue{\sigma}}}{\sty{T}}$ using rule \RuleName{CSend}.
  \end{itemize}
  Finally, if the proof is by rule \RuleName{CRecv}, then $\mathcal{C}$ has the form $\config{\consfront{\pq}{\msg{\ell'}{v'}}{\queue{\rho}}}{t}{\consback{\queue{\sigma}}{\pp}{\msg{\ell}{v}}}$, and we have $v' : \basetype'$, and $\hastype{\mathcal{C}'}{\sty{U}'}$, where $\mathcal{C'} = \config{\queue{\rho'}}{t}{\consback{\queue{\sigma}}{\pp}{\msg{\ell}{v}}}$, and $\reduceexternal{\sty{U}'}{\sty{U}}{\pq}{\ell'}{\basetype'}$.
  By the definition of $\sem{\mathcal{C}}$, we have $\reduceaction{\sem{\mathcal{C'}}}{\sem{\mathcal{C}}}{\recvaction{\pq}{\msg{\ell'}{v'}}}$.
  By the inductive hypothesis, we have $\sty{T'}$ such that $\reduceinternal{\sty{U}'}{\sty{T'}}{\pp}{\ell}{\basetype}$, $\hastype{\config{\queue{\rho}'}{t}{\queue{\sigma}}}{\sty{T'}}$, and $\reduceaction{\sem{\mathcal{C}'}}{\sem{\config{\queue{\rho}'}{t}{\queue{\sigma}}}}{\sendaction{\pp}{\msg{\ell}{v}}}$. By \cref{inductive-predicates-swapping} there is some $\sty{T}$ such that $\reduceinternal{\sty{U}}{\sty{T}'}{\pp}{\ell}{\basetype}$ and $\reduceexternal{\sty{T'}}{\sty{T}}{\pq}{\ell'}{\basetype'}$.
  We also have $\reduceaction{\sem{\mathcal{C}}}{\sem{\config{\queue{\rho}}{t}{\queue{\sigma}}}}{\sendaction{\pp}{\msg{\ell}{v}}}$ and $\reduceaction{\sem{\mathcal{C'}}}{\sem{\config{\queue{\rho}}{t}{\queue{\sigma}}}}{\recvaction{\pq}{\msg{\ell'}{\basetype'}}}$, by \cref{reduction-swapping}. Finally, we have $\hastype{\config{\queue{\rho}}{t}{\queue{\sigma}}}{\sty{T}}$ using rule \RuleName{CRecv}.

  (3) is trivial, using \RuleName{CRecv}.
\end{proof}

\begin{definition}
  A computation $\comptypedg{\Theta}{\Phi}{\emptycontext}{t}{\basetype}{\sty{T}}$ \emph{progresses} when
  there is a reduction $\reduceactionmany{t}{u}{\tauaction}$ such that one of the following holds.
  \begin{itemize}
    \item $u$ has the form $\rctx{R}[\returnterm{v}]$, where $\rctx{R}[{-}]$ does not contain any non-recursive let-bindings;
    \item $u$ has the form $\rctx{R}[\sendterm{\ell}{v}{\pp}{t}]$, where $t$ progresses;
    \item $u$ has the form $\rctx{R}[\recvterm{\pp}{\ell_i}{x_i : \basetype_i}{u_i}{i \in I}]$, and $\Recvs{\pp}(\sty{T})$ holds, and $u_i[x_i \mapsto v_i]$ progresses for all $v: \basetype_i$ and $\reduceexternal{\sty{T}}{\sty{U}}{\pp}{\ell_i}{\basetype_i}$.
  \end{itemize}
  Moreover we require the following.
  \begin{itemize}
    \item If $\Sends{\pp}(\sty{T})$ then there is some $h$, such that for every finite sequence of reductions
      $\reduceaction{t = t_1}{\reduceaction{\cdots}{t_n}{\alpha_h}}{\alpha_1}$
      where $\comptypedg{\emptycontext}{\Phi}{\emptycontext}{t}{\basetype}{\sty{T}_i}$, $\sty{T} = \sty{T}_1$,
      $\alpha_i = \sendaction{\pq}{\msg{\ell'}{v}}$ implies $\reduceinternal{\sty{T}_i}{\sty{T}_{i + 1}}{\pq}{\ell'}{\basetype'}$ with $v : \basetype$,
      and $\alpha_i = \tauaction$ implies $\sty{T}_i = \sty{T}_{i + 1}$,
      then: there is no $i$ such that $\alpha_i$ has the form $\recvaction{\pp}{m}$, and there is some $i$ such that $\alpha_i = \sendaction{\pp}{\msg{\ell}{v}}$.
    \item If $\reduceexternal{\sty{T}}{\sty{U}}{\pp}{\ell}{\basetype}$, then there is some $h$, such that for every finite sequence of reductions
      $\reduceaction{t = t_1}{\reduceaction{\cdots}{t_n}{\alpha_h}}{\alpha_1}$
      where $\comptypedg{\emptycontext}{\Phi}{\emptycontext}{t}{\basetype}{\sty{T}_i}$, $\sty{T} = \sty{T}_1$,
      $\alpha_i = \sendaction{\pq}{\msg{\ell'}{v}}$ implies $\reduceinternal{\sty{T}_i}{\sty{T}_{i + 1}}{\pq}{\ell'}{\basetype'}$ with $v : \basetype$,
      $\alpha_i = \recvaction{\pq}{\msg{\ell'}{v}}$ implies $\reduceexternal{\sty{T}_i}{\sty{T}_{i + 1}}{\pq}{\ell'}{\basetype'}$ with $v : \basetype$,
      and $\alpha_i = \tauaction$ implies $\sty{T}_i = \sty{T}_{i + 1}$,
      then there is some $i$ such that $\alpha_i = \recvaction{\pp}{\msg{\ell}}$.
  \end{itemize}
\end{definition}

\begin{lemma}\label{computation-progress}
  Assume that $\comptypedg{\Theta}{\emptycontext}{\emptycontext}{t}{\basetype}{\sty{T}}$.
  Let $t = \letrecterm{f_1}{(x_{11}, \dots, x_{1n_1})}{u_1}{\cdots \letrecterm{f_n}{(x_{11}, \dots, x_{1n_n})}{u_n}{t'}}$
  be a computation, such that
  \[
    \comptypedg{\Theta, X}{\Phi_i}{\Gamma, x_{i1} : \basetype_{i1}, \dots }{u_i}{\basetype'_i}{\sty{U}'_i}
  \]
  where $\Phi_i = f_1 : (\basetype_{11}, \dots, \basetype_{1n_1}) \gto{\ltyperec{\ltypevar{X}}{\sty{U}_1}} \basetype'_1, \dots,f_{i} : (\basetype_{1{i-1}}, \dots, \basetype_{1n_{i-1}}) \gto{\ltypevar{X}} \basetype'_{i - 1}$,
  and such that 
  \[
    \letrecterm{f_1}{(x_{11}, \dots, x_{1n_1})}{u_1}{\letrecterm{f_1}{(x_{11}, \dots, x_{1n_1})}{u_{i-1}}{u_i}}(\gamma, v_1,\dots,v_{n_i})
  \]
  progresses for all $v_1 : \basetype_1, \dots, v_n : \basetype_n$.
  Also assume that
  \[
    \comptypedg{\emptycontext}{\Phi_{n+1}}{\Gamma}{t'}{\basetype_i}{\sty{T}}
  \]
  For all $\gamma$, the computation 
  $t(\gamma)$ progresses.
\end{lemma}
\begin{proof}
  By induction on the typing derivation for $t'$
  In each case, we note that after a reduction, the reduct is well-typed by \cref{subject-reduction-computations}.
  \begin{itemize}
  \item For \TRULE{$\subtype$}, we note that \cref{subtype-lifting} enables us to transfer our $\Sends{\pp}$ and $\reduceexternalsymbol$ assumptions to the subtype, so that the result follows from the inductive hypothesis.
  \item For \TRULE{Ret}, \TRULE{+} and \TRULE{<}, we have $\unfold{\sty{T}} = \ltypeend$, and the result is therefore trivial.
  \item For \TRULE{Send}, the result follows from the inductive hypothesis, increasing $h$ by $1$ to account for the extra reduction.
  \item For \TRULE{Recv}, the result follows from the indcutive hypothesis, again increasing $h$ by $1$ to account for the extra reduction.
  \item For \TRULE{Let}, the computation to reduce first progresses, by the inductive hypothesis. If it reduces to a $\returnterm$, then we can continue the reduction by rule \TRULE{LetR}, and the result follows from the inductive hypothesis on the other computation. Otherwise, it reduces to either a $\sendtermkeyword$ or to a $\recvtermkeyword$, which then progresses. To take care of the second part of the definition of progresses, we proceed by induction on $\Sends{\pp}{(\sty{T} \cdot \sty{T'})}$ to show that, either we can satisfy the second part using the fact that the first computation progresses, or that the first computation will reduce to a $\returntermkeyword$ in finitely many steps, after which we can satisfy the second part using the second computation. Similarly for $\reduceexternalsymbol$.
  \item For \TRULE{If}, the result is immediate from the inductive hypothesis.
  \item For \TRULE{LetRec}, the result is immediate from the inductive hypothesis, adding a new binding to the sequence of recursive functions. Note that this new recursive function progresses by guardedness of recursive calls.
  \item For \TRULE{App}, the result is immediate from our assumptions about recursive functions.
  \end{itemize}
\end{proof}

\begin{lemma}\label{progress-happens}
  Let $\mathcal{C} = \config{\queue{\rho}}{t}{\queue{\sigma}}$ be a configuration, such that
  $\hastype{\mathcal{C}}{\sty{U}}$ holds.
  \begin{itemize}
    \item If $\unfold{\sty{U}} = \ltypeend$, then there is some reduction $\reduceactionmany{\mathcal{C}}{\mathcal{D}}{\tauaction}$ with $\Result{(\mathcal{D})}$ defined.
    \item If $\Sends{\pp}(\sty{U})$, then there is some reduction $\reduceactionmany{\mathcal{C}}{\mathcal{D}}{\sendaction{\pp}{m}}$.
  \end{itemize}
\end{lemma}
\begin{proof}
  In both cases, we proceed by well-founded induction on the size of the queue $\rho$.
  For (1), we invoke \cref{computation-progress}.
  \begin{itemize}
    \item If $\reduceactionmany{t}{\rctx{R}[\returnterm{v}]}{\tauaction}$, then since $\hastype{\mathcal{C}}{\sty{U}}$ holds, both $\queue{\rho}$ and $\queue{\sigma}$ are necessarily empty, so we can take $\mathcal{D} = \mathcal{C}$.
    \item We cannot have $\reduceactionmany{t}{\rctx{R}[\sendterm{\ell}{v}{\pp}{u}]}{\tauaction}$, because the reduct does not have type $\ltypeend$, and thus neither can $t$.
    \item If $\reduceactionmany{t}{\rctx{R}[\recvterm{\pp}{\ell_i}{x_i : \basetype_i}{u_i}{i \in I}]}{\tauaction}$, then by \cref{semantic-configuration-receives-immediately}, the receive queue $\queue{\rho}$ contains a message from $\pp$. Thus the configuration reduces, this reduction reduces the size of the queue, and the reduct has type $\sty{T}$ by \cref{reduce-semantic-configuration}, so the result follows from the inductive hypothesis.
  \end{itemize}
  For (2), if $\queue{\sigma}$ contains a message for $\pp$, then we have the required reduction trivially.
  Otherwise, an induction on the derivation of $\hastype{\mathcal{C}}{\sty{U}}$, using \cref{inductive-predicates-swapping}, shows that we can assign to $t$ some type $\sty{T}$ such that $\Sends{\pp}(\sty{T})$.
  We can therefore appeal to \cref{computation-progress} to get the correct reduction.
\end{proof}

\begin{lemma}
  Let $\mathcal{C}$ be a configuration.
  If $\hastype{\mathcal{C}}{\sty{T}}$, then $\mathcal{C} \bisim_{\sty{T}} \sem{\mathcal{C}}$.
\end{lemma}
\begin{proof}
  We verify each case of \cref{def:bisim} in turn.
  \begin{enumerate}
    \item For (1), (a) follows from \cref{reduce-semantic-configuration}(1); (b) is trivial, because there are no $\tauaction$ reductions for computation trees.
    \item For (2), if $\Result(\mathcal{C}) = x$, then $\mathcal{C} = \config{\queueempty}{\treereturn{x}}{\queueempty}$. This implies $\unfold{\sty{T}} = \ltypeend$, and we also have $\sem{\mathcal{C}} = \treereturn{x}$, which satisfies $\Result(\treereturn{x}) = x$.
      If $\Result{\sem{\mathcal{C}}} = x$, then $\sem{\mathcal{C}} = \treereturn{x}$, so we have $\unfold{\sty{T}} = \ltypeend$. By \cref{progress-happens},
      there is some $\mathcal{D}$ such that $\reduceactionmany{\mathcal{C}}{\mathcal{D}}{\tauaction}$, with $\Result(\mathcal{D})$ defined. By \cref{reduce-semantic-configuration}, we have $\sem{C} = \sem{D}$, and hence that $\Result{(\mathcal{D})} = x$.
    \item For (3), (a) follows from \cref{reduce-semantic-configuration}(2). For (b), if $\Sends{\pp}(\sty{T})$, then there is a reduction $\reduceactionmany{\mathcal{C}}{\mathcal{D}}{\sendaction{\pp}{m}}$ by \cref{progress-happens}. Hence if $\reduceaction{\sem{\mathcal{C}}}{t}{\sendaction{\pp}{m'}}$, we have $m = m'$ and $t = \sem{\mathcal{D}}$ by \cref{reduce-semantic-configuration}.
    \item For (4), (a) follows from \cref{reduce-semantic-configuration}(3). For (b), if $\reduceexternal{\sty{T}}{\sty{U}}{\pp}{\ell}{\basetype}$ and $\reduceaction{\sem{\mathcal{C}}}{t}{\recvaction{\pp}{m}}$, then there is a unique $\mathcal{D}$ such that $\reduceaction{\mathcal{C}}{\mathcal{D}}{\recvaction{\pp}{m}}$, and we obtain $\sem{\mathcal{D}} \bisim_{\sty{U}} t$ by \cref{reduce-semantic-configuration} and \cref{reduction-unique-bisim}.
  \end{enumerate}
\end{proof}

\subsection{Proofs for \cref{sec:session-calculus}}

\begin{definition}
    We write $\sty{T} \extgeq \sty{U}$ when $\sty{T}$ can be written as $\sty{T} = \mergeltype{\sty{T}_i}{i \in I}$, with $I$ a finite $\sty{U} = \mergeltype{\sty{T}_i}{i \in J}$ for some non-empty $J \subseteq I$.
\end{definition}

\begin{lemma}\label{extgeq-subtype}
    If $\sty{T} \extgeq \sty{U}$, then $\sty{T} \suptype \sty{U}$.
\end{lemma}
\begin{proof}
  Immediate from \cref{admissible-subtyping} and the definition of full merging.
\end{proof}

\begin{lemma}\label{merge-merges}
    If $\sty{T}_i \extgeq \sty{U}_i$ for each $i \in I$, and the merge $\sty{T} = \mergeltype{\sty{T}_i}{i \in I}$ exists, then the merge $\sty{U} = \mergeltype{\sty{U}_i}{i \in I}$ also exists, and satisfies $\sty{T} \extgeq \sty{U}$.
\end{lemma}
\begin{proof}
    We have $\sty{T}_i = \mergeltype{\sty{T}_{ij}}{j \in J_i}$ with $\sty{U}_i = \mergeltype{\sty{T}_{ij}}{j \in J'_i}$ for some $J'_i \subseteq J_i$.
    Hence $\sty{T} = \mergeltype{\sty{T}_{ij}}{i \in I, j \in J_i}$, so $\sty{T} \extgeq \mergeltype{\sty{T}_{ij}}{i \in I, j \in J'_i} = \mergeltype{\sty{U}_i}{i \in I} = \sty{U}$.
\end{proof}

\begin{lemma}\label{merge-reduceexternals}\
  Let $I$ be a finite set, and let $J \subset I$ be non-empty.
  If the merge $\sty{T} = \mergeltype{\sty{T}_i}{i \in I}$ is defined, and $\reduceexternal{\sty{T}_i}{\sty{U}_i}{\pp}{\ell}{\basetype}$ for each $i \in I$, then the merge $\sty{U} = \mergeltype{\sty{U}_i}{i \in J}$ is defined and satisfies $\reduceexternal{\sty{T}}{\sty{U}}{\pp}{\ell}{\basetype}$.
\end{lemma}
\begin{proof}
    Since $I$ is finite, there is some natural number $h$ that bounds the height of the assumed derivations of $\reduceexternal{\sty{T}_i}{\sty{U}_i}{\pp}{\ell}{\basetype}$.
    We proceed by induction on $h$.
    If $h = 0$, then the derivations are all by \RuleName{$\reduceexternalsymbol$-base}, and the result is trivial by the definition of merge.
    If $h > 0$, then the proofs may all be by \RuleName{$\reduceexternalsymbol$-rec}. In this case, we have $\sty{T}_i = \ltyperec{\ltypevar{X}}{\sty{T}'_i}$, and the merge $\mergeltype{\unfold{\sty{T}_i}}{i \in I} = \subst{(\mergeltype{\sty{T}'_i}{i \in I})}{\ltypevar{X} \mapsto \ltyperec{\ltypevar{X}}{\sty{T}}}$ is defined. The result therefore follows from the inductive hypothesis.
    Otherwise, the derivations are either all by \RuleName{$\reduceexternalsymbol$-$\oplus$}, or all by \RuleName{$\reduceexternalsymbol$-$\&$}.
    In both cases the result is immediate from the inductive hypothesis and the definition of merge; in the $\&$ case also using the fact that a finite union of finite sets is again finite.
\end{proof}

\begin{lemma}\label{unfold-project}
  Let $\gty{G} = \gtyperec{\gtypevar{X}}{\gty{G'}}$ be a closed global type over $\role{R} = \{\pr_1, \dots, \pr_n\}$, such that the projection $\project{\gty{G}}{\pr}$ is defined for all $\pr \in \role{R}$.
  Then the projections $\project{\subst{\gty{G'}}{\gtypevar{X} \mapsto \gty{G}}}{\pr}$ are defined, and satisfy
  \[
    \project{\subst{\gty{G'}}{\gtypevar{X} \mapsto \gty{G}}}{\pr}
    =
    \begin{cases}
        \ltypeend & \text{if}~\project{\gty{G}}{\pr} = \ltypeend
        \\
        \subst{(\project{\gty{G}'}{\pr})}{\ltypevar{X} \mapsto \project{\gty{G'}}{\pr}}
        ~\text{if}~\project{\gty{G}}{\pr} = \ltyperec{\ltypevar{X}}{\sty{T}}
    \end{cases}
  \]
\end{lemma}

\begin{lemma}\label{global-reduction}
  Let $\gty{G}$ be a closed global type over $\role{R} = \{\pr_1, \dots, \pr_n\}$, such that the projection $\project{\gty{G}}{\pr}$ is defined for all $\pr \in \role{R}$.
  Assume that $\reduceinternal{(\project{\gty{G}}{\pr_j})}{\sty{T}_{\pr_j}}{\pr_k}{\ell}{\basetype}$.
  Then there exist session types $\sty{T}_{\pr}$ for each $\pr \in \role{R}\setminus \{\pr_j\}$, and some global type $\gty{G}'$ over $\{\pr_1, \dots, \pr_n\}$, such that $\sty{T}_{\pr} = \project{\gty{G}'}{\pr}$ for all $\pr \in \role{R}$, and such that
  \[
    \reduceexternal{(\project{\gty{G}}{\pr_k})}{\sty{T}_{\pr_k}}{\pr_j}{\ell}{\basetype}
    \qquad
    (\project{\gty{G}}{\pr}) \extgeq \sty{T}_{\pr}~\text{for each}~\pr \in \mathsf{R} \setminus \{\pr_j, \pr_k\}
  \]
\end{lemma}
\begin{proof}
  By induction on $\reduceinternal{(\project{\gty{G}}{\pr_j})}{\sty{T}_{\pr_j}}{\pr_k}{\ell}{\basetype}$.
  \begin{itemize}
    \item If the proof is by rule \RuleName{$\reduceinternalsymbol$-base}, then we have $\project{\gty{G}}{\pr_j} = \ltypesend{\pr_k}{\ell_i}{\basetype_i}{\sty{T}_{\pr_ji}}{i \in I}$ with $(\ell, \basetype, \sty{T}_{\pr_j}) = (\ell_i, \basetype_i, \sty{T}_{\pr_ji})$ for some $i \in I$.
      We proceed by induction on $\gty{G}$.
      By the definition of projection, $\gty{G}$ can only be a communication.
      There are two cases to consider.
      \begin{itemize}
        \item If $\gty{G} = \gtypecomm{\pr_j}{\pr_k}{\ell_i}{\basetype_i}{\gty{G}_i}{i \in I}$, then we have $\sty{T}_{\pr_ji} = \project{\gty{G}_i}{\pr_j}$. Since $\project{\gty{G}}{\pr_k} = \ltyperecv{\pr_j}{\ell_i}{\basetype_i}{(\project{\gty{G}_i}{\pr_k})}{i \in I}$ and $\project{\gty{G}}{\pr_i} = \mergeltype{\project{\gty{G}_i}{\pr}}{i \in I} \extgeq \sty{T}_{\pr}$ for each $\pr \in \mathsf{R} \setminus \{\pr_j, \pr_k\}$, we can conclude by taking $\gty{G'} = \gty{G}_i$.
        \item Otherwise, we have $\gty{G} = \gtypecomm{\pp}{\pq}{\ell'_{i'}}{\basetype'_{i'}}{\gty{G}_{i'}}{i \in I'}$, with $\pr_j \not\in \{\pp, \pq\}$.
          By the definition of merging it follows that $\project{\gty{G}_{i'}}{\pr_j} = \ltypesend{\pr_k}{\ell_i}{\basetype_i}{\sty{T}_{\pr_jii'}}{i \in I}$, for each $i' \in I'$, and $\sty{T}_{\pr_ji} = \mergeltype{\sty{T}_{\pr_jii'}}{i' \in I'}$.
          By the inductive hypothesis, there exist $\sty{T}_{\pr i'}$ and $\gty{G}'_{i'}$, such that $\sty{T}_{\pr i'} = \project{\gty{G}'_{i'}}{\pr}$ for all $\pr \in \role{R}$, and such that
  \[
    \reduceexternal{(\project{\gty{G}_{i'}}{\pr_k})}{\sty{T}_{\pr_ki'}}{\pr_j}{\ell}{\basetype}
    \qquad
    (\project{\gty{G}_{i'}}{\pr}) \extgeq \sty{T}_{\pr i'}~\text{for each}~\pr \in \mathsf{R} \setminus \{\pr_j, \pr_k\}
  \]
          Define $\gty{G}' = \gtypecomm{\pp}{\pq}{\ell'_{i'}}{\basetype'_{i'}}{\gty{G}'_{i'}}{i \in I'}$.
          The projection $\sty{T}_{\pr_k} = \project{\gty{G'}}{\pr_j}$ exists and satisfies $\reduceexternal{(\project{\gty{G}}{\pr_k})}{\sty{T}_{\pr_k}}{\pr_j}{\ell}{\basetype}$, either by the congruence rules for $\reduceexternalsymbol$, if $\pr_k \in \{\pp, \pq\}$, or by \cref{merge-reduceexternals} otherwise.
          For every $\pr \in \mathsf{R} \setminus \{\pr_j, \pr_k\}$, the projection $\sty{T}_{\pr_k} = \project{\gty{G'}}{\pr_j}$ also exists and satisfies $(\project{\gty{G}}{\pr}) \extgeq \sty{T}_{\pr}$, either by the definition of merging, if $\pr \in \{\pp, \pq\}$, or by \cref{merge-merges}.
          Finally, we have $\project{\gty{G}'}{\pr_j} = \sty{T}_{\pr_j}$ by definition.
      \end{itemize}
    \item If the proof is by rule \RuleName{$\reduceinternalsymbol$-$\oplus$}, then we can write $\project{\gty{G}}{\pr_j} = \ltypesend{\pr'}{\ell_i}{\basetype_i}{\sty{T}'_{\pr_ji}}{i \in I}$ and $\sty{T}_{\pr_j} = \ltypesend{\pr'}{\ell_i}{\basetype_i}{\sty{T}'_{\pr_ji}}{i \in J}$, where $\pr' \neq \pr_k$, $J\subseteq I$, and $\reduceinternal{\sty{T}'_{\pr_ji}}{\sty{T}_{\pr_ji}}{\pr_k}{\ell}{\basetype}$ for all $i \in I$.
      We proceed by induction on $\gty{G}$.
      By the definition of projection, $\gty{G}$ can only be a communication.
      There are two cases to consider.
      \begin{itemize}
        \item If $\gty{G} = \gtypecomm{\pr_j}{\pr'}{\ell_i}{\basetype_i}{\gty{G}_i}{i \in I}$, then we have $\sty{T}'_{\pr_ji} = \project{\gty{G}_i}{\pr_j}$.
          For every $i \in J$, the inductive hypothesis provides $\gty{G}'_i$ and $\sty{T}_{\pr i}$ such that $\sty{T}_{\pr i} = \project{\gty{G}'_i}{\pr}$, and such that
  \[
    \reduceexternal{(\project{\gty{G}_i}{\pr_k})}{\sty{T}_{\pr_ki}}{\pr_j}{\ell}{\basetype}
    \qquad
    (\project{\gty{G}_{i}}{\pr}) \extgeq \sty{T}_{\pr i}~\text{for each}~\pr \in \mathsf{R} \setminus \{\pr_j, \pr_k\}
  \]
          Define $\gty{G'} = \gtypecomm{\pr_j}{\pr'}{\ell_i}{\basetype_i}{\gty{G}'_i}{i \in J}$, $\sty{T}_{\pr'} = \ltyperecv{\pr_j}{\ell_i}{\basetype_i}{\sty{T}_{\pr'i}}{i \in J}$, and $\sty{T}_{\pr} = \mergeltype{\sty{T}_{\pr i}}{i \in J}$ for $\pr \in \role{R} \setminus \{\pr_j, \pr'\}$.
          The latter exists and satisfies $\reduceexternal{(\project{\gty{G}}{\pr_k})}{\sty{T}_{\pr_k}}{\pr_j}{\ell}{\basetype}$ by \cref{merge-reduceexternals} for $\pr = \pr_k$,
          and exists and satisfies $(\project{\gty{G}}{\pr}) \extgeq \sty{T}_{\pr}$ for $\pr \neq \pr_k$ by \cref{merge-merges}.
          We have $(\project{\gty{G}}{\pr'}) \extgeq \sty{T}_{\pr'}$ by the definition of merge.
        \item Otherwise, we have $\gty{G} = \gtypecomm{\pp}{\pq}{\ell'_{i'}}{\basetype'_{i'}}{\gty{G}_{i'}}{i \in I'}$, with $\pr_j \not\in \{\pp, \pq\}$. We then appeal to the inductive hypthesis, as for \RuleName{$\reduceinternalsymbol$-base}.
      \end{itemize}
    \item If the proof is by rule \RuleName{$\reduceinternalsymbol$-$\&$}, then we can write $\project{\gty{G}}{\pr_j} = \ltyperecv{\pr'}{\ell_i}{\basetype_i}{\sty{T}'_{\pr_ji}}{i \in I}$ and $\sty{T}_{\pr_j} = \ltyperecv{\pr'}{\ell_i}{\basetype_i}{\sty{T}'_{\pr_ji}}{i \in I}$, where $\reduceinternal{\sty{T}'_{\pr_ji}}{\sty{T}_{\pr_ji}}{\pr_k}{\ell}{\basetype}$ for all $i \in I$.
      We proceed by induction on $\gty{G}$.
      By the definition of projection, $\gty{G}$ can only be a communication.
      There are two cases to consider.
      \begin{itemize}
        \item If $\gty{G} = \gtypecomm{\pr'}{\pr_j}{\ell_i}{\basetype_i}{\gty{G}_i}{i \in I}$, then we have $\sty{T}_{\pr_ji} = \project{\gty{G}_i}{\pr_j}$.
          For every $i \in I$, the inductive hypothesis provides $\gty{G}'_i$ and $\sty{T}_{\pr i}$ such that $\sty{T}_{\pr i} = \project{\gty{G}'_i}{\pr}$, and such that
  \[
    \reduceexternal{(\project{\gty{G}_i}{\pr_k})}{\sty{T}_{\pr_ki}}{\pr_j}{\ell}{\basetype}
    \qquad
    (\project{\gty{G}_{i}}{\pr}) \extgeq \sty{T}_{\pr i}~\text{for each}~\pr \in \mathsf{R} \setminus \{\pr_j, \pr_k\}
  \]
          Define $\gty{G'} = \gtypecomm{\pr_j}{\pr'}{\ell_i}{\basetype_i}{\gty{G}'_i}{i \in I}$, $\sty{T}_{\pr'} = \ltypesend{\pr_j}{\ell_i}{\basetype_i}{\sty{T}_{\pr'i}}{i \in I}$, and $\sty{T}_{\pr} = \mergeltype{\sty{T}_{\pr i}}{i \in I}$ for $\pr \in \role{R} \setminus \{\pr_j, \pr'\}$.
          If $\pr' = \pr_k$, then we have $\reduceexternal{(\project{\gty{G}}{\pr_k})}{\sty{T}_{\pr_k}}{\pr_j}{\ell}{\basetype}$, either by rule \RuleName{$\reduceexternalsymbol$-$\oplus$} if $\pr' = \pr_k$, or by \cref{merge-reduceexternals} otherwise.
          For $\pr \in \role{R} \setminus \{\pr_j, \pr_k\}$ we also have $(\project{\gty{G}}{\pr}) \extgeq \sty{T}_{\pr}$, either by the definition of merge, if $\pr = \pr'$, or by \cref{merge-merges} otherwise.
        \item Otherwise, we have $\gty{G} = \gtypecomm{\pp}{\pq}{\ell'_{i'}}{\basetype'_{i'}}{\gty{G}_{i'}}{i \in I'}$, with $\pr_j \not\in \{\pp, \pq\}$. We then appeal to the inductive hypothesis, as for \RuleName{$\reduceinternalsymbol$-base}.
    \item If the proof is by rule \RuleName{$\reduceinternalsymbol$-rec}, then we can write $\project{\gty{G}}{\pr_j} = \ltyperec{\ltypevar{X}}{\sty{T}'_{\pr_j}}$, and we have $\reduceinternal{\subst{\sty{T'}_{\pr_j}}{\ltypevar{X} \mapsto \sty{T}_{\pr_j}}}{\sty{T}_{\pr_j}}{\pr_k}{\ell}{\basetype}$.
        We proceed by induction on $\gty{G}$, noting that $G$ can only be either a recursive type, or a communication.
        \begin{itemize}
            \item If $\gty{G} = \ltyperec{\gtypevar{X}}{\gty{G'}}$, then we have $\sty{T}'_{\pr_j} = \project{\gty{G'}}{\sty{T}'_{\pr_j}}$. By \cref{unfold-project}, the projections $\project{\subst{\gty{G'}}{\gtypevar{X} \mapsto \gty{G}}}{\pr}$ are defined, and are obtained from the projections of $\gty{G}$ by unfolding. The result then follows from the inductive hypothesis, using \cref{inductive-predicates-unfold}.
            \item Otherwise, we have $\gty{G} = \gtypecomm{\pp}{\pq}{\ell'_{i'}}{\basetype'_{i'}}{\gty{G}_{i'}}{i \in I'}$, with $\pr_j \not\in \{\pp, \pq\}$. We then appeal to the inductive hypothesis, as for \RuleName{$\reduceinternalsymbol$-base}.
        \end{itemize}
      \end{itemize}
  \end{itemize}
\end{proof}

\sessionsubjred*
\begin{proof}
   Assume that $\mathcal{M} = {(\namepart{\pr_1}{\mathcal{C}_1}, \namepart{\pr_2}{\mathcal{C}_2}, \dots, \namepart{\pr_n}{\mathcal{C}_n})}$ has type $\gty{G}$,
   and write $\mathcal{M'} = {(\namepart{\pr_1}{\mathcal{D}_1}, \namepart{\pr_2}{\mathcal{D}_2}, \dots, \namepart{\pr_n}{\mathcal{D}_n})}$.
   There are two cases to consider.
   \begin{itemize}
     \item If $\beta = \tauaction_{\pr_j}$, then we have $\reduceaction{\mathcal{C}_j}{\mathcal{D}_j}{\tauaction}$, and $\mathcal{C}_i = \mathcal{D}_i$ for each $i \neq j$. Thus for all $i$, the configurations $\mathcal{C}_i$ and $\mathcal{D}_i$ have the same session type; for $i = j$ this is by \cref{subject-reduction}, while for $i \neq j$ this is trivial. Thus $\mathcal{M}'$ has type $\gty{G}$.
     \item If $\beta = \commaction{\pr_j}{\pr_k}{\msg{\ell}{v}}$ with $v : \basetype$, then we have $\reduceaction{\mathcal{C}_j}{\mathcal{D}_j}{\sendaction{\pr_k}{\msg{\ell}{v}}}$, $\reduceaction{\mathcal{C}_k}{\mathcal{D}_k}{\recvaction{\pr_j}{\msg{\ell}{v}}}$, and $\mathcal{C}_i = \mathcal{D}_i$ for each $i \not\in \{j, k\}$.
       By \cref{subject-reduction}, the configuration $\mathcal{C}_j$ has some session type $\sty{T}_{\pr_j}$ such that $\reduceinternal{(\project{G}{\pr_j})}{\sty{T}_{\pr_j}}{\pr_k}{\ell}{\basetype}$.
       By \cref{global-reduction}, it follows that there are session types $\sty{T}_{\pr}$ ($\pr \neq \pr_j$) and a global type $\gty{G}'$ over $\{\pr_1, \dots, \pr_n\}$, such that $\sty{T}_{\pr} = \project{\gty{G}'}{\pr}$ for all $\pr$, and such that
  \[
    \reduceexternal{(\project{\gty{G}}{\pr_k})}{\sty{T}_{\pr_k}}{\pr_j}{\ell}{\basetype}
    \qquad
    (\project{\gty{G}}{\pr}) \extgeq \sty{T}_{\pr}~\text{for}~\pr \not\in \{\pr_j, \pr_k\}
  \]
      It follows that $\mathcal{D}_{\pr}$ has session type $\sty{T}_{\pr}$ for all $\pr$; for $\pr = \pr_k$ this is by \cref{subject-reduction}, while for $\pr \not\in\{\pr_j, \pr_k\}$ this is by \cref{extgeq-subtype} and \cref{configuration-subsumption}.
      The session $\mathcal{M}'$ therefore has type $\gty{G}$.
   \end{itemize} 
\end{proof}

\liveness*
\begin{proof}
  For (1), first note that if $\queue{\sigma}_{ij}$ is non-empty, then $\pr_i$ can send a message to another participant.
  By fairness, such a reduction is present in the sequence, and indeed, unless $\pr_i$ receives a message as desired, there is some $i'' \geq i$ for which $\queue{\sigma}_{ji''}$ is empty.
  Then, by \cref{comp-recvs}, we have $\Recvs{\pp}{(\project{\gty{G}}{\pr_j})}$.
  We can now proceed by induction on the proof of $\Recvs{\pp}{(\project{\gty{G}}{\pr_j})}$.
  \begin{itemize}
    \item If the proof is by \RuleName{Recvs-base}, then $\project{G}{\pr_j}$ is an external choice on $\pp$, where $\gty{G}$ is the global type associated to the $i''$th configuration.
      We proceed by induction on the global type $\gty{G}$.
      Note that $\gty{G}$ can only be a communication, since otherwise, $\project{G}{\pr_j}$ would not be an external choice.
      We consider two cases.
      \begin{itemize}
        \item If $\gty{G} = \gtypecomm{\pp}{\pr_j}{\ell_j}{\basetype_j}{\gty{G}_j}{j \in J}$, then $\project{\gty{G}}{\pp}$ is an internal choice. The interpretation of the corresponding configuration then sends some $m$ to $\pr_j$, and thus by \cref{semantics-correct}, so does the configuration itself. Thus the global action $\commaction{\pp}{\pr_j}{m}$ appears in the reduction sequence by fairness.
        \item Otherwise, we have $\gty{G} = \gtypecomm{\pq}{\pq'}{\ell_j}{\basetype_j}{\gty{G}_j}{j \in J}$, with $\pr \not\in \{\pq, \pq'\}$.
          It follows from \cref{semantics-correct} that the global action $\commaction{\pq}{\pq'}{\msg{\ell_i}{v}}$ is enabled, for some $i \in I$ and $v : \basetype_i$.
          Thus that action appears in the reduction sequence, and the result follows from the inductive hypothesis.
      \end{itemize}
    \item If the proof is by \RuleName{Recvs-$\&$}, then $\project{G}{\pr_j}$ is an external choice on some other participant $\pq$.
      Just as in the previous case, a communication $\pq \to \pr_j$ is enabled, and thus appears in the reduction sequence; the result then follows from the inductive hypothesis.
    \item If the proof is by \RuleName{Recvs-rec}, then we again proceed by induction on $\gty{G}$. If $\gty{G}$ is a communication, then that communication appears in the reduction sequence, and we continue using the inductive hypothesis.
      Otherwise, $\gty{G}$ is a recursive type, and we can again continue using the inductive hypothesis.
  \end{itemize}

  For (2), we have $\queue{\sigma_j} = \consfront{\pp}{\msg{\ell}{v}}{\queue{\sigma}'}$, where $m = \msg{\ell}{v}$, and $v$ has some type $\basetype$.
  We then have that $\conftypedg{\config{\queue{\sigma'}}{t}{\queue{\rho}}}{\basetype_i}{\sty{T}}$ for some $\sty{T}$ such that $\reduceexternal{\sty{T}}{\project{\gty{G}}{\pr_j}}{\pp}{\ell}{\basetype}$, where $\gty{G}$ is the global type of $\mathcal{M}_{i + 1}$.
  To show that a suitable reduction appears in the given reduction sequence, we then proceed by induction on $\reduceexternal{\sty{T}}{\project{\gty{G}}{\pr_j}}{\pp}{\ell}{\basetype}$, in a similar manner to the proof of (1).
\end{proof}

\section{Comparison with the asynchronous subtyping of \cite{GPPSY2023}}

We prove that, for closed session types, our definition of asynchronous subtyping is equivalent to that of \cite{GPPSY2023}.
(Unlike us, \cite{GPPSY2023} do not give a definition of subtyping for non-closed session types.)
Our subtyping is therefore \emph{sound} and \emph{complete} in the sense of \cite{GPPSY2023}.
The equivalence is \cref{thm:gppsy-equivalence} below.

The definition in \cite{GPPSY2023} makes heavy use of session \emph{trees}.
These are generated coinductively as follows.
\[
\str{T}, \str{U} \Coloneqq \ltypeend
~~|~~
\ltypesend{\role{p}}{\ell_i}{\basetype_i}{\str{T}_i}{i \in I}
~~|~~
\ltyperecv{\role{p}}{\ell_i}{\basetype_i}{\str{T}_i}{i \in I}
\]

Every closed session type has an \emph{unfolding} as a session tree $\typetree{\sty{T}}$, defined by

\[
  \typetree{\sty{T}} =
  \begin{cases}
    \ltypeend
    & \text{if}~\unfold{\sty{T}} = \ltypeend
    \\
    \ltypesend{\role{p}}{\ell_i}{\basetype_i}{\typetree{\sty{T}_i}}{i \in I}
    & \text{if}~\unfold{\sty{T}} = \ltypesend{\role{p}}{\ell_i}{\basetype_i}{\sty{T}_i}{i \in I}
    \\
    \ltyperecv{\role{p}}{\ell_i}{\basetype_i}{\typetree{\sty{T}_i}}{i \in I}
    & \text{if}~\unfold{\sty{T}} = \ltyperecv{\role{p}}{\ell_i}{\basetype_i}{\sty{T}_i}{i \in I}
  \end{cases}
\]

\subsection{The subtyping definition of \cite{GPPSY2023}}

In this section, we spell out the definition of asynchronous subtyping used in \cite{GPPSY2023}.
We write $\subtype_{\mathrm{GPPSY}}$ for this relation, to distinguish the definition in \cite{GPPSY2023} from ours.

\begin{definition}
  A session tree is \emph{single-input} (SI) when each external choice has only one branch; 
  \emph{single-output} (SO) when each internal choice has only one branch;
  \emph{single-input-single-output} (SISO) when it is both.
  \[
\begin{array}{l@{~}c@{~}l@{~}c@{~}l@{~}c@{~}l}
\str{T}^{\SI} &\Coloneqq& \ltypeend
&|&
\ltypesend{\role{p}}{\ell_i}{\basetype_i}{\str{T}^{\SI}_i}{i \in I}
&|&
\ltyperecvone{\role{p}}{\ell}{\basetype}{\str{T}^{\SI}}
\\
\str{T}^{\SO} &\Coloneqq& \ltypeend
&|&
\ltypesendone{\role{p}}{\ell}{\basetype}{\str{T}^{\SO}}
&|&
\ltyperecv{\role{p}}{\ell_i}{\basetype_i}{\str{T}^{\SO}_i}{i \in I}
\\
\str{T}^{\SISO} &\Coloneqq& \ltypeend
&|&
\ltypesendone{\role{p}}{\ell}{\basetype}{\str{T}^{\SISO}}
&|&
\ltyperecvone{\role{p}}{\ell}{\basetype}{\str{T}^{\SISO}}
\end{array}
\]
  If $\str{T}$ is a session tree, we write $\siset{\str{T}}$ and $\soset{\str{T}}$ for the sets of SI and SO \emph{decompositions} of $\str{T}$ respectively:
  \begin{gather*}
    \siset{\ltypeend} = \{\ltypeend\}
    \quad
    \begin{array}{r@{~}c@{~}l}
    \siset{\ltypesend{\pp}{\ell_i}{\basetype_i}{\str{T}_i}{i \in I}}
    &=& \{\ltypesend{\pp}{\ell_i}{\basetype_i}{\str{T}'_i}{i \in I} \mid \str{T}'_i \in \siset{\str{T}_i}\}
    \\
    \siset{\ltyperecv{\pp}{\ell_i}{\basetype_i}{\str{T}_i}{i \in I}}
    &=& \bigcup_{i \in I} \{\ltyperecvone{\pp}{\ell_i}{\basetype_i}{\str{T}'} \mid \str{T}' \in \siset{\str{T}_i}\}
    \end{array}
    \\
    \soset{\ltypeend} = \{\ltypeend\}
    \quad
    \begin{array}{r@{~}c@{~}l}
    \soset{\ltypesend{\pp}{\ell_i}{\basetype_i}{\str{T}_i}{i \in I}}
    &=& \bigcup_{i \in I} \{\ltypesendone{\pp}{\ell_i}{\basetype_i}{\str{T}'} \mid \str{T}' \in \soset{\str{T}_i}\}
    \\
    \soset{\ltyperecv{\pp}{\ell_i}{\basetype_i}{\str{T}_i}{i \in I}}
    &=& \{\ltyperecv{\pp}{\ell_i}{\basetype_i}{\str{T}'_i}{i \in I} \mid \str{T}'_i \in \soset{\str{T}_i}\}
    \end{array}
  \end{gather*}
\end{definition}

\begin{definition}
  We let $\actset{\str{T}}$ be the smallest family of sets indexed by SISO session trees $\str{T}$, such that the following hold.
  \[
    \actset{\ltypeend} = \emptyset
    \quad
    \actset{\ltypesendone{\pp}{\ell}{\basetype}{\str{T}}} = \{\pp!\} \cup \actset{\str{T}}
    \quad
    \actset{\ltyperecvone{\pp}{\ell}{\basetype}{\str{T}}} = \{\pp?\} \cup \actset{\str{T}}
  \]
\end{definition}

\begin{definition}
  For every participant $\pp$, we define two sets of \emph{prefixes} inductively, by the following, where $\pq$ ranges over participants distinct from $\pp$.
  \begin{align*}
    \mathcal{A}^{(\pp)}
    &\Coloneqq
    \ltyperecvoneprefix{\pq}{\ell}{\basetype}
    ~|~
    \ltyperecvone{\pq}{\ell}{\basetype}{\mathcal{A}^{(\pp)}}
    \\
    \mathcal{B}^{(\pp)}
    &\Coloneqq
    \ltyperecvoneprefix{\pr}{\ell}{\basetype}
    ~|~
    \ltypesendoneprefix{\pp}{\ell}{\basetype}
    ~|~
    \ltyperecvone{\pr}{\ell}{\basetype}{\mathcal{B}^{(\pp)}}
    ~|~
    \ltypesendone{\pq}{\ell}{\basetype}{\mathcal{B}^{(\pp)}}
  \end{align*}
\end{definition}

\begin{definition}
    The \emph{refinement} order $\refines$ on SISO session trees is generated coinductively by the following rules.
    \begin{gather*}
        \begin{prooftree}
            \hypo{\str{T} \refines \str{U}}
            \infer[double]1{\ltyperecvone{\pp}{\ell}{\basetype}{\str{T}} \refines \ltyperecvone{\pp}{\ell}{\basetype}{\str{U}}}
        \end{prooftree}
        \quad
        \begin{prooftree}
            \hypo{\str{T} \refines \mathcal{A}^{(\pp)}.\,\str{U}}
            \hypo{\actset{\str{T}} = \actset{\mathcal{A}^{(\pp)}.\,\str{U}}}
            \infer[double]2{\ltyperecvone{\pp}{\ell}{\basetype}{\str{T}} \refines \mathcal{A}^{(\pp)}.\,\ltyperecvone{\pp}{\ell}{\basetype}{\str{U}}}
        \end{prooftree}
        \\
        \begin{prooftree}
            \hypo{\str{T} \refines \str{U}}
            \infer[double]1{\ltypesendone{\pp}{\ell}{\basetype}{\str{T}} \refines \ltypesendone{\pp}{\ell}{\basetype}{\str{U}}}
        \end{prooftree}
        \quad
        \begin{prooftree}
            \hypo{\str{T} \refines \mathcal{B}^{(\pp)}.\,\str{U}}
            \hypo{\actset{\str{T}} = \actset{\mathcal{B}^{(\pp)}.\,\str{U}}}
            \infer[double]2{\ltypesendone{\pp}{\ell}{\basetype}{\str{T}} \refines \mathcal{B}^{(\pp)}.\,\ltypesendone{\pp}{\ell}{\basetype}{\str{U}}}
        \end{prooftree}
        \\
        \begin{prooftree}
            \hypo{}
            \infer[double]1{\ltypeend \refines \ltypeend}
        \end{prooftree}
    \end{gather*}
\end{definition}

\begin{definition}
    The \emph{asynchronous subtyping relation} for closed session types, as defined in \cite{GPPSY2023}, is given by
    \[
      \sty{T} \subtype_{\mathrm{GPPSY}} \sty{U}
      ~\Leftrightarrow~
      \typetree{\sty{T}} \subtype_{\mathrm{GPPSY}} \typetree{\sty{U}}
    \]
    where for session trees, $\subtype_{\mathrm{GPPSY}}$ is given by
    \[
      \str{T} \subtype_{\mathrm{GPPSY}} \str{U}
      ~\Leftrightarrow~
      \forall \str{T}' \in \soset{\str{T}}, \str{U}' \in \siset{\str{U}}.\,
      \exists \str{T}'' \in \siset{\str{T}'}, \str{U}'' \in \soset{\str{U}'}.\,
      \str{T}'' \refines \str{U}''
    \]
\end{definition}

\subsection{Asynchronous subtyping for session trees}

To show the equivalence with $\subtype_{\mathrm{GPPSY}}$, we first define a subtyping relation $\subtype$ for session trees, in the same style that we use for session types, such that $\sty{T} \subtype \sty{U}$ is equivalent to $\typetree{\sty{T}} \subtype \typetree{\sty{U}}$.
It will then remain to show that, for session trees, $\subtype$ is the same relation as $\subtype_{\mathrm{GPPSY}}$.

First, we define relations and predicates for session trees, analogous to those we have for session types.
The definitions are inductively, and are given by the rules of \cref{fig:inductive-trees}.
\begin{figure}[t]
\begin{minipage}[t]{0.65\textwidth}
\begin{gather*}
  \shortintertext{\fbox{$\reduceinternal{\str{U}}{\str{U}'}{\pp}{\ell}{\basetype}$}}
  \\[-2.5ex]
  \begin{array}{l}
  \RuleName{$\reduceinternalsymbol$-base}
  ~
  \begin{prooftree}
    \hypo{}
    \infer1
    {\reduceinternal{\ltypesend{\pp}{\ell_i}{\basetype_i}{\str{U}_i}{i \in I}}{\str{U}_i}{\pp}{\ell_i}{\basetype_i}}
  \end{prooftree}
  \end{array}
  \\
  \begin{array}{l}
  \RuleName{$\reduceinternalsymbol$-$\oplus$}
  ~
  \begin{prooftree}
    \hypo{\pp \neq \pq}
    \hypo{J \subseteq I}
    \hypo{{\reduceinternal{\str{U}_i}{\str{U}'_i}{\pp}{\ell}{\basetype}}~\text{for all}~i \in J}
    \infer3
    {\reduceinternal{\ltypesend{\pq}{\ell'_i}{\basetype'_i}{\str{U}_i}{i \in I}}{\ltypesend{\pq}{\ell'_i}{\basetype'_i}{\str{U}'_{i}}{i \in J}}{\pp}{\ell}{\basetype}}
  \end{prooftree}
  \end{array}
  \\[0ex]
  \begin{array}{l}
  \RuleName{$\reduceinternalsymbol$-$\&$}
  ~
  \begin{prooftree}
    \hypo{{\reduceinternal{\str{U}_i}{\str{U}'_i}{\pp}{\ell}{\basetype}}~\text{for all}~i \in I}
    \infer1
    {\reduceinternal{\ltyperecv{\pq}{\ell'_i}{\basetype'_i}{\str{U}_i}{i \in I}}{\ltyperecv{\pq}{\ell'_i}{\basetype'_i}{\str{U}'_{i}}{i \in J}}{\pp}{\ell}{\basetype}}
  \end{prooftree}
  \end{array}
\end{gather*}
\vspace{-5ex}
\end{minipage}%
\begin{minipage}[t]{0.35\textwidth}
\begin{gather*}
  \shortintertext{\fbox{$\Sends{\pp}(\str{T})$}}
  \begin{array}{l}
  \RuleName{Sends-base}
  \\[1ex]
  ~\begin{prooftree}
    \hypo{}
    \infer1
    {\Sends{\pp}(\ltypesend{\pp}{\ell_i}{\basetype_i}{\str{T}_i}{i \in I})}
  \end{prooftree}
  \end{array}
  \\
  \begin{array}{l}
  \RuleName{Sends-$\oplus$}
  \\[1ex]
  ~\begin{prooftree}
    \hypo{\Sends{\pp}(\str{T}_j)~~\text{for all}~j \in J}
    \infer1
    {\Sends{\pp}(\ltypesend{\pq}{\ell_j}{\basetype_i}{\str{T}_j}{j \in J})}
  \end{prooftree}
  \end{array}
\end{gather*}
\end{minipage}

\begin{minipage}[t]{0.65\textwidth}
\begin{gather*}
  \shortintertext{\fbox{$\reduceexternal{\str{T}}{\str{T}'}{\pp}{\ell}{\basetype}$}}
  \\[-2.5ex]
  \begin{array}{l}
  \RuleName{$\reduceexternalsymbol$-base}
  ~
  \begin{prooftree}
    \hypo{}
    \infer1
    {\reduceexternal{\ltyperecv{\pp}{\ell_i}{\basetype_i}{\str{T}_i}{i \in I}}{\str{T}_i}{\pp}{\ell_i}{\basetype_i}}
  \end{prooftree}
  \end{array}
  \\
  \begin{array}{l}
  \RuleName{$\reduceexternalsymbol$-$\&$}
  ~
  \begin{prooftree}
    \hypo{\pp \neq \pq}
    \hypo{J \subseteq I}
    \hypo{{\reduceexternal{\str{T}_i}{\str{T}'_i}{\pp}{\ell}{\basetype}}~\text{for all}~i \in J}
    \infer3
    {\reduceexternal{\ltyperecv{\pq}{\ell'_i}{\basetype'_i}{\str{T}_i}{i \in I}}{\ltyperecv{\pq}{\ell'_i}{\basetype'_i}{\str{T}'_{i}}{i \in J}}{\pp}{\ell}{\basetype}}
  \end{prooftree}
  \end{array}
  \\[0ex]
  \begin{array}{l}
  \RuleName{$\reduceexternalsymbol$-$\oplus$}
  ~
  \begin{prooftree}
    \hypo{{\reduceexternal{\str{T}_i}{\str{T}'_i}{\pp}{\ell}{\basetype}}~\text{for all}~i \in I}
    \infer1
    {\reduceexternal{\ltypesend{\pq}{\ell'_i}{\basetype'_i}{\str{T}_i}{i \in I}}{\ltypesend{\pq}{\ell'_i}{\basetype'_i}{\str{T}'_{i}}{i \in J}}{\pp}{\ell}{\basetype}}
  \end{prooftree}
  \end{array}
\end{gather*}
\end{minipage}%
\begin{minipage}[t]{0.35\textwidth}
\begin{gather*}
  \shortintertext{\fbox{$\Recvs{\pp}(\str{U})$}}
  \begin{array}{l}
  \RuleName{Recvs-base}
  \\[1ex]
  ~\begin{prooftree}
    \hypo{}
    \infer1
    {\Recvs{\pp}(\ltyperecv{\pp}{\ell_i}{\basetype_i}{\str{U}_i}{i \in I})}
  \end{prooftree}
  \end{array}
  \\
  \begin{array}{l}
  \RuleName{Recvs-$\&$}
  \\[1ex]
  ~\begin{prooftree}
    \hypo{\Recvs{\pp}(\str{U}_j)~~\text{for all}~j \in J}
    \infer1
    {\Recvs{\pp}(\ltypesend{\pq}{\ell_j}{\basetype_i}{\str{U}_j}{j \in J})}
  \end{prooftree}
  \end{array}
\end{gather*}
\end{minipage}

\caption{Four inductively defined relations and predicates on session trees}
\label{fig:inductive-trees}
\end{figure}

We use this to define asynchronous subtyping for session trees.
    \begin{definition}
    \emph{Asynchronous subtyping} is the largest relation $\subtype$ on session trees, such that the following hold when $\str{T} \subtype \str{U}$.
    \begin{enumerate}
        \item If ${\str{T}} = \ltypesend{\pp}{\ell_i}{\basetype_i}{\str{T}_i}{i \in I}$, then for every $i \in I$, there is some $\str{U}_i$ such that $\reduceinternal{\str{U}}{\str{U}_i}{\pp}{\ell_i}{\basetype_i}$ and $\str{T}_i \subtype \str{U}_i$.
        \item If ${\str{T}} = \ltyperecv{\pp}{\ell_i}{\basetype_i}{\str{T}_i}{i \in I}$, then $\Recvs{\pp}(\str{U})$.
        \item If ${\str{U}} = \ltypesend{\pp}{\ell_i}{\basetype_i}{\str{U}_i}{i \in I}$, then $\Sends{\pp}(\str{T})$.
        \item If ${\str{U}} = \ltyperecv{\pp}{\ell_i}{\basetype_i}{\str{U}_i}{i \in I}$, then for every $i \in I$, there is some $\str{T}_i$ such that $\reduceexternal{\str{T}}{\str{T}_i}{\pp}{\ell_i}{\basetype_i}$ and $\str{T}_i \subtype \str{U}_i$.
    \end{enumerate}
\end{definition}

\begin{lemma}\label{inductive-tree-type}
  In the following, all of the session types are assumed to be closed.
  \begin{enumerate}
    \item If $\reduceinternal{\sty{T}}{\sty{U}}{\pp}{\ell}{\basetype}$ then $\reduceinternal{\typetree{\sty{T}}}{\typetree{\sty{U}}}{\pp}{\ell}{\basetype}$.
      If $\reduceinternal{\typetree{\sty{T}}}{\str{U}}{\pp}{\ell}{\basetype}$ then there is some $\sty{U}$ such that $\typetree{\sty{U}} = \str{U}$ and $\reduceinternal{\sty{T}}{\sty{U}}{\pp}{\ell}{\basetype}$.
    \item $\Sends{\pp}(\sty{T})$ is equivalent to $\Sends{\pp}(\typetree{\sty{T}})$.
    \item If $\reduceexternal{\sty{T}}{\sty{U}}{\pp}{\ell}{\basetype}$ then $\reduceexternal{\typetree{\sty{T}}}{\typetree{\sty{U}}}{\pp}{\ell}{\basetype}$.
      If $\reduceexternal{\typetree{\sty{T}}}{\str{U}}{\pp}{\ell}{\basetype}$ then there is some $\sty{U}$ such that $\typetree{\sty{U}} = \str{U}$ and $\reduceexternal{\sty{T}}{\sty{U}}{\pp}{\ell}{\basetype}$.
    \item $\Recvs{\pp}(\sty{T})$ is equivalent to $\Recvs{\pp}(\typetree{\sty{T}})$.
  \end{enumerate}
\end{lemma}
\begin{proof}
  The first part of (1) is a trivial induction on the derivation of $\reduceinternal{\sty{T}}{\sty{U}}{\pp}{\ell}{\basetype}$.
  The second part of (1) is an easy induction on $\reduceinternal{\typetree{\sty{T}}}{\str{U}}{\pp}{\ell}{\basetype}$, but using \cref{inductive-predicates-unfold}.
  For (2), both directions are again easy inductions, but in both directions we use \cref{inductive-predicates-unfold}.
  The proofs of (3) and (4) are similar.
\end{proof}

\begin{lemma}\label{subtype-to-subtree}
    If $\sty{T}$ and $\sty{U}$ are closed session types, then $\sty{T} \subtype \sty{U}$ is equivalent to $\typetree{\sty{T}} \subtype \typetree{\sty{U}}$.
\end{lemma}
\begin{proof}
  This is an immediate corollary of \cref{inductive-tree-type}.
\end{proof}

\subsection{Single-input and single-output}
We show that our definition of session tree subtyping is equivalent to one involving SISO session trees.
\begin{lemma}\label{subtree-congruences}\
    \begin{enumerate}
        \item If $\str{U} = \ltyperecv{\pp}{\ell_i}{\basetype_i}{\str{U}_i}{i \in I}$, then $\str{T} \subtype \str{U}$ holds iff, for every $i \in I$, there is some $\str{T}_i$ such that $\reduceexternal{\str{T}}{\str{T}_i}{\pp}{\ell_i}{\basetype_i}$ and $\str{T}_i \subtype \str{U}_i$.
        \item If $\str{T} = \ltypesend{\pp}{\ell_i}{\basetype_i}{\str{T}_i}{i \in I}$, then $\str{T} \subtype \str{U}$ holds iff, for every $i \in I$, there is some $\str{U}_i$ such that $\reduceinternal{\str{U}}{\str{U}_i}{\pp}{\ell_i}{\basetype_i}$ and $\str{T}_i \subtype \str{U}_i$.
    \end{enumerate}
\end{lemma}
\begin{proof}
  The proof is exactly the same as the proof \cref{subtype-congruences}, but without type variables and recursive types.
\end{proof}

\begin{lemma}\label{so-is-subtype}
    Let $\str{T}$ be a session tree.
    \begin{enumerate}
        \item $\str{T}' \subtype \str{T}$ for all $\str{T}' \in \soset{\str{T}}$.
        \item $\str{T} \subtype \str{T}'$ for all $\str{T}' \in \siset{\str{U}}$.
    \end{enumerate}
\end{lemma}
\begin{proof}
    We give the proof of (1); the proof of (2) is similar.
    We show each of the cases of the definition of $\str{T}' \subtype \str{T}$ in turn.
    \begin{enumerate}
        \item If $\str{T}' = \ltypesendone{\pp}{\ell}{\basetype}{\str{U}'}$, then $\reduceinternal{\str{T}}{\str{U}}{\pp}{\ell}{\basetype}$ for some $\str{U}$ with $\str{U'} \in \soset{\str{U}}$.
        \item If $\str{T}' = \ltyperecv{\pp}{\ell_i}{\basetype_i}{\str{U}'_i}{i \in I}$, then $\Recvs{\pp}(\str{T})$ holds trivially.
        \item If $\str{T} = \ltypesend{\pp}{\ell_i}{\basetype_i}{\str{U}'_i}{i \in I}$, then $\Sends{\pp}(\str{T})$ holds trivially.
        \item If $\str{T} = \ltyperecv{\pp}{\ell_i}{\basetype_i}{\str{U}'_i}{i \in I}$, then for each $i$ we have $\reduceexternal{\str{T}}{\str{U}}{\pp}{\ell_i}{\basetype_i}$ for some $\str{U}$ with $\str{U'}_i \in \soset{\str{U}}$.
    \end{enumerate}
\end{proof}

\begin{lemma}\label{subtype-finite}
    \begin{enumerate}
        \item Assume that, for every $\str{T}' \in \soset{\str{T}}$, there is some $\str{U}' \in \soset{U}$ such that there exists a $\str{S'}$ with $\reduceinternal{\str{U}'}{\str{S}'}{\pp}{\ell}{\basetype}$ and $\str{T}' \subtype \str{S}'$.
          Then there is some $\str{S}$ such that $\reduceinternal{\str{U}}{\str{S}}{\pp}{\ell}{\basetype}$ and, for all $\str{T}' \in \soset{\str{T}}$ there exists $\str{S'} \in \soset{\str{S}}$ such that $\str{T}' \subtype \str{S}'$.
        \item Assume that, for every $\str{T}' \in \soset{\str{T}}$, there is some $\str{U}' \in \soset{U}$ such that there exists a $\str{S'}$ with $\reduceexternal{\str{T}'}{\str{S}'}{\pp}{\ell}{\basetype}$ and $\str{S}' \subtype \str{U}'$.
          Then there is some $\str{S}$ such that $\reduceexternal{\str{T}}{\str{S}}{\pp}{\ell}{\basetype}$ and, for all $\str{T}' \in \soset{\str{T}}$ there exists $\str{S'} \in \soset{\str{S}}$ such that $\str{S}' \subtype \str{U}'$.
    \end{enumerate}
\end{lemma}
\begin{proof}
  We give the proof of (1); the proof of (2) is similar.
  Let $\soset{\str{T}} = \{\str{T'}_k \mid k \in K\}$.
  We first note that we can assume there is some natural number $h$, 
  for every $k \in K$, there exists $\str{U}'_k \in \soset{U}$ and $\str{S'}_k$ with $\reduceinternal{\str{U}'_k}{\str{S}'_k}{\pp}{\ell}{\basetype}$ and $\str{T}'_k \subtype \str{S}'_k$, where the proof of $\reduceinternal{\str{U}'_k}{\str{S}'_k}{\pp}{\ell}{\basetype}$ has height at most $h$.
  For if this were not the case, then for every $h$, there would be some $k \in K$, such that for every $\str{U}' \in \soset{U}$ and $\str{S'} \in \soset{S}$ with $\str{T}'_k \subtype \str{S}'$, any derivation of $\reduceinternal{\str{U}'}{\str{S}'}{\pp}{\ell}{\basetype}$ has height greater than $h$.
  This would enable us to construct some $\str{T}' \in \soset{\str{T}}$ such that there no suitable $\str{U}'$, by inspecting $\str{T}$:
  \begin{itemize}
    \item If $\str{T} = \ltypeend$, then take $\str{T}' = \ltypeend$.
    \item If $\str{T} = \ltypesend{\pp}{\ell_i}{\basetype_i}{\str{T}_i}{i \in I}$, then since $I$ is finite, there is some $i$ such that for every $h$, there is some $\str{T}'_0$ as above, that has the form $\ltypesendone{\pp}{\ell_i}{\basetype_i}{\str{T}''_0}$. Thus we can take $\str{T}' = \ltypesendone{\pp}{\ell_i}{\basetype_i}{\str{T}''}$ for some $\str{T''}$.
    \item If $\str{T} = \ltyperecv{\pp}{\ell_i}{\basetype_i}{\str{T}_i}{i \in I}$, then we can take $\str{T}' = \ltyperecv{\pp}{\ell_i}{\basetype_i}{\str{T}''_i}{i \in I}$ for some $\str{T''}_i$.
  \end{itemize}

  We can therefore show the result by well-founded induction on $h$, noting that all of the assumed derivations of $\reduceinternalsymbol$ are by the same rule.
  \begin{itemize}
    \item If the derivations are by rule \RuleName{$\reduceinternalsymbol$-base}, then the result is trivial.
    \item If the derivations are by rule \RuleName{$\reduceinternalsymbol$-$\oplus$}, then we have
      \begin{gather*}
        \str{U} = \ltypesend{\pq}{\ell_i}{\basetype_i}{\str{U}_i}{i \in I}
        \qquad
        \str{U}'_k = \ltypesendone{\pq}{\ell_{i_k}}{\basetype_{i_k}}{\str{U}'_{ki_k}}
        \\
        \pp \neq \pq
        \qquad
        \reduceinternal{\str{U}'_i}{\str{U}'_{ki_k}}{\pp}{\ell}{\basetype}~\text{for all}~k \in K
      \end{gather*}
      Define $I' = \{i_k \mid k \in K\}$.
      For every $i \in I'$, let $\str{S}_i$ be any session tree such that $\reduceinternal{\str{U}'_i}{\str{S}_i}{\pp}{\ell}{\basetype}$, and such that for all $k$, there is some $\str{U}''_i \in \soset{\str{S}_i}$ with $\str{T}'_k \subtype \str{S}'_k$.
      These exist by the inductive hypothesis.
      We can then take $\str{U}' = \ltypesend{\pq}{\ell_i}{\basetype_i}{\str{U}'_i}{i \in I'}$, using finiteness of $h$ to prove $\reduceinternal{\str{U}}{\str{U}'}{\pp}{\ell}{\basetype}$ by \RuleName{$\reduceinternalsymbol$-$\oplus$}.
    \item If the derivation is by rule \RuleName{$\reduceinternalsymbol$-$\&$}, then so is the derivation of every other $\reduceinternal{\str{T}}{\str{U}_k}{\pp}{\ell}{\basetype}$.
      We have
      \begin{gather*}
        \sty{T} = \ltyperecv{\pq}{\ell_i}{\basetype_i}{\sty{T}_i}{i \in I}
        \qquad
        \sty{U}_k = \ltyperecv{\pq}{\ell_i}{\basetype_i}{\sty{U}_{ki}}{i \in I}
        \\
        \reduceinternal{\sty{T}_i}{\sty{U}_{ki}}{\pp}{\ell}{\basetype}~\text{for all}~k \in K, i \in I
      \end{gather*}
      The proof is similar to the previous case, taking $\sty{U}' = \ltyperecv{\pq}{\ell_i}{\basetype_i}{\sty{U}'_i}{i \in I'}$ for some $\sty{U}'_i$.
  \end{itemize}
\end{proof}

\begin{lemma}\label{subtype-so-si}\
    \begin{enumerate}
        \item $\str{T} \subtype \str{U}$ holds if and only if for all $\str{T}' \in \soset{\str{T}}$, there exists $\str{U}' \in \soset{\str{U}}$ such that $\str{T}' \subtype \str{U}'$.
        \item $\str{T} \subtype \str{U}$ holds if and only if for all $\str{U}' \in \siset{\str{U}}$, there exists $\str{T}' \in \siset{\str{T}}$ such that $\str{T}' \subtype \str{U}'$.
    \end{enumerate}
\end{lemma}
\begin{proof}
  We give the proof of (1), the proof of (2) is similar.
  For the only if direction, we note that by \cref{so-is-subtype} and transitivity of subtyping, it is enough to show that, when $\str{T}'$ is SO, $\str{T}' \subtype \str{U}$ implies there exists an SO $\str{U}'$ such that $\str{T'} \subtype \str{U'}$.
  We construct $\str{U}'$ by inspecting $\str{U}$.
  \begin{itemize}
    \item If $\str{U} = \ltypeend$, then take $\str{U'} = \ltypeend$.
    \item If $\str{U} = \ltypesend{\pp}{\ell_i}{\basetype_i}{\str{U}_i}{i \in I}$, then $\Sends{\pp}(\str{T}')$. We proceed by induction on $\Sends{\pp}(\str{T}')$. In the base case, we have $\str{T} = \ltypesendone{\pp}{\ell'}{\basetype'}{\str{T}'}$, where $(\ell', \basetype') = (\ell_i, \basetype_i)$ and $\str{T'} \subtype \str{U}_i$. We then take $\str{U'} = \ltypesendone{\pp}{\ell}{\basetype}{\str{U}'_i}$, where $\str{T}' \subtype \str{U}'_i$. In the indcutive step, we have $\str{T} = \ltypesendone{\pq}{\ell}{\basetype}{\str{T}'}$ with $\pp \neq \pq$. Thus there is some $\str{S} \suptype \str{T}'$ such that $\reduceinternal{\str{U}}{\str{S}}{\pq}{\ell}{\basetype}$. $\str{S}$ necessarily has the form $\ltypesend{\pp}{\ell_i}{\basetype_i}{\str{U'}_i}{i \in J}$ with $J \subseteq I$ and $\reduceinternal{\str{U}_i}{\str{U}'_i}{\pq}{\ell}{\basetype}$.
      The inductive hypothesis then provides us with an element of $J$, such that we can take $\str{U}' = \ltypesendone{\pp}{\ell}{\basetype}{\str{U}'_j}$, where $\str{T}' \subtype \str{U}'_j$.
    \item If $\str{U} = \ltyperecv{\pp}{\ell_i}{\basetype_i}{\str{U}_i}{i \in I}$, for each $i$ there is some $\str{T}_i$ such that $\reduceexternal{\str{T}}{\str{T}_i}{\pp}{\ell_i}{\basetype_i}$. Take $\str{U}' = \ltyperecv{\pp}{\ell_i}{\basetype_i}{\str{U}'_i}{i \in I}$, where $\str{U}'_i \in \soset{\str{U}_i}$ satisfies $\str{U}_i \subtype \str{T}_i$.
  \end{itemize}
 
  For the other direction, suppose that,
  for all $\str{T}' \in \soset{\str{T}}$, there exists $\str{U}' \in \soset{\str{U}}$ such that $\str{T}' \subtype \str{U}'$.
  We show that $\str{T} \subtype \str{U}$ by verifying each case of the definition of subtyping.
  \begin{enumerate}
    \item If $\str{T} = \ltypesend{\pp}{\ell_i}{\basetype_i}{\str{T}_i}{i \in I}$, then for every $\str{T}'_i \in \soset{\str{T}_i}$, there is some $\str{U}_i \in \soset{U}$ such that there exists a $\str{U}'_i$ with $\reduceinternal{\str{U}_i}{\str{U}'_i}{\pp_i}{\ell_i}{\basetype_i}$ and $\str{T}'_i \subtype \str{U}'_i$.
       The result follows from \cref{subtype-finite}.
    \item If $\str{T} = \ltyperecv{\pp}{\ell_i}{\basetype_i}{\str{T}_i}{i \in I}$, then since $\soset{\str{T}}$ is inhabited, there is some $\str{U}' \in \soset{U}$ such that $\Recvs{\pp}(\str{U}')$. The latter implies $\Recvs{\pp}(\str{U})$.
    \item If $\str{U} = \ltypesend{\pp}{\ell_i}{\basetype_i}{\str{U}_i}{i \in I}$, then for every $\str{T'} \in \soset{\str{T}}$ we have $\Sends{\pp}(\str{T}')$; this implies $\Sends{\pp}(\str{T})$.
    \item If $\str{U} = \ltyperecv{\pp}{\ell_i}{\basetype_i}{\str{U}_i}{i \in I}$, then for every $i \in I$ and $\str{T'} \in \soset{\str{T}}$ there is some $\str{T'}_i$ such that $\reduceexternal{\str{T}'}{\str{T}'_i}{\pp}{\ell_i}{\basetype_i}$, and some $\str{U}'_i \in \soset{\str{U}_i}$ such that $\str{T}'_i \subtype \str{U}'_i$.
       The result follows from \cref{subtype-finite}.
  \end{enumerate}
\end{proof}

\begin{lemma}\label{subtype-siso}
    $\str{T} \subtype \str{U}$ holds if and only if
    \[
      \forall \str{T}' \in \soset{\str{T}}, \str{U}' \in \siset{\str{U}}.\,
      \exists \str{T}'' \in \siset{\str{T}'}, \str{U}'' \in \soset{\str{U}'}.\,
      \str{T}'' \subtype \str{U}''
    \]
\end{lemma}
\begin{proof}
   This is immediate from \cref{subtype-so-si}.
\end{proof}

\subsection{Subtyping and refinement}
In light of the above results, to prove that our definition of subtyping is equivalent to that of \cite{GPPSY2023}, it is enough to show that our subtyping $\subtype$ coincides with refinement $\refines$ for SISO session trees.

\begin{lemma}\label{transition-prefix}
  In the following, assume all session trees are SISO.
  \begin{enumerate}
    \item $\reduceinternal{\str{T}}{\str{T}'}{\pp}{\ell}{\basetype}$ if and only if either $\str{T} = \ltypesendone{\pp}{\ell}{\basetype}{h'}$, or there exist $\mathcal{B}^{(\pp)}$ and $\str{T}''$ such that $\str{T} = \mathcal{B}^{(\pp)}.\ltypesendone{\pp}{\ell}{\basetype}{\str{T}''}$ and $\str{T}' = \mathcal{B}^{(\pp)}.\str{T}''$.
    \item $\Recvs{\pp}(\str{T})$ if and only if there exist $\ell$, $\basetype$, $\str{T}'$ such that either $\str{T} = \ltyperecvone{\pp}{\ell}{\basetype}{\str{T}'}$, or there exist $\mathcal{A}^{(\pp)}$ and $\str{T}''$ such that $\str{T} = \mathcal{A}^{(\pp)}.\ltyperecvone{\pp}{\ell}{\basetype}{\str{T}''}$ and $\str{T}' = \mathcal{A}^{(\pp)}.\str{T}''$.
  \end{enumerate}
\end{lemma}
\begin{proof}
  Both of these are easy inductions.
\end{proof}

\begin{lemma}\label{actset-lemma}
  If $\str{T}$ and $\str{U}$ are SISO, and $\str{T} \subtype \str{U}$, then $\actset{\str{T}} = \actset{\str{U}}$.
\end{lemma}
\begin{proof}
  We show $\actset{\str{T}} \subseteq \actset{\str{U}}$; the proof of $\actset{\str{U}} \subseteq \actset{\str{T}}$ is similar.
  \begin{itemize}
    \item If $\str{T} = \ltypeend$, then $\actset{\str{T}} = \emptyset$, so this is trivial.
    \item If $\str{T} = \ltypesendone{\pp}{\ell}{\basetype}{\str{T}'}$, then $\reduceinternal{\str{U}}{\str{U}'}{\pp}{\ell}{\basetype}$ for some $\str{U}'$ such that $\str{T}' \subtype \str{U}'$.
      Thus $\actset{\str{T}} = \{\pp!\} \cup \actset{\str{T}'} \subseteq \{\pp!\} \cup \actset{\str{U}'} \subseteq \actset{\str{U}}$.
    \item If $\str{T} = \ltyperecvone{\pp}{\ell}{\basetype}{\str{T}'}$, then $\Recvs{\pp}(\str{U})$, so $\prectype{\str{U}}{\ltyperecvone{\pp}{\ell}{\basetype}{\str{U}'}}$ for some $\str{U}'$ such that $\str{T}' \subtype \str{U}'$.
      We then have that $\actset{\str{T}} = \{\pp?\} \cup \actset{\str{T}'} \subseteq \{\pp?\} \cup \actset{\str{U}'} \subseteq \actset{\str{U}}$.
  \end{itemize}
\end{proof}

\begin{lemma}\label{subtyping-prefix-congruences}
  In the following, assume all session trees are SISO.
  \begin{enumerate}
    \item If $\str{T} \subtype \str{U}$, then $\ltyperecvone{\pp}{\ell}{\basetype}{\str{T}} \subtype \ltyperecvone{\pp}{\ell}{\basetype}{\str{U}}$
    \item If $\str{T} \subtype \mathcal{A}^{(\pp)}.\str{U}$, then $\ltyperecvone{\pp}{\ell}{\basetype}{\str{T}} \subtype \mathcal{A}^{(\pp)}.\ltyperecvone{\pp}{\ell}{\basetype}{\str{U}}$
    \item If $\str{T} \subtype \str{U}$, then $\ltypesendone{\pp}{\ell}{\basetype}{\str{T}} \subtype \ltypesendone{\pp}{\ell}{\basetype}{\str{U}}$
    \item If $\str{T} \subtype \mathcal{B}^{(\pp)}.\str{U}$, then $\ltypesendone{\pp}{\ell}{\basetype}{\str{T}} \subtype \mathcal{B}^{(\pp)}.\ltypesendone{\pp}{\ell}{\basetype}{\str{U}}$
  \end{enumerate}
\end{lemma}
\begin{proof}
  The proofs of (1) and (3) are trivial.
  The proofs of (2) and (4) are easy inductions on the prefixes.
\end{proof}

\begin{lemma}\label{refinement-lemma}
  If $\str{T}$ and $\str{U}$ are both SISO, then we have $\str{T} \subtype \str{U}$ iff $\str{T} \refines \str{U}$.
\end{lemma}
\begin{proof}
  For the only if direction, assume that $\str{T} \subtype \str{U}$.
  We proceed by inspecting $\str{T}$.
  \begin{itemize}
    \item If $\str{T} = \ltypeend$, then we have $\str{U} = \ltypeend$, and $\ltypeend \refines \ltypeend$.
    \item If $\str{T} = \ltypesendone{\pp}{\ell}{\basetype}{\str{T}'}$, then $\reduceinternal{\str{U}}{\str{U}'}{\pp}{\ell}{\basetype}$ with $\str{T}' \subtype \str{U}'$, so by \cref{transition-prefix}, we can split into two cases.
      If $\str{U} = \ltypesendone{\pp}{\ell}{\basetype}{\str{U}'}$, then the result is immediate.
      Otherwise, there exist $\mathcal{B}^{(\pp)}$ and $\str{U}''$ such that $\str{U} = \mathcal{B}^{(\pp)}.\ltypesendone{\pp}{\ell}{\basetype}{\str{U}''}$ and $\str{U}' = \mathcal{B}^{(\pp)}.\str{U}''$.
      We have $\actset{\str{T'}} = \actset{\str{U'}} = \actset{\mathcal{B}^{(\pp)}.\str{U}''}$ by \cref{actset-lemma}, and $\str{T}' \refines \str{U}' = \mathcal{B}^{(\pp)}.\str{U}''$ coinductively.
    \item If $\str{T} = \ltyperecvone{\pp}{\ell}{\basetype}{\str{T}'}$, then $\Recvs{\pp}(\str{U})$, so by \cref{transition-prefix}, we can split into two cases.
      If $\str{U} = \ltyperecvone{\pp}{\ell}{\basetype}{\str{U}'}$, the result is immediate.
      Otherwise, we have $\str{U} = \mathcal{A}^{(\pp)}.\ltyperecvone{\pp}{\ell'}{\basetype'}{\str{U}''}$ and $\str{U}' = \mathcal{A}^{(\pp)}.\str{U}''$.
      Since $\reduceexternal{\str{U}}{\str{U}'}{\pp}{\ell'}{\basetype'}$, it follows from $\str{T} \subtype \str{U}$ that $\ell' = \ell$, $\basetype' = \basetype$, and $\str{U}' \subtype \str{T}'$.
      We have $\actset{\str{T}'} = \actset{\str{U}'} = \actset{\mathcal{A}^{(\pp)}.\str{U}''}$ by \cref{actset-lemma}, and $\str{T}' \refines \str{U}' = \mathcal{A}^{(\pp)}.\str{T}''$ coinductively.
  \end{itemize}

  The if direction follows from \cref{subtyping-prefix-congruences}.
\end{proof}

\subsection{The equivalence}

\begin{theorem}\label{thm:gppsy-equivalence}
    Let $\sty{T}$ and $\sty{U}$ be closed session types.
    We have $\sty{T} \subtype \sty{U}$ if and only if $\sty{T} \subtype_{\mathrm{GPPSY}} \sty{U}$.
\end{theorem}
\begin{proof}
    This is an immediate consequence of \cref{subtype-to-subtree}, \cref{subtype-siso}, and \cref{refinement-lemma}.
\end{proof}

\end{document}